\pdfoutput=1
\RequirePackage{ifpdf}
\ifpdf 
\documentclass[pdftex]{sigma}
\else
\documentclass{sigma}
\fi

\numberwithin{equation}{section}

\newtheorem{Theorem}{Theorem}[section]
\newtheorem{Lemma}[Theorem]{Lemma}
\newtheorem{Proposition}[Theorem]{Proposition}
 { \theoremstyle{definition}
\newtheorem{Definition}[Theorem]{Definition}
\newtheorem{Remark}[Theorem]{Remark}
\newtheorem*{Question}{Question}
}

\usepackage{pict2e}

\usepackage[all]{xy}
\usepackage{multirow}
\usepackage{xcolor}
\usepackage{soul}

\begin{document}


\newcommand{\arXivNumber}{1611.02560}

\renewcommand{\PaperNumber}{013}

\FirstPageHeading

\ShortArticleName{B\^ocher and Abstract Contractions of 2nd Order Quadratic Algebras}

\ArticleName{B\^ocher and Abstract Contractions\\ of 2nd Order Quadratic Algebras}

\Author{Mauricio A.~ESCOBAR RUIZ~$^{\dag^1\dag^2}$, Ernest G. KALNINS~$^{\dag^3}$, Willard MILLER Jr.~$^{\dag^2}$\\ and Eyal SUBAG~$^{\dag^4}$}

\AuthorNameForHeading{M.A.~Escobar-Ruiz, E.G.~Kalnins, W.~Miller Jr.\ and E.~Subag}

\Address{$^{\dag^1}$~Instituto de Ciencias Nucleares, UNAM, Apartado Postal 70-543, 04510 Mexico D.F.\ Mexico}
\EmailDD{\href{mailto:mauricio.escobar@nucleares.unam.mx}{mauricio.escobar@nucleares.unam.mx}}

\Address{$^{\dag^2}$~School of Mathematics, University of Minnesota, Minneapolis, Minnesota, 55455, USA}
\EmailDD{\href{mailto:miller@ima.umn.edu}{miller@ima.umn.edu}}
\URLaddressDD{\url{https://www.ima.umn.edu/~miller/}}

\Address{$^{\dag^3}$~Department of Mathematics, University of Waikato, Hamilton, New Zealand}
\EmailDD{\href{mailto:math0236@waikato.ac.nz}{math0236@waikato.ac.nz}}

\Address{$^{\dag^4}$~Department of Mathematics, Pennsylvania State University, State College,\\
\hphantom{$^{\dag^4}$}~Pennsylvania, 16802, USA}
\EmailDD{\href{mailto:eyalsubag@gmail.com}{eyalsubag@gmail.com}}

\ArticleDates{Received November 19, 2016, in f\/inal form February 27, 2017; Published online March 06, 2017}

\Abstract{Quadratic algebras are generalizations of Lie algebras which include the symmetry algebras of 2nd order superintegrable systems in 2~dimensions as special cases. The superintegrable systems are exactly solvable physical systems in classical and quantum mechanics. Distinct superintegrable systems and their quadratic algebras can be related by geometric contractions, induced by B\^ocher contractions of the conformal Lie algebra ${\mathfrak{so}}(4,\mathbb {C})$ to itself. In this paper we give a precise def\/inition of B\^ocher contractions and show how they can be classif\/ied. They subsume well known contractions of ${\mathfrak{e}}(2,\mathbb {C})$ and ${\mathfrak{so}}(3,\mathbb {C})$ and have important physical and geometric meanings, such as the derivation of the Askey scheme for obtaining all hypergeometric orthogonal polynomials as limits of Racah/Wilson polynomials. We also classify abstract nondegenerate quadratic algebras in terms of an invariant that we call a canonical form. We describe an algorithm for f\/inding the canonical form of such algebras. We calculate explicitly all canonical forms arising from quadratic algebras of 2D nondegenerate superintegrable systems on constant curvature spaces and Darboux spaces. We further discuss contraction of quadratic algebras, focusing on those coming from superintegrable systems.}


\Classification{22E70; 16G99; 37J35; 37K10; 33C45; 17B60; 81R05; 33C45}

\section{Introduction}\label{int1}
Second order 2D superintegrable systems and their associated quadratic symmetry algebras are basic in mathematical physics. Among the simplest such solvable systems are the 2D Kepler and hydrogen atom and the isotropic and Higgs oscillators~\cite{MPW2013, SCQS}. All the systems are multiseparable, with the quantum separable solutions characterized as eigenfunctions of commuting operators in the quadratic algebras. The separation equations are the Gaussian hypergeometric equation and its various conf\/luent forms in full generality, as well as the Heun equation and its conf\/luent forms in full generality~\cite{RKM2016}. Solutions of the hypergeometric and Heun equations are linked through their solution of the same superintegrable system. The conf\/luences are related to B\^ocher contractions of the conformal algebra ${\mathfrak{so}}(4,\mathbb {C})$ to itself~\cite{KMS2016}. The interbasis expansion coef\/f\/icients relating distinct separable systems lead to other special functions, several of them functions of discrete variables, such as the Racah, Wilson and Hahn polynomials in full genera\-li\-ty~\cite{KMP2014}. The contractions also allow the derivation of the Askey scheme for the classif\/ication of hypergeometric orthogonal polynomials. The classif\/ication of quasi-exactly solvable (QES) systems based on the Heun operator coincide exactly with QES separation equations for these superintegrable systems~\cite{TurbinerHeun,Turbiner2016}.

In short, the structure and classif\/ication of these quadratic algebras and their relations via contractions are matters of considerable signif\/icance in mathematical physics. Historically, the superintegrable systems have been classif\/ied and their associated quadratic algebras then computed. Here we are reversing the process: we f\/irst classify abstract quadratic algebras and then determine which of these correspond to 2nd order superintegrable systems. Also we determine how the abstract quadratic algebras are related via contractions and examine which of these contractions can be realized geometrically as B\^ocher contractions. The eventual goal is to isolate the algebras and contractions that do not correspond to geometrical superintegrable systems and to determine their signif\/icance.

\looseness=-1 B\^ocher invented a recipe for a limit procedure which showed how to f\/ind what we now know are all $R$-separable coordinate systems for free Laplace and wave equations in~$n$ dimensions~\cite{Bocher}. We have recently recognized that these limits can be interpreted as contractions of ${\mathfrak{so}}(n+2,\mathbb {C})$ to itself and classif\/ied; we call them B\^ocher contractions. In this paper we give for the f\/irst time the precise def\/inition of these contractions and their properties and classif\/ication for the case $n=2$.

We start with some basic facts. We def\/ine a quantum (Helmholtz) superintegrable system as an integrable Hamiltonian system on an $n$-dimensional pseudo-Riemannian manifold with potential: $H=\Delta_n+V$ that admits $2n-1$ algebraically independent partial dif\/ferential opera\-tors~$L_j$ commuting with $H$, the maximum possible: $ [H,L_j]=0$, $j=1,2,\dots, 2n-1$. Similarly a classical superintegrable on such a manifold, with Hamiltonian ${\mathcal H}=\sum g^{ij}p_ip_j+V$, is an integrable system that admits $2n-1$ functionally independent constants of the motion~${\mathcal L}_j$, polynomial in the momenta, in involution with $\mathcal H$, the maximum possible. Superintegrability captures the properties of quantum Hamiltonian systems that allow the Schr\"odinger eigenvalue problem (or Helmholtz equation) $H\Psi=E\Psi$ to be solved exactly, analytically and algebraically~\cite{EVAN, FORDY, MPW2013,TTW2001, SCQS} and the classical trajectories to be computed algebraically. A system is of order $K$ if the maximum order of the symmetry operators (or the polynomial order of the classical constants of the motion), other than $H$, is~$K$. For $n=2$, $K=1,2$ all systems are known, e.g.,~\cite{DASK2007, KKM20041,KKM20041-II,KKM20041-III,KKM20041-IV,KKM20041-V,KKM20041-VI}. For $K=1$ the symmetry algebras are just Lie algebras.

We review brief\/ly the facts for {\it free} 2nd order superintegrable systems (i.e., no potential, $K=2$) in the case $n=2, 2n-1=3$. The complex spaces with Laplace--Beltrami operators admitting at least three 2nd order symmetries were classif\/ied by Koenigs (1896)~\cite{Koenigs}. They are: the two constant curvature spaces (f\/lat space and the complex sphere), the four Darboux spaces (one of which, D4, contains a~parameter)~\cite{KKMW}, and 5 families of 4-parameter Koenigs spaces, see Section~\ref{Helmholtzclasses}. For 2nd order systems with non-constant potential the generating symmetry operators of each system close under commutation (or via Poisson brackets in the classical case) to determine a
quadratic algebra, and the irreducible representations of the quantum algebra determine the eigenvalues of $H$ and their multiplicities. More precisely, in the classical case, closedness means that the Poisson algebra generated by the constants of motion is f\/initely generated as an associative algebra. The quantum case is def\/ined analogously. Here we consider only the {nondegenerate superintegrable systems}: Those with 4-parameter potentials (including the additive constant) (the maximum possible):
\begin{gather}\label{nondegpot}V({\bf x})= a_1V_{(1)}({\bf x})+a_2V_{(2)}({\bf x})+a_3V_{(3)}({\bf x})+a_4,\end{gather}
where $\{V_{(1)}({\bf x}), V_{(2)}({\bf x}), V_{(3)}({\bf x}), 1\}$ is a linearly independent set. Here the possible classical and quantum
potentials are identical and there is a 1-1 relationship between classical and quantum systems. The classical constants of the motion determine the quantum symmetry operators, modulo symmetrization. The classical symmetry algebra generated by ${\mathcal H}$, ${\mathcal L}_1$, ${\mathcal L}_2$ always closes under commutation and gives the following nondegenerate quadratic algebra structure:

\begin{Definition}\label{def1}
An abstract {\it nondegenerate $($classical$)$ quadratic algebra} is a Poisson algebra with functionally independent generators ${\mathcal H}$, ${\mathcal L}_1$, ${\mathcal L}_2$, and parameters $a_1$, $a_2$, $a_3$, $a_4$, such that all generators are in involution with ${\mathcal H}$ and the following relations hold:
\begin{gather*}
\{{\mathcal L}_j,{\mathcal R}\}=\sum_{0\leq e_1+e_2+e_3\leq 2} M^{(j)}_{e_1,e_2,e_3} {\mathcal L}_1^{e_1}{\mathcal L}_2^{e_2}
{\mathcal H}^{e_3}, \qquad e_k\geq 0,\qquad {\mathcal L}_k^0=1,\\
{\mathcal R}^2={\mathcal F}\equiv\sum_{0\leq e_1+e_2+e_3\leq 3} N_{e_1,e_2,e_3} {\mathcal L}_1^{e_1}{\mathcal L}_2^{e_2} {\mathcal H}^{e_3}.
\end{gather*}
Here, ${\mathcal R}\equiv \{{\mathcal L}_1,{\mathcal L}_2\}$. In both equations the constants $M^{(j)}_{e_1, e_2, e_3}$ and $N_{e_1, e_2, e_3}$ are polynomials in the parameters $a_1$, $a_2$, $a_3$ of degree $2-e_1-e_2-e_3$ and $3-e_1-e_2-e_3$, respectively. The symmetry algebras obeyed by the quantum superintegrable systems have a similar structure, slightly more complicated due to the need for symmetrization of the noncommuting operators. In the case $a_1=a_2=a_3=a_4=0$, the corresponding quadratics algebras are called {\it free}.
\end{Definition}

Note that we can think of a nondegenerate (classical or quantum) quadratic algebra as a~family of algebras parametrized by the constants~$a_i$. The algebra is called quadratic because the Poisson brackets $\{{\mathcal L}_j,{\mathcal R}\}$ are 2nd order polynomials in the generators ${\mathcal L}_i$, ${\mathcal H}$, whereas for a Lie algebra they are 1st order. Nondegenerate 2D superintegrable systems always have a quadratic algebra structure in which the parameters~$a_j$ are those of the potential; we call these quadratic algebras {\it geometrical}.

Although the full sets of classical structure equations can be rather complicated, the function~$\mathcal F$ contains all of the structure information for nondegenerate systems. In particular, it is easy to show that, e.g.,~\cite{KM2014},
\begin{gather}\label{R2eqn}\{{\mathcal L}_1,{\mathcal R}\}=\frac12\frac{\partial {\mathcal F}}{\partial {\mathcal L}_2},
\qquad \{{\mathcal L}_2,{\mathcal R}\}=-\frac12\frac{\partial {\mathcal F}}{\partial {\mathcal L}_1},\end{gather} for any algebra satisfying Def\/inition~\ref{def1}, so $\mathcal F$ determines the structure equations explicitly.

For a nondegenerate superintegrable system with potential (\ref{nondegpot}) the structure equations are determined by ${\mathcal F}( {\mathcal H}, {\mathcal L}_1,{\mathcal L}_2,a_1,a_2,a_3,a_4)$ as def\/ined above. The ef\/fect of a St\"ackel transform~\cite{KMP2010} generated by the specif\/ic special choice of the potential function, say $V_{(3)}$ is to determine a new superintegrable system with Casimir ${\tilde{\mathcal R}}^2={\mathcal F}( -a_3, {\mathcal L}_1,{\mathcal L}_2,a_1,a_2,-{\mathcal H},a_4)$. Of course, the switch of~$a_3$ and~$\mathcal H$ is only for illustration; there is a St\"ackel transform that replaces any~$a_j$ by~$-\mathcal H$ and~$\mathcal H$ by~$-a_j$ and similar transforms that apply to any basis that we choose for the potential space.

If we consider the free systems (zero potential which is the case with all parameters equal zero) on the spaces classif\/ied by Koenigs, then the vector space of 2nd order symmetries may be larger than 3: 6-dimensional for constant curvature spaces, 4-dimensional for Darboux spaces, and 3-dimensional for Koenigs spaces. In general the Poisson algebras generated by taking Poisson brackets of these 2nd order elements are inf\/inite-dimensional; they do not close (in the sense that was explained above). However, in \cite{KM2014}, the possible 3-dimensional subspaces of 2nd order free symmetries that generate quadratic algebras were classif\/ied, up to conjugacy by symmetry groups of these spaces: ${\mathfrak{e}}(2,\mathbb {C})$ for f\/lat space, ${\mathfrak{o}}(3,\mathbb {C})$ for nonzero constant curvature spaces, and a 1-dimensional translation subalgebra for Darboux spaces. For Koenigs spaces the f\/irst order symmetry algebra is 0-dimensional and the space of 2nd order symmetries is 3-dimensional which always generates a unique quadratic algebra.

\begin{Theorem} For each of the spaces classified by Koenigs, there is a bijection between free quadratic algebras of $2$nd order symmetries, classified up to conjugacy, and $2$nd order nondege\-ne\-rate superintegrable systems on these spaces.
 \end{Theorem}

The proof of this theorem is constructive \cite{KM2014}. Given a free quadratic algebra $\tilde Q$ one can compute the potential~$V$ and
the symmetries of the quadratic algebra~$Q$ of the nondegenerate superintegrable system. (The quadratic algebra structure guarantees that the
Bertrand--Darboux equations for the potential are satisf\/ied identically. In this sense the free systems ``know'' the possible nondegenerate superintegrable systems they can support. Since there is a 1-1 relationship between quantum and classical nondegenerate systems, the information about all of these systems is encoded in the free quadratic algebras generated by 2nd order constants of the motion (Killing tensors) of constant curvature, Darboux and Koenigs spaces. Note that for f\/lat space the generators for the free quadratic algebras can be expressed as 2nd order elements in the universal enveloping algebra of ${\mathfrak{e}}(2,\mathbb {C})$, and for nonzero constant curvature spaces the generators for the free quadratic algebras can be expressed as 2nd order elements in the universal enveloping algebra of ${\mathfrak{so}}(3,\mathbb {C})$~\cite{KM2014}.

All 2nd order 2D superintegrable systems with potential and their quadratic algebras are known. There are 33 nondegenerate systems, on a variety of manifolds classif\/ied up to conjugacy, see Section~\ref{Helmholtzclasses} where the numbering for constant curvature systems is taken from~\cite{KKMP}, (the numbers are not always consecutive because the lists in~\cite{KKMP} also include degenerate systems) and the numbering for Darboux spaces is taken from~\cite{KKMW}. For each system we give the 4-parameter potential and the abstract free structure equation ${\mathcal R}^2-{\mathcal F}=0$. Note that many of the abstract structure equations for the superintegrable systems are identical, even for superintegrable systems on dif\/ferent manifolds. Of course the geometrical structure equations are distinct because the generators ${\mathcal L}_1$, ${\mathcal L}_2$, ${\mathcal H}$ are distinct for each geometrical superintegrable system.

Under the St\"ackel transform (we discuss this in Section~\ref{2.1}) these systems divide into 6 equivalence classes with representatives on f\/lat space and the 2-sphere, see \cite{Kress2007} and Section~\ref{3.2}.

\subsection{The Helmholtz nondegenerate superintegrable systems} \label{Helmholtzclasses}

{\bf Flat space systems}: ${\mathcal H}\equiv p_x^2+p_y^2+V=E$.
\begin{enumerate}\itemsep=0pt

\item $E1$: $V=\alpha\big(x^2+y^2\big)+\frac{\beta}{x^2}+\frac{\gamma}{y^2}$, ${\mathcal R}^2={\mathcal L}_1{\mathcal L}_2({\mathcal H}+{\mathcal L}_2)$,
\item $E2$: $V=\alpha\big(4x^2+y^2\big)+\beta x+\frac{\gamma}{ y^2}$, ${\mathcal R}^2={\mathcal L}_1^2({\mathcal H}+{\mathcal L}_1)$,

\item $E3'$: $V=\alpha\big(x^2+y^2\big)+\beta x+\gamma y$,\quad ${\mathcal R}^2=0$,

\item $ E7$: $V=\frac{\alpha(x+iy)}{\sqrt{(x+iy)^2-b}}+\frac{\beta (x-iy)}{\sqrt{(x+iy)^2-b} \big(x+iy+\sqrt{(x+iy)^2-b}\big)^2} +\gamma \big(x^2+y^2\big)$, ${\mathcal R}^2={\mathcal L}_1{\mathcal L}_2^2+b{\mathcal L}_2{\mathcal H}^2$,

\item $E8$: $V=\frac{\alpha (x-iy) }{(x+iy)^3}+\frac{\beta}{(x+iy)^2}+\gamma\big(x^2+y^2\big)$, ${\mathcal R}^2={\mathcal L}_1{\mathcal L}_2^2$,

\item $E9$: $V=\frac{\alpha}{\sqrt{x+iy}}+\beta y+\frac{\gamma (x+2iy)}{\sqrt{x+iy}}$, ${\mathcal R}^2={\mathcal L}_1({\mathcal L}_1+{\mathcal H})^2$,

\item $E10$: $V=\alpha(x-iy)+\beta \big(x+iy-\frac32(x-iy)^2\big)+\gamma\big(x^2+y^2-\frac12(x-iy)^3\big)$, ${\mathcal R}^2={\mathcal L}_1^3$,

\item $ E11$: $V=\alpha(x-iy)+\frac{\beta (x-iy)}{\sqrt{x+iy}}+\frac{\gamma }{\sqrt{x+iy}}$, ${\mathcal R}^2={\mathcal L}_1{\mathcal H}^2$,

\item ${ E15}$: $V=f(x-iy)$, where $f$ is arbitrary, ${\mathcal R}^2={\mathcal L}_1^3$ (the exceptional case, characterized by the fact that the symmetry generators are functionally linearly dependent \cite{KKM20041,KKM20041-II,KKM20041-III,KKM20041-IV,KKM20041-V,KKM20041-VI, KKMP}),

 \item $E16$: $V=\frac{1}{\sqrt{x^2+y^2}}\Big(\alpha+\frac{\beta}{y+\sqrt{x^2+y^2}}+\frac{\gamma}{y-\sqrt{x^2+y^2}}\Big)$, ${\mathcal R}^2={\mathcal L}_1\big({\mathcal L}_1{\mathcal H}+{\mathcal L}_2^2\big)$,

\item $E17$: $V=\frac{\alpha}{\sqrt{x^2+y^2}}+\frac{\beta}{(x+iy)^2}+\frac{\gamma}{(x+iy)\sqrt{x^2+y^2}}$, ${\mathcal R}^2={\mathcal L}_1{\mathcal L}_2^2$,

 \item $ E19$: $V=\frac{\alpha(x+iy)}{\sqrt{(x+iy)^2-4}}+\frac{\beta}{\sqrt{(x-iy)(x+iy+2)}}+\frac{\gamma}{\sqrt{(x-iy)(x+iy-2)}}$, ${\mathcal R}^2={\mathcal L}_1\big({\mathcal L}_2^2+{\mathcal H}^2\big)$,

\item $E20$: $V=\frac{1}{\sqrt{x^2+y^2}}\Big(\alpha+\beta \sqrt{x+\sqrt{x^2+y^2}}+\gamma \sqrt{x-\sqrt{x^2+y^2}}\Big)$, ${\mathcal R}^2={\mathcal H}\big({\mathcal L}_1^2+{\mathcal L}_2^2\big)$.
\end{enumerate}

{\bf Systems on the complex 2-sphere}: ${\mathcal H}\equiv {\mathcal J}_{23}^2+{\mathcal J}_{13}^2+{\mathcal J}_{12}^2+V=E$. Here, ${\mathcal J}_{k\ell}=s_k p_{s_\ell}-s_\ell p_{s_k}$ and $ s_1^2+s_2^2+s_3^2=1$.

\begin{enumerate}\itemsep=0pt

\item $S1$: $V=\frac{\alpha }{(s_1+is_2)^2}+\frac{\beta s_3}{(s_1+is_2)^2}+\frac{\gamma(1-4s_3^2)}{(s_1+is_2)^4}$, ${\mathcal R}^2={\mathcal L}_1^3$,

\item $S2$: $V=\frac{\alpha }{s_3^2}+\frac{\beta }{(s_1+is_2)^2}+\frac{\gamma(s_1-is_2)}{(s_1+is_2)^3}$, ${\mathcal R}^2={\mathcal L}_1{\mathcal L}_2^2$,

\item $S4$: $V=\frac{\alpha}{(s_1+is_2)^2}+\frac{\beta s_3}{\sqrt{s_1^2+s_2^2}}+\frac{\gamma}{(s_1+is_2)\sqrt{s_1^2+s_2^2}}$, ${\mathcal R}^2={\mathcal L}_1{\mathcal L}_2^2$,

\item $ S7$: $V=\frac{\alpha s_3}{\sqrt{s_1^2+s_2^2}}+\frac{\beta s_1}{s_2^2\sqrt{s_1^2+s_2^2}}+\frac{\gamma}{s_2^2}$, ${\mathcal R}^2={\mathcal L}_1^2{\mathcal L}_2+{\mathcal L}_2^2{\mathcal L}_1 -\frac{1}{16}{\mathcal L}_1^2{\mathcal H}$,

\item $S8$: $V=\frac{\alpha s_2}{\sqrt{s_1^2+s_3^2}}+\frac{\beta (s_2+is_1+s_3)}{\sqrt{(s_2+is_1)(s_3+is_1)}} +\frac{\gamma(s_2+is_1-s_3)}{\sqrt{(s_2+is_1)(s_3-is_1)}}$, ${\mathcal R}^2={\mathcal L}_1^2{\mathcal L}_2+{\mathcal L}_1{\mathcal L}_2^2-\frac14{\mathcal L}_1{\mathcal L}_2{\mathcal H}$,

\item $S9$: $V=\frac{\alpha}{s_1^2}+\frac{\beta}{s_2^2}+\frac{\gamma}{s_3^2}$, ${\mathcal R}^2={\mathcal L}_1^2{\mathcal L}_2+{\mathcal L}_1{\mathcal L}_2^2+\frac{1}{16}{\mathcal L}_1{\mathcal L}_2{\mathcal H}$.
\end{enumerate}

{\bf Darboux 1 systems}: ${\mathcal H}\equiv\frac{1}{4x}\big(p_x^2+p_y^2\big)+V=E$.
\begin{enumerate}\itemsep=0pt
\item ${D1A}$: $ V=\frac{b_1(2x-2b+iy)}{x\sqrt{x-b+iy}}+\frac{b_2}{x\sqrt{x-b+iy}}+\frac{b_3}{x}+b_4$, ${\mathcal R}^2={\mathcal L}_1^3+{\mathcal L}_2{\mathcal L}_1{\mathcal H} -b{\mathcal L}_1^2{\mathcal H}-2ib{\mathcal H}^2{\mathcal L}_2$,

\item ${ D1B}$: $V=\frac{b_1(4x^2+y^2)}{x}+\frac{b_2}{x}+\frac{b_3}{xy^2}+b_4$, ${\mathcal R}^2={\mathcal L}_1^3+{\mathcal L}_2{\mathcal L}_1{\mathcal H}$,

 \item ${ D1C}$ $V=\frac{b_1(x^2+y^2)}{x}+\frac{b_2}{x}+\frac{b_3y}{x}+b_4$, ${\mathcal R}^2={\mathcal L}_2{\mathcal H}^2$.
\end{enumerate}

{\bf Darboux 2 systems}: ${\mathcal H}\equiv\frac{x^2}{x^2+1}\big(p_x^2+p_y^2\big)+V=E$.
\begin{enumerate}\itemsep=0pt
\item ${D2A}$: $V=\frac{x^2}{x^2+1}\big(b_1\big(x^2+4y^2\big)+\frac{b_2}{x^2}+b_3y\big)+b_4$, ${\mathcal R}^2={\mathcal L}_1^3+{\mathcal L}_1^2{\mathcal H} +\frac14{\mathcal L}_1{\mathcal H}^2$,
\item ${D2B}$: $V=\frac{x^2}{x^2+1}\big(b_1\big(x^2+y^2\big)+\frac{b_2}{x^2}+\frac{b_3}{y^2}\big)+b_4$, ${\mathcal R}^2={\mathcal L}_1{\mathcal L}_2^2 +{\mathcal L}_1{\mathcal L}_2{\mathcal H}-\frac{1}{16}{\mathcal L}_2{\mathcal H}^2$,
 \item ${D2C}$: $V=\frac{x^2}{\sqrt{x^2+y^2}(x^2+1)}\Big(b_1+\frac{b_2}{y+\sqrt{x^2+y^2}}+\frac{b_3}{y-\sqrt{x^2+y^2}}\Big)+b_4$, ${\mathcal R}^2={\mathcal L}_1{\mathcal L}_2^2+{\mathcal L}_1^2{\mathcal H}-\frac14{\mathcal L}_1{\mathcal H}^2$.
\end{enumerate}

{\bf Darboux 3 systems}: ${\mathcal H}\equiv\frac12\frac{e^{2x}}{e^x+1}\big(p_x^2+p_y^2\big)+V=E$.
\begin{enumerate}\itemsep=0pt
 \item ${ D3A}$: $ V=\frac{b_1}{1+e^x}+\frac{b_2e^x}{\sqrt{1+2e^{x+iy}}(1+e^x)}+\frac{b_3e^{x+iy}}{\sqrt{1+2e^{x+iy}}(1+e^x)}+b_4$,
 ${\mathcal R}^2={\mathcal H}\big({\mathcal L}_1^2+{\mathcal L}_2^2-{\mathcal H}^2\big)$,
\item ${ D3B}$: $V=\frac{e^x}{e^x+1}\big(b_1+e^{-\frac{x}{2}}\big(b_2\cos\frac{y}{2}+b_3\sin\frac{y}{2}\big)\big)+b_4$, ${\mathcal R}^2={\mathcal L}_1{\mathcal L}_2^2+{\mathcal H}{\mathcal L}_1^2-\frac14{\mathcal H}^2{\mathcal L}_1$,
 \item ${ D3C}$: $V= \frac{e^x}{e^x+1}\Big(b_1+e^x(\frac{b_2}{\cos^2\frac{y}{2}}+\frac{b_3}{\sin^2\frac{y}{2}})\Big)+b_4$, ${\mathcal R}^2={\mathcal L}_1{\mathcal L}_2^2+{\mathcal L}_1^2{\mathcal H}-\frac18{\mathcal L}_1{\mathcal H}^2$,
 \item ${ D3D}$: $V=\frac{e^{2x}}{1+e^x}\big(b_1 e^{-iy}+b_2 e^{-2iy}\big)+\frac{b_3}{1+e^x}+b_4$, ${\mathcal R}^2={\mathcal L}_1{\mathcal L}_2^2
 +{\mathcal L}_1{\mathcal L}_2{\mathcal H}+{\mathcal L}_2{\mathcal H}^2-{\mathcal H}^3$.
\end{enumerate}

{\bf Darboux 4 systems}: ${\mathcal H}\equiv -\frac{\sin^2 2x}{2\cos 2x+b}\big(p_x^2+p_y^2\big)+V=E$.
\begin{enumerate}\itemsep=0pt

\item ${ D4(b)A}$: $ V=\frac{\sin^2 2 x}{2 \cos 2 x+b}\Big(\frac{b_1}{\sinh^2 y}+\frac{b_2}{\sinh^2 2 y}\Big)+\frac{b_3}{2 \cos 2 x+b}+b_4$, ${\mathcal R}^2={\mathcal L}_1{\mathcal L}_2^2$,
\item $ D4(b)B$: $V=\frac{\sin^2 2x}{2 \cos 2x +b}\Big(\frac{b_1}{\sin^2 2x}+b_2e^{4y}+b_3e^{2y}\Big)+b_4$, ${\mathcal R}^2={\mathcal L}_1{\mathcal L}_2^2+{\mathcal L}_1^2{\mathcal L}_2+{b}{\mathcal H}{\mathcal L}_2^2-4{\mathcal H}^2{\mathcal L}_2$,
\item ${ D4(b)C}$: $V=\frac{e^{2y}}{\frac{b+2}{\sin^2 x}+\frac{b-2}{\cos^2 x}}\Big(\frac{b_1}{Z+(1-e^{2y}) \sqrt{Z}}+\frac{b_2}{Z+(1+e^{2y})\sqrt{Z}} +\frac{b_3\ e^{-2y}}{\cos^2 x}\Big) +b_4$,\\ ${\mathcal R}^2= -\frac{b}{16^3}{\mathcal H}^3+{\mathcal L}_1^2{\mathcal L}_2+{\mathcal L}_1{\mathcal L}_2^2-\frac{b}{16}{\mathcal L}_1{\mathcal L}_2{\mathcal H}-\frac{b}{16}{\mathcal L}_2^2{\mathcal H} +\frac{1}{256}{\mathcal L}_1{\mathcal H}^2$.
\end{enumerate}
\emph{Note:} Systems $D4(b)A$, $D4(b)B$, $D4(b)C$ are in fact families of distinct systems parametrized by~$b$, and $E15$ is a family of systems parametrized by the function~$f$. The parameters~$b$ can be normalized away in systems $E7$, $D1A$, but it is convenient to keep them.

{\bf Generic Koenigs spaces}: (We do not list the relatively unenlightening expressions of ${\mathcal R}^2$ for the Koenigs spaces.
Each involves 4 arbitrary parameters obtained via a generic St\"ackel transformation from a constant curvature system.)
\begin{enumerate}\itemsep=0pt
\item ${K[1,1,1,1]}$:
${\mathcal H}\equiv \frac{1}{V(b_1,b_2,b_3,b_4)}\big(p_x^2+p_y^2 +V(a_1,a_2,a_3,a_4)\big)=E$,\\
$V(a_1,a_2,a_3,a_4)=\frac{a_1}{x^2}+\frac{a_2}{y^2}+\frac{4a_3}{(x^2+y^2-1)^2}-\frac{4a_4}{(x^2+y^2+1)^2}$,
 \item ${K[2,1,1]}$: ${\mathcal H}\equiv \frac{1}{V(b_1,b_2,b_3,b_4)}\big(p_x^2+p_y^2 +V(a_1,a_2,a_3,a_4)\big)=E$,\\
 $V(a_1,a_2,a_3,a_4)=\frac{a_1}{x^2}+\frac{a_2}{y^2}-a_3\big(x^2+y^2\big)+a_4$,
 \item ${K[2,2]}$: ${\mathcal H}\equiv\frac{1}{V(b_1,b_2,b_3,b_4)}\big(p_x^2+p_y^2 +V(a_1,a_2,a_3,a_4)\big)=E$,\\
$V(a_1,a_2,a_3,a_4)= \frac{a_1}{(x+iy)^2}+\frac{a_2(x-iy)}{(x+iy)^3} +a_3-a_4\big(x^2+y^2\big)$,
\item ${K[3,1]}$: ${\mathcal H}\equiv\frac{1}{V(b_1,b_2,b_3,b_4)}\big(p_x^2+p_y^2 +V(a_1,a_2,a_3,a_4)\big)=E$,\\
$V(a_1,a_2,a_3,a_4)= a_1-a_2x +a_3\big(4x^2+{y}^2\big)+\frac{a_4}{{y}^2}$,
\item ${K[4]}$: ${\mathcal H}\equiv =\frac{1}{V(b_1,b_2,b_3,b_4)}\big(p_x^2+p_y^2 +V(a_1,a_2,a_3,a_4)\big)=E$,\\
$V(a_1,a_2,a_3,a_4)= a_1-a_2(x+iy) +a_3\big(3(x+iy)^2+2(x-iy)\big) -a_4\big(4\big(x^2+y^2\big)+2(x+iy)^3\big)$,
\item ${K[0]}$: ${\mathcal H}\equiv =\frac{1}{V(b_1,b_2,b_3,b_4)}\big(p_x^2+p_y^2 +V(a_1,a_2,a_3,a_4)\big)=E$,\\
$V(a_1,a_2,a_3,a_4)=a_1-(a_2x+a_3y)+a_4\big(x^2+y^2\big)$.
\end{enumerate}

\subsection{Contractions}\label{1.2}

In \cite{KM2014} it has been shown that all the 2nd order superintegrable systems are obtained by taking coordinate limits of the generic system~$S_9$~\cite{KKMP}, or are obtained from these limits by a St\"ackel transform (an invertible structure preserving mapping of superintegrable systems~\cite{KKM20041,KKM20041-II,KKM20041-III,KKM20041-IV,KKM20041-V, KKM20041-VI}). Analogously all quadratic symmetry algebras of these systems are limits of that of~$ S_9$. These coordinate limits induce limit relations between the special functions associated as eigenfunctions of the quantum superintegrable systems. The limits also induce contractions of the associated quadratic algebras, and via the models of the irreducible representations of these algebras, limit relations between the associated special functions. The Askey scheme for ortho\-gonal functions of hypergeometric type is an example of this~\cite{KMP2014}. For constant curvature systems the required limits are all induced by In\"on\"u--Wigner-type Lie algebra contractions of~${\mathfrak{o}}(3,\mathbb {C})$ and~${\mathfrak{e}}(2,\mathbb {C})$~\cite{Wigner,NP,WW}. In\"on\"u--Wigner-type Lie algebra contractions have long been applied to relate separable coordinate systems and their associated special functions, see, e.g.,~\cite{Pog96,Pog01} for some more recent examples, but the application to quadratic algebras is due to the authors and their collaborators.

Recall the def\/inition of (natural) {\it Lie algebra contractions}: Let $(A; [\, ; \,]_A)$, $(B; [\, ;\, ]_B)$ be two complex Lie algebras. We say that $B$ is a {\it contraction} of $A$ if for every $\epsilon\in (0,1]$ there exists a~linear invertible map $t_\epsilon \colon B\to A$ such that for every $X, Y\in B$, $ \lim\limits_{\epsilon\to 0}t_\epsilon^{-1}[t_\epsilon X,t_\epsilon Y]_A = [X, Y ]_B$. Thus, as $\epsilon\to 0$ the 1-parameter family of basis transformations can become singular but the structure constants of the Lie algebra go to a f\/inite limit, necessarily that of another Lie algebra. The contractions of the symmetry algebras of 2D constant curvature spaces have long since been classif\/ied~\cite{KM2014}. There are 6~nontrivial contractions of ${\mathfrak{e}}(2,\mathbb {C})$ and 4 of ${\mathfrak{o}}(3,\mathbb {C})$. They are each induced by coordinate limits. Just as for Lie algebras we can def\/ine a contraction of a quadratic algebra in terms of 1-parameter families of basis changes in the algebra. As $\epsilon\to 0$ the 1-parameter family of basis transformations becomes singular but the structure constants go to a~f\/inite limit~\cite{KM2014}.

\begin{Theorem} Every Lie algebra contraction of $A={\mathfrak{e}}(2,\mathbb {C})$ or $A={\mathfrak{o}}(3,\mathbb {C})$ induces a contraction of a free $($zero potential$)$ quadratic algebra $\tilde Q$ based on $A$, which in turn induces a contraction of the quadratic algebra $Q$ with potential. This is true for both classical and quantum algebras.
\end{Theorem}

Similarly the coordinate limit associated with each contraction takes $H$ to a new superintegrable system with the contracted quadratic algebra. This relationship between coordinate limits, Lie algebra contractions and quadratic algebra contractions for superintegrable systems on constant curvature spaces breaks down for Darboux and Koenigs spaces. For Darboux spaces the Lie symmetry algebra is only 1-dimensional, and there is no Lie symmetry algebra at all for Koenigs spaces. Furthermore, there is the issue of f\/inding a more systematic way of classifying the 44 distinct Helmholtz superintegrable systems on dif\/ferent manifolds, and their relations. These issues can be clarif\/ied by considering the Helmholtz systems as Laplace equations (with potential) on f\/lat space. As announced in~\cite{KMS2016}, the proper object to study is the conformal symmetry algebra ${\mathfrak{so}}(4,\mathbb {C})$ of the f\/lat space Laplacian and its contractions. The basic idea is that families of (St\"ackel-equivalent) Helmholtz superintegrable systems on a variety of manifolds correspond to a single conformally superintegrable Laplace equation on f\/lat space. We exploit this here in the case $n=2$, but it generalizes easily to all dimensions $n\ge 2$. The conformal symmetry algebra for Laplace equations with
constant potential on f\/lat space is the conformal algebra ${\mathfrak{so}}(n+2,\mathbb {C})$.

In his 1894 thesis \cite{Bocher} B\^ocher introduced a limit procedure based on the roots of quadratic forms to f\/ind families of $R$-separable solutions of the ordinary (zero potential) f\/lat space Laplace equation in $n$ dimensions. An important feature of his work was the introduction of special projective coordinates in which the action of the conformal group ${\mathfrak{so}}(n+2,\mathbb {C})$ on solutions of the Laplace equation can be linearized. For $n=2$ these are tetraspherical coordinates. In Sections~\ref{Tc} and~\ref{Bc} we describe in detail the Laplace equation mechanism and how it can be applied to systematize the classif\/ication of Helmholtz superintegrable systems and their relations via limits. We show that B\^ocher's limit procedure can be interpreted as constructing generalized In\"on\"u--Wigner Lie algebra contractions of ${\mathfrak{so}}(4,\mathbb {C})$ to itself. We call these B\^ocher contractions and show that they induce contractions of the conformal quadratic algebras associated with Laplace superintegrable systems. All of the limits of the Helmholtz systems classif\/ied before for $n=2$ \cite{HKMS2015, KM2014} are induced by the larger class of B\^ocher contractions~\cite{KMS2016}. In this paper we replace B\^ocher's prescription by a~precise def\/inition of B\^ocher contractions and introduce special B\^ocher contractions, which are simpler and more easily classif\/ied.

\section{2D conformal superintegrability of the 2nd order} \label{2}

Classical nD systems of Laplace type are of the form
 \begin{gather*}
 {\mathcal H}\equiv \sum_{i,j=1}^ng^{ij}p_ip_j+V=0.\end{gather*}
A conformal symmetry of this equation is a function $ {\mathcal S}({\bf x},{\bf p})$ in the variables ${\bf x}=(x_1,\dots,x_n)$, polynomial in the momenta ${\bf p}=(p_1,\dots,p_n)$, such that $\{ {\mathcal S}, {\mathcal H}\}={\mathcal R}_{ {\mathcal S}} {\mathcal H}$ for some function ${\mathcal R}_{{\mathcal S}}({\bf x},{\bf p})$, polynomial in the momenta. Two conformal symmetries ${ {\mathcal S}}, {{\mathcal S}}'$ are identif\/ied if ${\mathcal S}={\mathcal S}'+{\mathcal R}{\mathcal H}$ for some function ${\mathcal R}({\bf x},{\bf p})$, polynomial in the momenta. (For short we will say that ${\mathcal S}={\mathcal S}'$, $\operatorname{mod} \, {\mathcal H}$ and that $\mathcal S$ is a conformal constant of the motion (or conformal symmetry) if $\{{\mathcal S},{\mathcal H}\}=0,\mod({\mathcal H})$.) The system is {\it conformally superintegrable} for $n>2$ if there are $2n-1$ functionally independent conformal symmetries, ${ \mathcal S}_1,\dots,{\mathcal S}_{2n-1}$ with ${ \mathcal S}_1={\mathcal H}$. It is {\it second order conformally superintegrable} if each symmetry~${\mathcal S}_i$ can be chosen to be a polynomial of at most second order in the momenta. There are obvious operator counterparts to these def\/initions for the operator Laplace equation ${ H\Psi}\equiv(\Delta_n +V)\psi=0$.

For $n=2$ the def\/inition must be restricted, since for a potential $V=0$ there will be an inf\/inite-dimensional space of conformal symmetries. We assume $V\ne 0$, possibly a constant.

Every $2D$ Riemannian manifold is conformally f\/lat, so we can always f\/ind a Cartesian-like coordinate system with coordinates ${\bf x}\equiv
(x,y)\equiv (x_1,x_2)$ such that the Laplace equation takes the form
\begin{gather}\label{Laplace4} {\tilde {\mathcal H}}=\frac{1}{\lambda(x,y)}\big(p_x^2+p_y^2\big)+{\tilde V}({\bf x})=0.\end{gather}
However, this equation is equivalent to the f\/lat space equation
\begin{gather}\label{Laplace5}{\mathcal H}\equiv p_x^2+p_y^2+ V({\bf x})=0,\qquad V({\bf x})=\lambda({\bf x}){\tilde V}({\bf x}).\end{gather}
In particular, the conformal symmetries of (\ref{Laplace4}) are identical with the conformal symmetries of~(\ref{Laplace5}). Thus without loss of generality we can assume the manifold is f\/lat space with $\lambda\equiv 1$.

In general the space of 2nd order conformal symmetries could be inf\/inite-dimensional. However, the requirement that~$H$ have a multiparameter potential reduces the possible symmetries to a f\/inite-dimensional space. The result, from the Bertrand--Darboux conditions, is that the pure 2nd order polynomial terms in conformal symmetries belong to the space spanned by symmetrized products of the conformal Killing vectors
\begin{gather} P_1=p_x,\qquad P_2=p_y,\qquad J=x p_y-y p_x,\qquad D=x p_x+y p_y,\nonumber\\
K_1=\big(x^2-y^2\big)p_x +2xyp_y,\qquad K_2=\big(y^2-x^2\big)p_y+2xyp_x.\label{conformalKV}
\end{gather}
For a given multiparameter potential only a subspace of these conformal tensors occurs.

\subsection{The conformal St\"ackel transform} \label{2.1}

We review brief\/ly the concept of the conformal St\"ackel transform \cite{KMP2010}. Suppose we have a second order {\it conformal} superintegrable system
\begin{gather}\label{confl} {\mathcal H}\equiv\frac{1}{\lambda(x,y)}\big(p^2_x+p_y^2\big)+V(x,y)=0,\qquad {\mathcal H}\equiv {\mathcal H}_0+V
\end{gather}
with $V$ the general potential solution for this system, and suppose $U(x,y) $ is a particular potential solution, nonzero in an open set. The {\it conformal St\"ackel transform} induced by~$U$ is the system
\begin{gather}\label{helms} {\tilde { \mathcal H}}=E,\qquad {\tilde { \mathcal H}}\equiv\frac{1}{{\tilde \lambda}}\big(p_{x}^2+p_{y}^2\big)+{\tilde V},
\end{gather}
where ${\tilde \lambda}=\lambda U$, ${\tilde V}=\frac{V}{U}$. In \cite{Laplace2011,KMS2016} we proved

\begin{Theorem}\label{stackelt}
The transformed $($Helmholtz$)$ system $\tilde {\mathcal H}$ is superintegrable $($in the nonconformal sense$)$.
\end{Theorem}

This result shows that any second order conformal Laplace superintegrable system admitting a nonconstant potential $U$ can be St\"ackel transformed to a Helmholtz superintegrable system. This operation is invertible, but the inverse is not a St\"ackel transform. By choosing all possible special potentials $U$ associated with the f\/ixed Laplace system (\ref{confl}) we generate the equivalence class of all Helmholtz
superintegrable systems (\ref{helms}) obtainable through this process. As is easy to check, any two Helmholtz superintegrable systems lie in the
same equivalence class if and only if they are St\"ackel equivalent in the standard sense, see~\cite[Theorem 4]{KMS2016}. All Helmholtz superintegrable systems are related to conformal Laplace systems in this way, so the study of all Helmholtz superintegrability on conformally f\/lat manifolds can be reduced to the study of all conformal Laplace superintegrable systems on f\/lat space. All of these results have direct analogs for operator Laplace systems.

The basic structure of quadratic algebras for nondegenerate Helmholtz superintegrable systems is preserved under the transformation to Laplace equations, except that all identities hold $\operatorname{mod}\, {\mathcal H}$:

\begin{Theorem}[\cite{KMS2016}] The symmetries ${\mathcal S}_1$, ${\mathcal S}_2$ of the $2D$ nondegenerate conformal superintegrable Hamiltonian $\mathcal H$ generate a~quadratic algebra
 \begin{gather*}
 \{{\mathcal R},{\mathcal S}_1\}=f^{(1)}({\mathcal S}_1,\mathcal{S}_2,\alpha_1,\alpha_2,\alpha_3,\alpha_4),\qquad \{{\mathcal R},
 {\mathcal S}_2\}=f^{(2)} ({\mathcal S}_1,{\mathcal S}_2,\alpha_1,\alpha_2,\alpha_3,\alpha_4),\\
{\mathcal R}^2=f^{(3)}({\mathcal S}_1,\mathcal{S}_2,\alpha_1,\alpha_2,\alpha_3,\alpha_4),
\end{gather*}
where $\mathcal{R}=\{{\mathcal S}_1,\mathcal{S}_2\}$ and all identities hold $\operatorname{mod}\, {\mathcal H}$. Here the~$\alpha_j$ are the parameters in the nondegenerate potential.
\end{Theorem}

A crucial observation now is that the free parts (those parts that one obtains by setting all the $a_i$ to zero) of the generators for 2nd order
conformal superintegrable systems lie in the universal enveloping algebra of the conformal Lie algebra, $\operatorname{mod}\, {\mathcal H}$. Thus for the 2D case it follows that contractions of ${\mathfrak{so}}(4,\mathbb {C})$ induce contractions of the conformal quadratic algebras of 2nd order superintegrable systems with nondegenerate potentials, and contractions of one system into another. In~\cite{KMS2016} it is shown how these Laplace contractions then induce contractions of Helmholtz superintegrable systems.

\section{Tetraspherical coordinates and Laplace systems} \label{Tc}
As already mentioned, the free parts of the 2nd order conformal symmetries of the Laplace equation ${\mathcal H}\equiv p_x^2+p_y^2+ V({\bf x})=0$
lie in the universal enveloping algebra of ${\mathfrak{so}}(4,\mathbb {C})$ with generators~(\ref{conformalKV}). To linearize the action of these ${\mathfrak{so}}(n+2,\mathbb {C})$ operators on Laplace equations in $n$ dimensions, B\^ocher introduced a family of projective coordinates on the null cone in $n+2$ dimensions. In our case $n=2$ these are the tetraspherical coordinates $(x_1,\dots,x_4)$. They satisfy $x_1^2+x_2^2+x_3^2+x_4^2=0$ (the null cone) and $\sum\limits_{k=1}^4x_k\partial_{x_k}=0$. They are projective coordinates on the null cone and have 2 degrees of freedom. Their principal advantage over f\/lat space Cartesian coordinates is that the action of the conformal algebra~(\ref{conformalKV}) and of the conformal group ${\rm SO}(4,\mathbb {C})$ is linearized in tetraspherical coordinates.

\subsection[Relation to Cartesian coordinates $(x,y)$ and coordinates on the 2-sphere $(s_1,s_2,s_3)$]{Relation to Cartesian coordinates $\boldsymbol{(x,y)}$ and coordinates\\ on the 2-sphere $\boldsymbol{(s_1,s_2,s_3)}$}
\begin{gather*} x_1=2XT,\qquad x_2=2YT,\qquad x_3=X^2+Y^2-T^2,\qquad x_4=i\big(X^2+Y^2+T^2\big),\\
 x=\frac{X}{T}=-\frac{x_1}{x_3+ix_4},\qquad y=\frac{Y}{T}=-\frac{x_2}{x_3+ix_4},\qquad
 x=\frac{s_1}{1+s_3},\qquad y=\frac{s_2}{1+s_3}.\end{gather*}
 The projective variables $X$, $Y$, $T$ are def\/ined by these relations
\begin{gather*} s_1=\frac{2x}{x^2+y^2+1},\qquad s_2=\frac{2y}{x^2+y^2+1},\qquad s_3=\frac{1-x^2-y^2}{x^2+y^2+1},\\
 {\mathcal H}\equiv p_x^2+p_{y}^2+{\tilde V}=(x_3+ix_4)^2\left(\sum_{k=1}^4p_{x_k}^2+V\right)
=(1+s_3)^2\left(\sum_{j=1}^3p_{s_j}^2+V\right),\\
 {\tilde V}=(x_3+ix_4)^2V, \!\!\qquad (1+s_3)=-i\frac{(x_3+ix_4)}{x_4},\!\!\qquad s_1=\frac{ix_1}{x_4},\!\!\qquad s_2=\frac{ix_2}{x_4},\!\!\qquad s_3=-\frac{ix_3}{x_4}.\end{gather*}
Thus the Laplace equation $ {\mathcal H}\equiv p_x^2+p_{y}^2+{\tilde V}$ in Cartesian coordinates becomes $\sum\limits_{k=1}^4p_{x_k}^2+V=0$ in tetraspherical coordinates.

\subsection[Relation to f\/lat space and 2-sphere 1st order conformal constants of the motion]{Relation to f\/lat space and 2-sphere 1st order conformal constants\\ of the motion}
We def\/ine
\begin{gather*} L_{jk}=x_j\partial_{x_k}-x_k \partial_{x_j}, \qquad 1\le j,k\le 4,\qquad j\ne k,\end{gather*}
where $L_{jk}=-L_{kj}$. The generators for f\/lat space conformal symmetries~(\ref{conformalKV}) are related to these via
\begin{gather*}
P_1=L_{13}+iL_{14},\qquad P_2=L_{23}+iL_{24},\qquad D=iL_{34},\qquad J=L_{12},\\
 K_j=L_{j3}-iL_{j4},\qquad j=1,2.\end{gather*}

The generators for $2$-sphere conformal symmetries are related to the $L_{jk}$ via
\begin{gather} L_{12}=J_{12}=s_1p_{s_2}-s_2p_{s_1},\qquad L_{31}=J_{31},\qquad L_{23}=J_{23},\nonumber\\
 L_{j4}=-ip_{s_j},\qquad j=1,2,3.\label{2sphere}\end{gather}
In identifying tetraspherical coordinates we can always permute the parameters $1,\dots,4$. Also, we can apply an arbitrary ${\rm SO}(4,\mathbb {C})$ transformation to the tetraspherical coordinates, so the above relations between Euclidean and tetraspherical coordinates are far from being unique.

\subsection{The 6 Laplace superintegrable systems with nondegenerate potentials}\label{3.2}

The systems are all of the form
\begin{gather*}
\left(\sum_{j=1}^4\partial_{x_j}^2+V({\bf x})\right)\Psi=0
\end{gather*}
in tetraspherical coordinates, or $\big(\partial_x^2+\partial_y^2+{\tilde V}\big)\Psi=0$ as a f\/lat space system in Cartesian coordinates. Each Laplace system is an equivalence class of St\"ackel equivalent Helmholtz systems. In each case the expression for ${\mathcal R}^2$ in the conformal symmetry algebra can be put in a normal form which is a polynomial in~${\mathcal L}_j$, $a_k$~of order $\le 3$. We show the terms of order $\ge 2$ in the~${\mathcal L}_j$ alone. The parameter~$\alpha$ is linear in the~$a_j$. The remaining terms are of lower order in the ${\mathcal L}_j$:~${\rm LOT}$. The potentials are
\begin{gather*}
V_{[1,1,1,1]}=\frac{a_1}{x_1^2}+\frac{a_2}{x_2^2}+\frac{a_3}{x_3^2}+\frac{a_4}{x_4^2},\\
{\tilde V}_{[1,1,1,1]}=\frac{a_1}{x^2}+\frac{a_2}{y^2}+\frac{4a_3}{(x^2+y^2-1)^2}-\frac{4a_4}{(x^2+y^2+1)^2},\\
{\mathcal R}^2={\mathcal L}_1{\mathcal L}_2({\mathcal L}_1+{\mathcal L}_2)+\alpha{\mathcal L}_1{\mathcal L}_2+{\rm LOT}.
\end{gather*}
St\"ackel equivalent systems: $S9$, $S8$, $S7$, $D4B(b)$, $D4C(b)$, $K[1,1,1,1]$.
\begin{gather*}
 V_{[2,1,1]}=\frac{a_1}{x_1^2}+\frac{a_2}{x_2^2}+\frac{a_3(x_3-ix_4)}{(x_3+ix_4)^3}+\frac{a_4}{(x_3+ix_4)^2},\\
{\tilde V}_{[2,1,1]}=\frac{a_1}{x^2}+\frac{a_2}{y^2}-a_3\big(x^2+y^2\big)+a_4,\\
 {\mathcal R}^2={\mathcal L}_1^2{\mathcal L}_2+\alpha{\mathcal L}_2^2+{\rm LOT}.
 \end{gather*}
St\"ackel equivalent systems: $ S4$, $S2$, $E1$, $E16$, $D4A(b)$, $D3B$, $D2B$, $D2C$, $K[2,1,1]$.
\begin{gather}\label{V[22norm']} V_{[2,2]}=\frac{a_1}{(x_1+ix_2)^2}+\frac{a_2(x_1-ix_2)}{(x_1+ix_2)^3}
+\frac{a_3}{(x_3+ix_4)^2}+\frac{a_4(x_3-ix_4)}{(x_3+ix_4)^3},\\
{\tilde V}_{[2,2]}=\frac{a_1}{(x+iy)^2}+\frac{a_2(x-iy)}{(x+iy)^3}+a_3-a_4\big(x^2+y^2\big),\nonumber\\
 {\mathcal R}^2={\mathcal L}_1^2{\mathcal L}_2+{\rm LOT}.\nonumber\end{gather}
St\"ackel equivalent systems: $E8$, $E17$, $E7$, $E19$, $D3C$, $D3D$, $K[2,2]$.
\begin{gather*}
V_{[3,1]}=\frac{a_1}{(x_3+ix_4)^2}+\frac{a_2x_1}{(x_3+ix_4)^3} +\frac{a_3(4{x_1}^2+{x_2}^2)}{(x_3+ix_4)^4}+\frac{a_4}{{x_2}^2},\\
 {\tilde V}_{[3,1]}=a_1-a_2x +a_3\big(4x^2+{y}^2\big)+\frac{a_4}{{y}^2},\\
{\mathcal R}^2={\mathcal L}_1^3+\alpha{\mathcal L}_2^2+{\rm LOT}.\end{gather*}
St\"ackel equivalent systems: $S1$, $E2$, $D1B$, $D2A$, $K[3,1]$.
\begin{gather*}
V_{[4]}=\frac{a_1}{(x_3+ix_4)^2}+a_2\frac{x_1+ix_2}{(x_3+ix_4)^3} +a_3\frac{3(x_1+ix_2)^2-2(x_3+ix_4)(x_1-ix_2)}{(x_3+ix_4)^4}\\
\hphantom{V_{[4]}=}{} +a_4\ \frac{4(x_3+ix_4)(x_3^2+x_4^2)+2(x_1+ix_2)^3}{(x_3+ix_4)^5},\\
 {\tilde V}_{[4]}=a_1-a_2(x+iy) +a_3\big(3(x+iy)^2+2(x-iy)\big) -a_4\big(4\big(x^2+y^2\big)+2(x+iy)^3\big),\\
{\mathcal R}^2={\mathcal L}_1^3+\alpha{\mathcal L}_1{\mathcal L}_2+{\rm LOT}.\end{gather*}
St\"ackel equivalent systems: $E10$, $E9$, $D1A$, $K[4]$.
\begin{gather*}
V_{[0]}=\frac{a_1}{(x_3+ix_4)^2}+\frac{a_2x_1+a_3x_2}{(x_3+ix_4)^3}+a_4\frac{x_1^2+x_2^2}{(x_3+ix_4)^4},\\
 {\tilde V}_{[0]}=a_1-(a_2x+a_3y)+a_4\big(x^2+y^2\big),\\
 {\mathcal R}^2=\alpha{\mathcal L}_1{\mathcal L}_2+{\rm LOT}.\end{gather*}
St\"ackel equivalent systems: $ E20$, $E11$, $E3'$, $D1C$, $D3A$, $K[0]$.

\section{Def\/inition and composition of B\^ocher contractions} \label{Bc}

Before introducing precise def\/initions, let us note that all geometrical contractions of ${\mathfrak{e}}(2,\mathbb {C})\to {\mathfrak{e}}(2,\mathbb {C})$ and ${\mathfrak{so}}(3,\mathbb {C})\to {\mathfrak{so}}(3,\mathbb {C}), {\mathfrak{e}}(2,\mathbb {C})$, i.e., pointwise coordinate limits of functions on f\/lat space or the sphere as classif\/ied in~\cite{KM2014}, induce geometrical contractions of ${\mathfrak{so}}(4,\mathbb {C})\to {\mathfrak{so}}(4,\mathbb {C})$. Recall that a basis for ${\mathfrak{so}}(4,\mathbb {C})$ is~(\ref{conformalKV}) where the subset $P_1$, $P_2$, $J$ forms a basis for ${\mathfrak{e}}(2,\mathbb {C})$. As an example, consider the coordinate limit $x=\epsilon x'$, $y=\epsilon y'$. This induces the contraction $\epsilon P_1=P'_1$, $\epsilon P_2=P'_2$, $J=J'$ of ${\mathfrak{e}}(2,\mathbb {C})$ and, further, the contraction $D=D'$, $K_1=\epsilon K_1'$, $K_2=\epsilon K_2'$ of~${\mathfrak{so}}(4,\mathcal{C})$. The other contractions of ${\mathfrak{e}}(2,\mathbb {C})$ work similarly.

For ${\mathfrak{so}}(3,\mathbb {C})$ we have the basis $J_{23}$, $J_{31}$, $J_{12}$, where
\begin{gather*} s_1^2+s_2^2+s_3^2=1,\qquad s_1p_{s_1}+s_2p_{s_2}+s_3p_{s_3}=0.\end{gather*}
The generators for the conformal symmetry algebra of the ${\mathfrak{so}}(3,\mathbb {C})$ Laplace equation are related to the $L_{jk}$ basis for
${\mathfrak{so}}(4,\mathbb {C})$ via~(\ref{2sphere}). Now consider the example limit $s_1=\epsilon x'$, $s_2=\epsilon y'$. It induces the contraction
\begin{gather*}\epsilon J_{23}= -p_{y'},\qquad \epsilon J_{31}=p_{x'},\qquad J_{12}=x'p_{y'}-y'p_{x'}\end{gather*}
of ${\mathfrak{so}}(3,\mathbb {C})$ to ${\mathfrak{e}}(2,\mathbb {C})$ and the contraction
\begin{gather*} L_{12}=x'p_{y'}-y'p_{x'}=J',\qquad i\epsilon L_{14}=p_{x'}=P_1',\qquad i\epsilon L_{24}=p_{y'}=P_2',\\
-\frac{2}{\epsilon}(iL_{14}+ L_{13})=\big({x'}^2-{y'}^2\big)p_{x'}+2x'y'p_{y'}+O(\epsilon)=K_1'+O(\epsilon),\\
-\frac{2}{\epsilon}(iL_{24}+ L_{23})=\big({y'}^2-{x'}^2\big)p_{y'}+2x'y'p_{x'}+O(\epsilon)=K_2'+O(\epsilon),\end{gather*}
of ${\mathfrak{so}}(4,\mathbb {C})$ to itself. The other contractions of ${\mathfrak{so}}(3,\mathbb {C})$ work similarly.

We now present a general def\/inition of B\^ocher contractions of ${\mathfrak{so}}(4,\mathbb {C})$ to itself and demonstrate that the above induced contractions can be reformulated as B\^ocher contractions. Let ${\bf x}={\bf A}(\epsilon){\bf y}$, and ${\bf x}=(x_1,\dots,x_4)$, ${\bf y}=(y_1,\dots,y_4)$ be column vectors, and ${\bf A}=(A_{jk}(\epsilon))$, be a~$4\times 4$ matrix with matrix elements
\begin{gather}\label{Amatrix} A_{kj}(\epsilon)=\sum_{\ell=-N}^Na^\ell_{kj}\epsilon^\ell,\end{gather} where $N$ is a nonnegative integer and the $a^\ell_{kj}$ are complex constants. (Here, $N$ can be arbitrarily large, but it must be f\/inite in any particular case.) We say that the matrix $\bf A$ def\/ines a {\it B\^ocher contraction} of the conformal algebra ${\mathfrak{so}}(4,\mathbb {C})$ to itself provided
\begin{gather}\label{condition1} 1) \quad {\det} ({\bf A})=\pm 1, \ {\rm constant\ for\ all\ }\epsilon\ne 0,\\
\label{Bocher2} 2) \quad {\bf x}\cdot{\bf x}\equiv \sum_{j=1}^4x_i(\epsilon)^2={\bf y}\cdot{\bf y}+O(\epsilon).
\end{gather}
If, in addition, ${\bf A}\in O(4,\mathbb {C})$ for all $\epsilon\ne 0$ the matrix $\bf A$ def\/ines a {\it special B\^ocher contraction}. For a special B\^ocher contraction ${\bf x}\cdot{\bf x}={\bf y}\cdot{\bf y}$, with no error term.

We explain why this is a contraction in the generalized In\"on\"u--Wigner sense. Let $L_{ts}=x_t\partial_{x_s}-x_s\partial_{x_t}$, ${s\ne t}$
be a generator of ${\mathfrak{so}}(4,\mathbb {C})$ and ${\bf \tilde A}(\epsilon)={\bf A}^{-1}(\epsilon)$ be the matrix inverse. (Note that $\bf \tilde A$ also has an expansion of the form~(\ref{Amatrix}) in~$\epsilon$.) We have the expansion
\begin{gather}\label{L2} L_{ts}=\sum_{k,\ell}(A_{tk}{\tilde A}_{\ell s}-A_{sk}{\tilde A}_{\ell t})y_k\partial_{y_\ell}
=\epsilon^{\alpha_{ts}}\left(\sum_{k\ell}F_{k\ell}\ y_k\partial_{y_\ell}+O(\epsilon)\right),\end{gather}
where $\bf F$ is a constant nonzero matrix. Thus the integer $\alpha_{ts}$ is the smallest power of~$\epsilon$ occurring in the expansion of $L_{ts}$. Now consider the product $L_{ts}( {\bf x}\cdot {\bf x})$. On one hand it is obvious that $L_{ts}({\bf x}\cdot {\bf x})\equiv 0$, but on the other hand the expansions~(\ref{Bocher2}),~(\ref{L2}) yield
\begin{gather*} L_{ts}( {\bf x}\cdot {\bf x})=\epsilon^{\alpha_{ts}}\left(\sum_{k\ell}F_{k\ell}\ y_k\partial_{y_\ell}\right)\left(\sum_{j}y_j^2\right) +O\big(\epsilon^{\alpha_{ts}}\big).\end{gather*}
Thus, $\big(\sum_{k\ell}F_{k\ell}\ y_k\partial_{y_\ell}\big)\big(\sum_{j}y_j^2\big)\equiv 0$ for $\bf F$ a constant nonzero matrix. However, the only dif\/ferential operators of the form $\sum_{k\ell}F_{k\ell} y_k\partial_{y_\ell}$ that map ${\bf y}\cdot {\bf y}$ to zero are elements of ${\mathfrak{so}}(4,\mathbb {C})$:
\begin{gather*}\sum_{k\ell}F_{k\ell} y_k\partial_{y_\ell}=\sum_{j>k}b_{jk}L'_{jk},\qquad L'_{jk}=y_j\partial_{y_k}-y_k\partial_{y_j}.\end{gather*}
Thus
\begin{gather}\label{Bocherbasis1} \lim_{\epsilon \longrightarrow 0}\epsilon^{-\alpha_{ts}}L_{ts}=\sum_{j>k}b_{jk}L'_{jk}\equiv L'\end{gather}
and this determines a limit of $L_{ts}$ to~$L'$. Similarly, if we apply this same procedure to the operator $L=\sum_{t>s}c(\epsilon)_{ts}L_{ts}$ for any rational polynomials $c_{ts}(\epsilon)$ we will obtain an operator $L'=\sum_{j>k}b_{jk}L'_{jk}$ in the limit. Further, due to condition~(\ref{condition1}), by choosing the $c(\epsilon)_{ts}$ appropriately we can obtain any $L'\in {\mathfrak{so}}(4,\mathbb {C})$ in the limit. (Indeed, modulo rational functions of $\epsilon$, this is just the adjoint action of $O(4,\mathbb {C})$ on ${\mathfrak{so}}(4,\mathbb {C})$. In this sense the mapping $L\to L'$ is onto.)

\begin{Theorem} Suppose the matrix ${\bf A}(\epsilon)$ defines a B\^ocher contraction of ${\mathfrak{so}}(4,\mathbb {C})$. Let $\{L_{t_is_i},$ $i=1,\dots, 6\}$ be an ordered linearly independent for ${\mathfrak{so}}(4,\mathbb {C})$ such that $\alpha_{t_1s_1}\le \alpha_{t_2s_2}\le \cdots \le \alpha_{t_6s_6}$. Then there is an ordered linearly independent set $\{ L_j,\, j=1,\dots,6\}$ for ${\mathfrak{so}}(4,\mathbb {C})$ such that
\begin{enumerate}\itemsep=0pt
\item[$1)$] $L_j\in {\rm span}\{L_{t_is_i},\, i=1,\dots, j\}$,
\item[$2)$] there are integers $\alpha_1\le \alpha_2\le\cdots\le\alpha_6$ such that
\begin{gather*}\lim_{\epsilon \to 0}\frac{L_j}{\epsilon^{\alpha_j}}=L_j',\qquad 1\le\ j\le 6,\end{gather*}
and $\{L'_j,\, j=1,\dots,6\}$ forms a basis for ${\mathfrak{so}}(4,\mathbb {C})$ in the $y_k$ variables.
\end{enumerate}
\end{Theorem}

\begin{proof}The proof is by induction on $j$. For $j=1$ the result follows from (\ref{Bocherbasis1}). Assume the assertion is true for $j\le j_0<6$. Then, due to the nonsingularity condition (\ref{condition1}), we can always f\/ind polynomials in $\epsilon$, $\{a_1(\epsilon),a_2(\epsilon),\dots, a_{j_0}(\epsilon)\}$ such that
\begin{gather*} L_{j_0+1}=L_{t_{j_0+1},s_{j_0+1}}-\sum_{i=1}^{j_0} a_iL_i=\epsilon^{\alpha_{j_0+1}}L'_{j_0+1}+O\big(\epsilon^{\alpha_{j_0+2}}\big),\end{gather*}
where $L'_{j_0+1}$ is linearly independent of $\{L_i',\, 1\le i\le j_0\}$ and $\alpha_{j_0+1}\ge \alpha_{j_0}$.
\end{proof}

In \cite{KMS2016} we have used this theorem to compute explicitly the bases for the basic B\^ocher contractions.

\subsection{Composition of B\^ocher contractions}
Let $\bf A$ and $\bf B$ def\/ine B\^ocher contractions of ${\mathfrak{so}}(4,\mathbb {C})$ to itself. Thus there exist expansions
\begin{gather*} {\bf x}(\epsilon_1)\cdot {\bf x}(\epsilon_1)= {\bf y}\cdot {\bf y}+ O\big(\epsilon_1^a\big),\qquad {\bf y}(\epsilon_2)\cdot {\bf y}(\epsilon_2)= {\bf z}\cdot{\bf z}+ O\big(\epsilon_2^b\big),\end{gather*}
where
\begin{gather*} {\bf x}= {\bf A}(\epsilon_1){\bf y},\qquad {\bf y}(\epsilon_2)= {\bf B}(\epsilon_2) {\bf z}.\end{gather*}
 Now let
 \begin{gather*}{\bf x}(\epsilon_1,\epsilon_2)= {\bf A}(\epsilon_1) {\bf y}(\epsilon_2)={\bf A}(\epsilon_1){\bf B}(\epsilon_2){\bf z}.\end{gather*} Then
 \begin{gather*} {\bf x}(\epsilon_1,\epsilon_2)\cdot {\bf x}(\epsilon_1,\epsilon_2)={\bf y}(\epsilon_2)\cdot {\bf y}(\epsilon_2)
 +O_{\epsilon_2}\big(\epsilon_1^a\big)={\bf z}\cdot{\bf z}+O\big(\epsilon_2^b\big)+\epsilon_1^af(\epsilon_1,\epsilon_2,{\bf y}).\end{gather*}

Now set $\epsilon_1=\epsilon^m$, $\epsilon_2=\epsilon$. It follows from these expansions that we can always f\/ind an $m>0$
such that
\begin{gather*}
{\bf x}\big(\epsilon^m,\epsilon\big)\cdot {\bf x}\big(\epsilon^m,\epsilon\big)={\bf z}\cdot{\bf z}+O\big(\epsilon^q\big)\end{gather*}
and
\begin{gather*}\lim_{\epsilon\to 0} \epsilon^{-\alpha_{ts}}L_{ts}=\sum_{j>k}c_{jk}L''_{jk}\equiv L''\end{gather*}
for some $q>0$, with $L''$ in the ${\mathfrak{so}}(4,\mathbb {C})$ Lie algebra of operators such that $L''({\bf z}\cdot {\bf z})=0$. Thus this composition of the $A$ and $B$ contractions yields a new B\^ocher contraction. For special B\^ocher contractions the composition is def\/ined without restriction and the resulting contraction is uniquely determined for $\epsilon_1$, $\epsilon_2$ going to~0 independently. However, if we set $\epsilon_2=\epsilon_1^m$, in general the resulting contraction will depend on~$m$.

 \subsection{Special B\^ocher contractions}
Special B\^ocher contractions are much easier to understand and manipulate than general B\^ocher contractions: composition is merely matrix
multiplication. The contractions that arise from the B\^ocher recipe are not ``special''. However, we shall show that we can
associate a special B\^ocher contraction with each contraction obtained from B\^ocher's recipe, such that the special
contraction contains the same basic geometrical information. The (projective) tetraspherical coordinates are associated with points $(x,y)$
in 2D f\/lat space via the relation
\begin{gather} \label{tetrflatrelations1} (x,y)\equiv (x_1,x_2,x_3,x_4) =[x_3+ix_4]\left(-x, -y, \frac12\big(1-x^2-y^2\big), -\frac{i}{2}\big(1+x^2+y^2\big)\right).\end{gather}
In particular,
\begin{gather} \label{tetrflatrelations2} x=-\frac{x_1}{x_3+ix_4},\qquad y=-\frac{x_2}{x_3+ix_4},\qquad \frac{x_3+ix_4}{x_3-ix_4}=\frac{-1}{x^2+y^2}.\end{gather}
For coordinates on the 2-sphere we have
\begin{gather*} 
(s_1,s_2,s_3)\equiv (x_1,x_2,x_3,x_4)=x_4(-is_1, -is_2, is_3, 1).\end{gather*}

The action of B\^ocher contractions on the f\/lat space coordinates $(x,y)$ is an af\/f\/ine mapping and this af\/f\/ine action carries all of the geometrical information about the contraction. For example, the $[1,1,1,1]\downarrow [2,1,1]$ contraction
\begin{gather*}x_3=-\frac{i}{\sqrt{2}\ \epsilon}x_3'-\frac{i}{\sqrt{2} \epsilon}x_4',\qquad x_4=\frac{i}{\sqrt{2}}\left(\frac{1}{\epsilon}-\epsilon\right)x_3'
-\frac{1}{\sqrt{2}}\left(\frac{1}{\epsilon}+\epsilon\right)x_4',\end{gather*}
and $x_1=x'_1$, $x_2=x'_2$, gives
\begin{gather*} x=-\frac{x_1}{x_3+ix_4}=\frac{\epsilon x_1'}{ \sqrt{2}(x_3'+ix_4')}+O\big(\epsilon^2\big)=\epsilon' x' +O\big({\epsilon'}^2\big),\qquad y=\epsilon' y'+O\big({\epsilon'}^2\big),\end{gather*}
for $\epsilon'=\epsilon/\sqrt(2)$. Thus the geometric content of the action of this contraction in f\/lat space is $x=\epsilon' x'$, $y=\epsilon'y'$. The terms of order ${\epsilon'}^2$ disappear in the limit. On the complex sphere we have
\begin{gather*} s_1=\frac{ix_1}{x_4}=-\frac{\sqrt{2}i \epsilon x_1'}{x_3'+ix_4'}+O\big(\epsilon^2\big)=\epsilon' x'+O\big({\epsilon'}^2\big),\qquad s_2=\epsilon' y'+O\big({\epsilon'}^2\big),\\
s_3=-\frac{ix_3}{x_4}=1+O\big(\epsilon^2\big),\end{gather*}
where $\epsilon'=\sqrt{2} i \epsilon$ and $x'$, $y'$ are f\/lat space coordinates. Thus the geometric content of the action of this contraction on the 2-sphere is $s_1=\epsilon' x'$, $s_2=\epsilon'y'$. Note that distinct contractions on f\/lat space and the sphere are induced by the same B\^ocher contraction.

Using the fact that the contraction limits are completely determined by the geometric limits, we can derive special B\^ocher contractions that
produce the same geometric limits. We again consider the example discussed above. We will design a special B\^ocher contraction with the property $x=\epsilon x'$, $y=\epsilon y'$ such that equations (\ref{tetrflatrelations1}),~(\ref{tetrflatrelations2}) hold. In this case we require $x=x_1/(x_3+ix_4)=\epsilon x'=\epsilon\ x_1'/(x_3'+ix_4')$,
$y=x_2/(x_3+ix_4)=x_2'/(x_3'+ix_4')$. The solution is, essentially unique up to conformal transformation:
\begin{gather*}x_1= x_1',\qquad x_3=x_3'(\epsilon+1/\epsilon)/2+ix_4'(-\epsilon+1/\epsilon)/2,\\
x_2=x_2',\qquad x_4=ix_3'(\epsilon-1/\epsilon)/2+x_4'(\epsilon+1/\epsilon)/2.\end{gather*}
This contraction satisf\/ies $x_1^2+x_2^2+x_3^2+x_4^2={x_1'}^2+{x'}_2^2+{x_3'}^2+{x_4'}^2$ and agrees with $[1,1,1,1]\downarrow [2,1,1]$ on Laplace equations.

Similarly we can use each of the geometric contractions of f\/lat space and the 2-sphere as classif\/ied in \cite{KM2014}, to construct
special B\^ocher contractions that take $V_{[1,1,1,1]}$ to each of $V_{[2,1,1]}$, $V_{[2,2]}$, $V_{[3,1]}$, $V_{[4]}$. For example
\begin{alignat*}{3} 
& V_{[1,1,1,1]}\to V_{[3,1]}\colon \quad && x_1=x'_1+\frac{x'_3}{\epsilon}+\frac{ix'_4}{\epsilon},\qquad x_3=-\frac{x'_1}{\epsilon}+x'_3\left(1-\frac{1}{2\epsilon^2}\right)-\frac{ix_4'}{2\epsilon^2},& \\
&&& x_2=x'_2,\qquad x_4=-\frac{ix'_1}{\epsilon}-\frac{ix'_3}{2\epsilon^2}+x'_4\left(1+\frac{1}{2\epsilon^2}\right).\end{alignat*}

A more general way to construct special B\^ocher contractions is to make use of the normal forms for conjugacy classes of ${\mathfrak{so}}(4,\mathbb {C})$ under the adjoint action of ${\rm SO}(4,\mathbb {C})$. They are derived in \cite{Gant}:
\begin{gather*}
C_1= \left(\begin{matrix} 0&\lambda&0&0 \\ -\lambda&0&0&0\\ 0&0&0&0\\0&0&0&0\end{matrix}\right),\qquad C_2=
\left(\begin{matrix} 0&\lambda&0&0 \\ -\lambda&0&0&0\\ 0&0&0&\mu\\0&0&-\mu&0\end{matrix}\right), \\
C_3= \left(\begin{matrix} 0&1+i&0&0 \\ -1-i&0&-1+i&0\\ 0&1-i&0&0\\0&0&0&0\end{matrix}\right),\qquad C_4=
\frac12\left(\begin{matrix} 0&1&i&2\lambda \\ -1&0&2\lambda&i\\ -i&-2\lambda&0&-1\\-2\lambda&-i&1&0\end{matrix}\right). \end{gather*}
Every 1-parameter subgroup ${\bf A}(t)$ of ${\rm SO}(4,\mathbb {C})$ (i.e., ${\bf A}(t_1+t_2)={\bf A}(t_1){\bf A}(t_2)$), is conjugate to one of the forms ${\bf A}_j(t)=\exp(tC_j)$, $j=1,2,3,4$. By making an appropriate change of complex coordinate $t=t(\epsilon)$ we can obtain a special B\^ocher contraction matrix
\begin{gather}\label{form1} {\bf A}_1(t)= \frac12\left(\begin{matrix} \frac{\epsilon^2+1}{\epsilon}& -\frac{i(\epsilon^2-1)}{\epsilon}&0&0 \\
\frac{i(\epsilon^2-1)}{\epsilon}& \frac{\epsilon^2+1}{\epsilon}&0&0\\ 0&0&0&0\\0&0&0&0\end{matrix}\right),\qquad \epsilon =e^{i\lambda t},\\
\label{form2} {\bf A}_2(t)= \frac12\left(\begin{matrix} \frac{\epsilon_1^2+1}{\epsilon_1}& -\frac{i(\epsilon_1^2-1)}{\epsilon_1}&0&0 \\
\frac{i(\epsilon_1^2-1)}{\epsilon_1}& \frac{\epsilon_1^2+1}{\epsilon_1}&0&0\\0&0& \frac{\epsilon_2^2+1}{\epsilon_2}&
-\frac{i(\epsilon_2^2-1)}{\epsilon_2}\\0&0&\frac{i(\epsilon_2^2-1)}{\epsilon_2}& \frac{\epsilon_2^2+1}{\epsilon_2}\end{matrix}\right),
\qquad \epsilon_1 =e^{i\lambda t},\qquad \epsilon_2 =e^{i\mu t},\\
\label{form3} {\bf A}_3(t)= \left(\begin{matrix} 1-\frac{1}{2\epsilon^2}&\frac{1}{\epsilon}&\frac{i}{2\epsilon^2}&0 \\
-\frac{1}{\epsilon}&1&\frac{i}{\epsilon}&0\\
\frac{i}{2\epsilon^2}&-\frac{i}{\epsilon}&1+\frac{1}{2\epsilon^2}&0\\0&0&0&1\end{matrix}\right),\qquad \epsilon =\frac{2}{t(1+i)},\\
\label{form4} {\bf A}_4(t)=
\frac12\left(\begin{matrix} \frac{\epsilon_1^2+1}{\epsilon_1}&\frac{1}{\epsilon_1\epsilon_2}&\frac{i}{\epsilon_1\epsilon_2}&
\frac{i(\epsilon_1^2-1)}{\epsilon_1} \\ -\frac{\epsilon_1}{\epsilon_2}&\frac{\epsilon_1^2+1}{\epsilon_1}&\frac{i(\epsilon_1^2-1)}{\epsilon_1}
&\frac{i\epsilon_1}{\epsilon_2}
\\ -\frac{i\epsilon_1}{\epsilon_2}&\frac{i(\epsilon_1^2-1)}{\epsilon_1}&\frac{\epsilon_1^2+1)}{\epsilon_1}&-\frac{\epsilon_1}{\epsilon_2}\\
\frac{i(\epsilon_1^2-1)}{\epsilon_1}&\frac{i}{\epsilon_1\epsilon_2}&\frac{1}{\epsilon_1\epsilon_2} &\frac{\epsilon_1^2+1)}{\epsilon_1}\end{matrix}\right),
\qquad \epsilon_1=e^{i\lambda t},\qquad \epsilon_2=\frac{1}{t}.\end{gather}
The contraction (\ref{form1}) takes $V_{[1,1,1,1]}$ to $V_{[2,11]}$, (\ref{form2}) takes it to $V_{[2,2]}$, and (\ref{form3}) takes it to $V_{[3,1]}$. The contractions (\ref{form4}), on the other hand, takes $V_{[1,1,1,1]}$ to $V_{[2,2]}$ again. Consider though the special case ${\bf H} (\epsilon)$, of (\ref{form4}) where $\epsilon_1=1$, $\epsilon_2=\epsilon$. It, too, maps $V_{[1,1,1,1]}$ to $V_{[2,2]}$, but the composition ${\bf H}(\epsilon){\bf H}(\epsilon^2)$ takes $V_{[1,1,1,1]}$ to $V_{[4]}$. (We note that the composition ${\bf H}(\epsilon){\bf H}(\epsilon^3)$ takes $V_{[1,1,1,1]}$ to $V_{[3,1]}$, showing that, in general, the result of a~composition ${\bf A}(\epsilon_1){\bf B}(\epsilon_2)$ depends on the relationship between~$\epsilon_1$ and~$\epsilon_2$.)

If the matrix ${\bf A}(\epsilon)$ def\/ines a general B\^ocher contraction, by transposing two rows if necessary, we can assume $\det{\bf A}(\epsilon)=1$ for all $\epsilon\ne 0$. Thus, ${\bf A}(\epsilon)$ is a curve on~${\rm SL}(4,\mathbb {C})$. We could use the results of~\cite{Gant} to list all the conjugacy classes of ${\mathfrak{sl}}(4,\mathbb {C})$ to attempt a classif\/ication. However, it would be necessary to check condition~(\ref{Bocher2}) in every case, whereas for special B\^ocher contractions this condition is satisf\/ied automatically.

Both B\^ocher's original recipes and the normal forms given above provide a generating basis for all B\^ocher contractions in two dimensions;
the general contractions are obtained by composing these generators.

\section[Classif\/ication of free abstract nondegenerate quadratic algebras. Identif\/ication of those from free nondegenerate 2nd order superintegrable systems]{Classif\/ication of free abstract nondegenerate quadratic\\ algebras. Identif\/ication of those from free nondegenerate\\ 2nd order superintegrable systems}

\subsection{Free nondegenerate classical quadratic algebras}
Recall from Def\/inition~\ref{def1} that the symmetry algebra of a free 2D superintegrable system on a constant curvature space, $\mathcal{A}$,
is a quadratic algebra which is completely determined by the function~$\mathcal F$. More specif\/ically, it is a Poisson algebra generated by three linearly independent elements $\{\mathcal{L}_1, \mathcal{L}_2,\mathcal{H}\}$ where $\mathcal{H}$ generates the center of $\mathcal{A}$ and the structure equations of the algebra are given by~(\ref{R2eqn}) with
\begin{gather*}
{\mathcal{R}}^2={\mathcal{F}}(\mathcal{H} ,\mathcal{L}_1,\mathcal{L}_2)\end{gather*}
for some third order homogeneous polynomial~$\mathcal{F}$. We call~$R^2$, which is the same as ${\mathcal{F}}(\mathcal{H},\mathcal{L}_1,\mathcal{L}_2)$, the Casimir of $\mathcal{A}$ in terms of $\{\mathcal{L}_1, \mathcal{L}_2,\mathcal{H}\}$. Motivated by the superintegrable case we def\/ine an \textit{abstract free nondegenerate $2D$ classical quadratic algebra} as follows.

\begin{Definition}\label{def5}
A free nondegenerate 2D classical quadratic algebra is a Poisson algebra $\mathcal{A}$ over $\mathbb {C}$ that is generated by
$\{\mathcal{L}_1, \mathcal{L}_2,\mathcal{H}\}$ where $\mathcal{H}$ generates the center of $\mathcal{A}$, \begin{gather*}
\left\{{\mathcal{R}},{\mathcal{L}}_1 \right\}=-\frac{1}{2}\frac{\partial {\mathcal{R}}^2 }{\partial {\mathcal{L}}_2},\qquad
\left\{{\mathcal{R}},{\mathcal{L}}_2 \right\}=\frac{1}{2}\frac{\partial {\mathcal{R}}^2 }{\partial {\mathcal{L}}_1},
\end{gather*}
 ${\mathcal{R}}= \{ {\mathcal{L}}_1, {\mathcal{L}}_2\}$, and ${\mathcal{R}}^2={\mathcal{F}}(\mathcal{H} ,\mathcal{L}_1,\mathcal{L}_2)$
 for some third order homogeneous polynomial $\mathcal{F}$.
\end{Definition}
Below we shall refer to free nondegenerate 2D classical quadratic algebras simply as (abstract) quadratic algebras.
\begin{Remark}
As an associative algebra $\mathcal{A}$ is the quotient of the free $\mathbb {C}$-algebra generated by $\{\mathcal{L}_1, \mathcal{L}_2,\mathcal{H},\mathcal{R}\}$ and its two sided ideal generated by ${\mathcal{R}}^2-{\mathcal{F}}$. For any choice of a polynomial of degree three for $\mathcal{F}$, the above equations def\/ine Lie brackets on~$\mathcal{A}$ that make it a~Poisson algebra, but higher order polynomials will not def\/ine Lie brackets on~$\mathcal{A}$.
\end{Remark}

For any other generating set $\widetilde{\mathcal{L}}_1$, $\widetilde{\mathcal{L}}_2$, $\widetilde{\mathcal{H}}$, $\widetilde{\mathcal{R}}$ of the
same Poisson algebra that satisf\/ies:
\begin{enumerate}\itemsep=0pt
 \item [(i)]The linear span over $\mathbb{C}$ of $\widetilde{\mathcal{L}}_1$, $\widetilde{\mathcal{L}}_2$, $\widetilde{\mathcal{H}}$ coincides with the linear span of ${\mathcal{L}}_1$, ${\mathcal{L}}_2$, ${\mathcal{H}}$.
 \item [(ii)] $\widetilde{\mathcal{H}}$ is in the center of the Poisson algebra, i.e., Poisson commutes with everything.
 \item [(iii)] $\widetilde{\mathcal{R}}= \{ \widetilde{\mathcal{L}}_1, \widetilde{\mathcal{L}}_2\}$.
 \item [(iv)] The generators $\widetilde{\mathcal{L}}_1$, $\widetilde{\mathcal{L}}_2$, $\widetilde{\mathcal{H}}$, $\widetilde{\mathcal{R}}$ satisfy the structure equations, i.e.,
 \begin{gather*}
\{\widetilde{\mathcal{R}},\widetilde{\mathcal{L}}_1 \}=-\frac{1}{2}\frac{\partial \widetilde{\mathcal{R}}^2 }
{\partial \widetilde{\mathcal{L}}_2},\qquad
 \{\widetilde{\mathcal{R}},\widetilde{\mathcal{L}}_2 \}=\frac{1}{2}\frac{\partial \widetilde{\mathcal{R}}^2 }{\partial \widetilde{\mathcal{L}}_1}.
\end{gather*}
\end{enumerate}
It easy to see that
\begin{gather}\label{basischange}
\left(\begin{matrix}
\widetilde{\mathcal{L}}_1\\
\widetilde{\mathcal{L}}_2 \\
\widetilde{\mathcal{H}}
\end{matrix}\right)
=\left(\begin{matrix}
A_{1,1} & A_{1,2}&A_{1,3} \\
A_{2,1}&A_{2,2} &A_{2,3} \\
0 &0 &A_{3,3}
\end{matrix}\right)
\left(\begin{matrix}
{\mathcal{L}}_1\\
{\mathcal{L}}_2 \\
{\mathcal{H}}
\end{matrix}\right)
\end{gather}
for some
\begin{gather}
\label{equa5}
A=\left(\begin{matrix}
A_{1,1} & A_{1,2}&A_{1,3} \\
A_{2,1}&A_{2,2} &A_{2,3} \\
0 &0 &A_{3,3}
\end{matrix}\right)\in {\rm GL}(3,\mathbb {C}).
\end{gather}
For a matrix as above we def\/ine
$A_2=\left(\begin{matrix} A_{1,1}&A_{1,2}\\ A_{2,1}&A_{2,2}\end{matrix}\right)\in {\rm GL}(2,\mathbb{C})$.
We denote the group of matrices of the form~(\ref{equa5}) by~$G$, it is a complex algebraic group. Moreover, if ${\mathcal{R}}^2={\mathcal{F}}$ and
$\tilde { \mathcal{R}}^2=\tilde { \mathcal{F}}$ then there is $A\in G$, such that
\begin{gather}\label{e8}
\widetilde{\mathcal{F}}(\widetilde{\mathcal{L}}_1,\widetilde{\mathcal{L}}_2,\widetilde{\mathcal{H}}) = \det (A_2)^2\mathcal{F}\big(A^{-1} \big(\widetilde{\mathcal{L}}_1,\widetilde{\mathcal{L}}_2,\widetilde{\mathcal{H}}\big) \big).
\end{gather}

Obviously, two quadratic algebras are isomorphic if and only if their Casimirs are related by $A\in G$ via equation~\eqref{e8}. This fact is fundamental for the classif\/ication of quadratic algebras.

Let $\mathbb{C}^{[3]}[x_1,x_2,x_3]$ be the complex algebraic variety of homogeneous polynomials of degree three in the variables~$x_1$, $x_2$, $x_3$. The group $G$ acts on $\mathbb{C}^{[3]}[x_1,x_2,x_3]$ via equation~(\ref{e8}). Obviously there is a bijection between isomorphism classes of quadratic algebras and orbits of~$G$ in $\mathbb{C}^{[3]}[x_1,x_2,x_3]$. We will determine all isomorphism classes of quadratic algebras by classifying all orbits of $G$ in $\mathbb{C}^{[3]}[x_1,x_2,x_3]$. We shall distinguish an element in each orbit that def\/ines the Canonical form for the Casimir of a~given quadratic algebra. Moreover we present an algorithm for f\/inding the canonical form of the Casimir for a given quadratic algebra which gives a~practical way to determine if two given quadratic algebras are isomorphic.

\subsection{The algorithm for casting the Casimir to its the canonical form}

In this section we introduce the notation $X_1=\mathcal{L}_1$, $X_2=\mathcal{L}_2$, $X_3=\mathcal{H}$ and similarly, $\widetilde{X}_1= \widetilde{\mathcal{L}}_1$, $\widetilde{X}_2=\widetilde{\mathcal{L}}_2$, $\widetilde{X}_3=\widetilde{\mathcal{H}}$. For any realization of the Casimir, ${ { {R}}}^2= {\mathcal{F}}(X_1,X_2,X_3)$, there are homogeneous polynomials in $X_1$, $X_2$ of order $j$, ${\mathcal{F}}^{(j)}$, such that
\begin{gather*} {\mathcal{F}}(X_1,X_2,X_3)= {\mathcal{F}}^{(3)}(X_1,X_2)+X_3{\mathcal{F}}^{(2)}(X_1,X_2) +{X_3}^2{\mathcal{F}}^{(1)}(X_1,X_2)+{X_3}^3{\mathcal{F}}^{(0)}.\end{gather*}
For any $f\in \mathbb{C}^{[3]}[X_1,X_2,X_3]$ we shall denote the stabilizer of $f$ in $G$ by $\operatorname{Stab}_{G}\{f\}$. We shall use the notation $\operatorname{Stab}_{G}\{f+O(\mathcal{H})\}$ for the subgroup of $G$ consisting of all elements that do not change the part in $f$ that is~$\mathcal{H}$ independent. That is $g\in \operatorname{Stab}_{G}\{f+O(\mathcal{H})\}$ preserves the lowest order term in $f$ as a polynomial of $\mathcal{H}=X_3$. Similarly $\operatorname{Stab}_{G}\{f+O(\mathcal{H}^2)\}$ stands for the subgroup of $G$ consisting of all elements that preserves the part in $f$ that is a polynomial of degree 1 in $\mathcal{H}$. Similarly we def\/ine $\operatorname{Stab}_{G}\{f+O(\mathcal{H}^3)\}$. For a given $f\in \mathbb{C}^{[3]}[X_1,X_2,X_3]$ we shall denote by $f^{(i)}(X_1,X_2)$ it homogeneous component that are uniquely def\/ined by \begin{gather*} f(X_1,X_2,X_3)= {f}^{(3)}(X_1,X_2)+X_3{f}^{(2)}(X_1,X_2) +{X_3}^2{f}^{(1)}(X_1,X_2)+{X_3}^3{f}^{(0)}.\end{gather*}

 Note that
\begin{gather*}
\operatorname{Stab}_{G}\big\{ f^{(3)} +O(\mathcal{H})\big\}\supseteq \operatorname{Stab}_{G}\big\{ f^{(3)} +{ \mathcal{H}} f^{(2)} +O\big(\mathcal{H}^2\big)\big\} \\
\qquad{} \supseteq \operatorname{Stab}_{G}\big\{ f^{(3)}+{ \mathcal{H}} f^{(2)}+{ \mathcal{H}}^2 f^{(1)}+O\big(\mathcal{H}^3\big)\big\} \supseteq \operatorname{Stab}_{G}\{ f\}.
 \end{gather*}

The algorithm for casting ${{{R}}}^2= {\mathcal{F}}(X_1,X_2,X_3)$ into its canonical form is as follows:
\begin{description}\itemsep=0pt
\item[Step1] Using a certain $g_1\in G$ we transform ${\mathcal{F}}(X_1,X_2,X_3)$ to a~form in which ${\mathcal{F}}^{(3)}$ is in a~canonical
form, ${\mathcal{F}}_c^{(3)}$.
\item[Step2] Using a certain $g_2\in \operatorname{Stab}_{G}\{ {\mathcal{F}}_c^{(3)} +O(\mathcal{H})\}$ we transform
${\mathcal{F}}(X_1,X_2,X_3)$ (that we got in step 1) to a form in which ${\mathcal{F}}^{(3)}+\mathcal{H}{\mathcal{F}}^{(2)}$ is in a canonical
form ${\mathcal{F}}_c^{(3)}+\mathcal{H}{\mathcal{F}}_c^{(2)}$.
\item[Step3] Using a certain $g_3\in \operatorname{Stab}_{G}\{ {\mathcal{F}}_c^{(3)}+\mathcal{H}{\mathcal{F}}_c^{(2)}+O(\mathcal{H}^2)\}$ we
transform ${\mathcal{F}}(X_1,X_2,X_3)$ (that we got in step 2) to a form in which ${\mathcal{F}}^{(3)}+\mathcal{H}{\mathcal{F}}^{(2)}
+\mathcal{H}^2{\mathcal{F}}^{(1)}$ is in a canonical form ${\mathcal{F}}_c^{(3)}+\mathcal{H}{\mathcal{F}}_c^{(2)}+\mathcal{H}^2{\mathcal{F}}_c^{(1)}$.
\item [Step4] Using a certain $g_4\in \operatorname{Stab}_{G}\{ {\mathcal{F}}_c^{(3)}+\mathcal{H}{\mathcal{F}}_c^{(2)}
+\mathcal{H}^2{\mathcal{F}}_c^{(1)}+O(\mathcal{H}^3)\}$ we transform ${\mathcal{F}}(X_1,X_2,X_3)$ (that we got in step 3)
to a form in which ${\mathcal{F}}^{(3)}+\mathcal{H}{\mathcal{F}}^{(2)}+\mathcal{H}^2{\mathcal{F}}^{(1)}+\mathcal{H}^3{\mathcal{F}}^{(0)}$ is
in a canonical form ${\mathcal{F}}_c^{(3)}+\mathcal{H}{\mathcal{F}}_c^{(2)}+\mathcal{H}^2{\mathcal{F}}_c^{(1)}+\mathcal{H}^3{\mathcal{F}}_c^{(0)}$.
This is the canonical form of $ \mathcal{F}$.
\end{description}
At the end of the section we list all possible canonical form of quadratic algebras in a table.

 \subsubsection[The four cases for $\mathcal{F}^{(3)}$]{The four cases for $\boldsymbol{\mathcal{F}^{(3)}}$}
Note that for two presentations of the Casimir of a given quadratic algebra: ${ { {R}}}^2= {\mathcal{F}}(X_1,X_2,X_3)$ and ${\widetilde{\mathcal{R}}}^2=\widetilde{\mathcal{F}}(\widetilde{X}_1,\widetilde{X}_2, \widetilde{X})$ that are related by
 equation~(\ref{basischange}) with
 $A=\left(\begin{matrix}
A_{1,1} & A_{1,2}&0 \\
A_{2,1}&A_{2,2} &0 \\
0 &0 &1
\end{matrix}\right)\in {\rm GL}(3,\mathbb {C})$
and
\begin{gather*}
{\widetilde{\mathcal{R}}}^2=\widetilde{\mathcal{F}}^{(3)}\big(\widetilde{X}_1,\widetilde{X}_2\big)+
 \widetilde{X}_3\widetilde{\mathcal{F}}^{(2)}\big(\widetilde{X}_1,\widetilde{X}_2\big)
 +\widetilde{X}_3^2\widetilde{\mathcal{F}}^{(1)}\big(\widetilde{X}_1,\widetilde{X}_2\big)+\widetilde{X}_3^3\widetilde{\mathcal{F}}^{(0)}
 \end{gather*}
we have
\begin{gather*}
 \widetilde{\mathcal{F}}^{(i)}\big(\widetilde{X}_1,\widetilde{X}_2\big) = \det (A_2)^2\mathcal{F}^{(i)}\big(A_2^{-1}\big(\widetilde{X}_1,\widetilde{X}_2\big) \big).
\end{gather*}
From this we can deduce the following lemma.
\begin{Lemma}\label{prop2}
Given $\mathcal{F}\in \mathbb {C}^{[3]}[x_1,x_2,x_3]$ we can find an explicit matrix $A\in G$ such that
for
\begin{gather*}
\widetilde{\mathcal{F}}\big(\widetilde{\mathcal{L}}_1,\widetilde{\mathcal{L}}_2,\widetilde{\mathcal{H}}\big) =
\det (A_2)^2\mathcal{F}\big(A^{-1}\big(\widetilde{\mathcal{L}}_1,\widetilde{\mathcal{L}}_2,\widetilde{\mathcal{H}}\big) \big)
\end{gather*}
we have $\widetilde{\mathcal{F}}^{(3)}(\widetilde{X}_1,\widetilde{X}_2)=C_I(X_1,X_2)$, where $C_I$ equal to exactly one of the following
\begin{gather*}
0,\qquad C_1(X_1,X_2)=X_1X_2(X_1+X_2),\qquad C_2(X_1,X_2)={X}_1^2 X_2,\qquad C_3(X_1,X_2)={X}_1^3.
\end{gather*}
\end{Lemma}

\begin{Proposition}
\begin{gather*}
\operatorname{Stab}_{G}(C_1+O(\mathcal{H}))=\left\{\left(\begin{matrix}
A &v\\
0 & c
\end{matrix}\right)|\,A\in\Omega(C_1), \,v\in \mathbb{C}^2, \,c\in \mathbb{C}^* \right\},\\
\end{gather*}
where
\begin{gather*}
\Omega(C_1)=\left\{\left(\begin{matrix}
0 & 1\\
1 & 0
\end{matrix}\right), \left(\begin{matrix}
0 & 1\\
-1 & -1
\end{matrix}\right),\left(\begin{matrix}
1 & 0\\
0 & 1
\end{matrix}\right) \right\} \\
\hphantom{\Omega(C_1)=}{}
\coprod\left\{\left(\begin{matrix}
-1 & -1\\
0 & -1
\end{matrix}\right), \left(\begin{matrix}
1 & 0\\
-1 & -1
\end{matrix}\right),\left(\begin{matrix}
-1 & -1\\
1 & 0
\end{matrix}\right) \right\}, \\
\operatorname{Stab}_{G}(C_2+O(\mathcal{H}))=\left\{\left(\begin{matrix}
a&0 &v_1\\
0&1 & v_2\\
0&0&c
\end{matrix}\right)|\, v_1,v_2\in \mathbb{C}, \,a,c\in \mathbb{C}^* \right\},\\
\operatorname{Stab}_{G}(C_3+O(\mathcal{H}))=\left\{\left(\begin{matrix}
d^2&0 &v_1\\
b&d & v_2\\
0&0&c
\end{matrix}\right)| \,b,v_1,v_2\in \mathbb{C}, \,c,d\in \mathbb{C}^* \right\}.
\end{gather*}
\end{Proposition}

\subsection{First case: three distinct roots}
Suppose that
\begin{gather*}
{\mathcal{F}}^{(3)}(X_1,X_2)+\mathcal{H}{\mathcal{F}}^{(2)}(X_1,X_2)= C_1(X_1,X_2)+ { \mathcal{H}}\big(c_5 X_1^2+c_6X_2^2+c_7 X_1X_2\big).
\end{gather*}

Acting with
\begin{gather*}
 A=\left(\begin{matrix}
1&0&-c_6\\
0&1&-c_5\\
0&0&1
\end{matrix}\right)^{-1}\in \operatorname{Stab}_{G}(C_1(X_1,X_2))
\end{gather*}
we get
\begin{gather*}
 C_1(X_1,X_2)+ { \mathcal{H}}\big(c_5 X_1^2+c_6X_2^2+c_7 X_1X_2\big)\\
 \qquad{} \longmapsto
 C_1(X_1,X_2)+{ \mathcal{H}}\big(c'_7 X_1X_2\big)+\mathcal{H}^2(c'_8X_1+c'_9X_2)+c_{10}\mathcal{H}^3
\end{gather*}
for some $c'_7$, $c'_8$, $c'_9$, $c'_{10}$, hence we can assume that the
\begin{gather*}
 {\mathcal{F}}^{(3)}(X_1,X_2)+\mathcal{H}{\mathcal{F}}^{(2)}(X_1,X_2)=C_1(X_1,X_2)+c_7 { \mathcal{H}} X_1X_2+O\big(\mathcal{H}^2\big)
\end{gather*}
using a matrix of the form
\begin{gather*}
 A=\left(\begin{matrix}
1&0&0\\
0&1&0\\
0&0&r
\end{matrix}\right)
\end{gather*}
we can further assume that $c_7\in \{0,1\}$. For the case of $c_7=0$ we obtain the following proposition:
\begin{Proposition}
The stabilizer of the form
\begin{gather*}
{\mathcal{F}}^{(3)}(X_1,X_2)+\mathcal{H}{\mathcal{F}}^{(2)}(X_1,X_2)+O\big(\mathcal{H}^2\big)= C_1(X_1,X_2)+O\big(\mathcal{H}^2\big)
\end{gather*}
is given by
\begin{gather*}
\operatorname{Stab}_{G}\big(C_1(X_1,X_2)+O\big(\mathcal{H}^2\big)\big)=
\left\{\left(\begin{matrix}
A &0_2\\
0 & c
\end{matrix}\right)|\,A\in \Omega (C_1), \,0_2=0\in \mathbb{C}^2, \, c\in \mathbb{C}^* \right\}.
\end{gather*}
\end{Proposition}

\begin{proof}
It is easy to see that
\begin{gather*}
 \operatorname{Stab}_{G}\big(C_1(X_1,X_2)+O\big(\mathcal{H}^2\big)\big)\supseteq \left\{\left(\begin{matrix}
A &0_2\\
0 & c
\end{matrix}\right)|\,A\in \Omega (C_1), \,0_2=0\in \mathbb{C}^2, \,c\in \mathbb{C}^* \right\}.
\end{gather*}
For inclusion in the other direction, let $M\in \operatorname{Stab}_{G}(C_1(X_1,X_2)+O(\mathcal{H}^2))$ then obviously~$M_2$ has to preserve
$C_1(X_1,X_2)$, i.e., $M_2 \in \Omega(C_1)$. Hence the matrix
\begin{gather*}
 \left(\begin{matrix}
\big(M^{-1}\big)_{1,1} & \big(M^{-1}\big)_{1,2} & 0\\
\big(M^{-1}\big)_{2,1} & \big(M^{-1}\big)_{2,2} &0\\
0&0&1
\end{matrix}\right)M= \left(\begin{matrix}
1 & 0 & M_{1,3}\\
0 & 1 &M_{2,3}\\
0&0&M_{3,3}
\end{matrix}\right)
\end{gather*}
as a product of two matrices in the stabilizer $\operatorname{Stab}_{G}(C_1(X_1,X_2)+O(\mathcal{H}^2))$ is also in the stabilizer. The result of the action of this matrix on $C_1(X_1,X_2)+O(\mathcal{H}^2)$ forces $ M_{1,3}= M_{2,3}=0$.
\end{proof}

For the case of $c_7=1$ we obtain the following proposition:
\begin{Proposition}
The stabilizer of the form
\begin{gather*}
{\mathcal{F}}^{(3)}(X_1,X_2)+\mathcal{H}{\mathcal{F}}^{(2)}(X_1,X_2)+O\big(\mathcal{H}^2\big)= C_1(X_1,X_2)+\mathcal{H}X_1X_2+O\big(\mathcal{H}^2\big)
\end{gather*}
is given by
\begin{gather*}
 \operatorname{Stab}_{G}\big(C_1(X_1,X_2)+\mathcal{H}X_1X_2+O\big(\mathcal{H}^2\big)\big)= \left\{\left(\begin{matrix}
1 & 0&0\\
0 & 1&0\\
0&0&1
\end{matrix}\right) \right\}.
\end{gather*}
\end{Proposition}
\begin{proof}
Following the same reasoning as in the previous proof we easily see that for $M\in \operatorname{Stab}_{G}(C_1(X_1,X_2)
+\mathcal{H}X_1X_2+O(\mathcal{H}^2))$ we must have $M_2=\left(\begin{matrix}
1 & 0\\
0 & 1
\end{matrix}\right) $ and then by direct calculation the rest of the proof follows.
\end{proof}

\subsubsection[$\mathcal{F}^{(3)}(X_1,X_2)=C_1(X_1,X_2)$ and $c_7=0$]{$\boldsymbol{\mathcal{F}^{(3)}(X_1,X_2)=C_1(X_1,X_2)}$ and $\boldsymbol{c_7=0}$}

Suppose that
\begin{gather*}
\mathcal{R}^2= {\mathcal{F}}^{(3)}(X_1,X_2)+\mathcal{H}{\mathcal{F}}^{(2)}(X_1,X_2)
+\mathcal{H}^2{\mathcal{F}}^{(1)}(X_1,X_2)+c_{10}\mathcal{H}^3 \\
\hphantom{\mathcal{R}^2}{}=C_1(X_1,X_2)+ { \mathcal{H}}^2 (c_8 X_1+c_9X_2 )+c_{10}\mathcal{H}^3.
\end{gather*}

Acting with $A=\left(\begin{matrix}
\alpha & \beta &0 \\
\gamma & \delta &0\\
0&0&c
\end{matrix} \right)^{-1}\in \operatorname{Stab}_{G}(C_1(X_1,X_2)+O(\mathcal{H}^2))$ on $\mathcal{R}^2$ we will have
\begin{gather*}
\mathcal{R}^2=C_1(X_1,X_2)+ { \mathcal{H}}^2\left(c_8 X_1+c_9X_2\right)+c_{10}\mathcal{H}^3 \\
\hphantom{\mathcal{R}^2}{} \longmapsto C_1(X_1,X_2)+\mathcal{H}^2(c'_8X_1+c'_9X_2)+c'_{10}\mathcal{H}^3,
\end{gather*}
where $c'_8=c^2(\alpha c_8+ \gamma c_9)$, $c'_9=c^2(\beta c_8+\delta c_9)$, $c'_{10}=c^3 c_{10}$, and
 $\left(\begin{matrix}
\alpha & \beta \\
\gamma & \delta
\end{matrix} \right)\in \Omega(C_1)$. Note that the size of the group $\Omega(C_1)$ is~6.

We now describe an algorithm for choosing a canonical form in this case. If $c_{10}\neq 0$ then acting with $\left(\begin{matrix}
1&0 & 0\\
0 & 1 &0\\
0&0&(c_{10})^{\frac{1}{3}}
\end{matrix}\right)$ we obtain $c_{10}'= 1$. Writing
$\left(\begin{matrix}
c_8' \\
c_9'
\end{matrix} \right)=\left(\begin{matrix}
re^{i\theta}\\
\rho e^{i\phi}
\end{matrix} \right)$
with $r, \rho \geq 0$ and $\theta, \phi \in [0,2\pi)$ we choose as our canonical form the expression for $c_8$ and $c_9$ according to the following rules (note that the order is important) f\/irst make $r$ is maximal, then $\theta $ minimal, then $\rho$ minimal, and f\/inally $\phi $ minimal. If $c_{10}=0$ then again we act with $A=\left(\begin{matrix}
\alpha & \beta &0 \\
\gamma & \delta &0\\
0&0&1
\end{matrix} \right)^{-1} $ with $\left(\begin{matrix}
\alpha & \beta \\
\gamma & \delta
\end{matrix} \right)\in \Omega(C_1)$ and choose $c_8$ and $c_9$ as above and then we can act with a matrix of the form
$\left(\begin{matrix}
1&0 & 0\\
0 & 1 &0\\
0&0&c
\end{matrix}\right)$ to normalize $c_8$ to zero or one.

\subsubsection[$\mathcal{F}^{(3)}(X_1,X_2)=C_1(X_1,X_2)$ and $c_7=1$]{$\boldsymbol{\mathcal{F}^{(3)}(X_1,X_2)=C_1(X_1,X_2)}$ and $\boldsymbol{c_7=1}$}
Suppose that
\begin{gather}
\mathcal{R}^2={\mathcal{F}}^{(3)}(X_1,X_2)+\mathcal{H}{\mathcal{F}}^{(2)}(X_1,X_2)+\mathcal{H}^2{\mathcal{F}}^{(1)}(X_1,X_2)+c_{10}\mathcal{H}^3 \nonumber\\
\hphantom{\mathcal{R}^2}{} =C_1(X_1,X_2)+\mathcal{H}X_1X_2+ { \mathcal{H}}^2\left(c_8 X_1+c_9X_2\right)+c_{10}\mathcal{H}^3.\label{e74}
\end{gather}
Since\begin{gather*}
\operatorname{Stab}_{G}\big(C_1(X_1,X_2)+\mathcal{H}X_1X_2+O\big(\mathcal{H}^2\big)\big)= \left\{\left(\begin{matrix}
1 & 0&0\\
0 & 1&0\\
0&0&1
\end{matrix}\right) \right\}
\end{gather*}
then for any $c_8,c_9,c_{10}\in \mathbb{C}$ equation (\ref{e74}) def\/ines a canonical form.

\subsection{Second case: a double root}

Suppose that
\begin{gather*}
{\mathcal{F}}^{(3)}(X_1,X_2)+\mathcal{H}{\mathcal{F}}^{(2)}(X_1,X_3)= C_2(X_1,X_2)+ { \mathcal{H}}\big(c_5 X_1^2+c_6X_2^2+c_7 X_1X_2\big).
\end{gather*}
Acting with
\begin{gather*}
 A=\left(\begin{matrix}
1&0&-\frac{1}{2}c_7\\
0&1&-c_5\\
0&0&1
\end{matrix}\right)^{-1}\in \operatorname{Stab}_{G}(C_2(X_1,X_2))
\end{gather*} on $\mathcal{R}^2$ we have
\begin{gather*}
C_2(X_1,X_2)+ { \mathcal{H}}\big(c_5 X_1^2+c_6X_2^2+c_7 X_1X_2\big) \\ \nonumber
\qquad{} \longmapsto C_2(X_1,X_2)+{ \mathcal{H}}\big(c'_6X_2^2\big)+\mathcal{H}^2(c'_8X_1+c'_9X_2)+c'_{10}\mathcal{H}^3
\end{gather*}
for some $c'_6$, $c'_8$, $c'_9$, $c'_{10}$. Hence we can assume that the
\begin{gather*}
 {\mathcal{F}}^{(3)}(X_1,X_2)+\mathcal{H}{\mathcal{F}}^{(2)}(X_1,X_3)=C_2(X_1,X_2)+c_6 { \mathcal{H}} X^2_2+O\big(\mathcal{H}^2\big)
\end{gather*}
using a matrix of the form $A=\left(\begin{matrix}
1&0&0\\
0&1&0\\
0&0&r
\end{matrix}\right)$
we can further assume that $c_6\in \{0,1\}$. For the case of $c_6=0$ we obtain the following proposition:
\begin{Proposition}
The stabilizer of the form
\begin{gather*}
{\mathcal{F}}^{(3)}(X_1,X_2)+\mathcal{H}{\mathcal{F}}^{(2)}(X_1,X_2)+O\big(\mathcal{H}^2\big)= C_2(X_1,X_2)+O\big(\mathcal{H}^2\big)
\end{gather*}
is given by
\begin{gather*}
\operatorname{Stab}_{G}\big(C_2(X_1,X_2)+O\big(\mathcal{H}^2\big)\big)=\left\{\left(\begin{matrix}
a &0&0\\
0&1&0\\
0&0 & c
\end{matrix}\right)|\,a,c\in \mathbb{C}^* \right\}.
\end{gather*}
\end{Proposition}

For the case of $c_6=1$ we obtain the following proposition:
\begin{Proposition}
The stabilizer of the form
\begin{gather*}
{\mathcal{F}}^{(3)}(X_1,X_2)+\mathcal{H}{\mathcal{F}}^{(2)}(X_1,X_2)+O\big(\mathcal{H}^2\big)= C_2(X_1,X_2)+\mathcal{H}X^2_2+O\big(\mathcal{H}^2\big)
\end{gather*}
is given by
\begin{gather*}
\operatorname{Stab}_{G}\big(C_2(X_1,X_2)+\mathcal{H}X^2_2+O\big(\mathcal{H}^2\big)\big)=\left\{\left(\begin{matrix}
r &0 &0\\
0 & 1&0\\
0&0&r^2
\end{matrix}\right)|\,r\in \mathbb{C}^* \right\}.
\end{gather*}
\end{Proposition}

\subsubsection[$\mathcal{F}^{(3)}(X_1,X_2)=C_2(X_1,X_2)$ and $c_6=0$]{$\boldsymbol{\mathcal{F}^{(3)}(X_1,X_2)=C_2(X_1,X_2)}$ and $\boldsymbol{c_6=0}$}

Suppose that
\begin{gather*}
\mathcal{R}^2={\mathcal{F}}^{(3)}(X_1,X_2)+\mathcal{H}{\mathcal{F}}^{(2)}(X_1,X_2)+\mathcal{H}^2{\mathcal{F}}^{(1)}(X_1,X_2)+c_{10}\mathcal{H}^3 \\ \hphantom{\mathcal{R}^2}{} = C_2(X_1,X_2)+ { \mathcal{H}}^2 (c_8 X_1+c_9X_2 )+c_{10}\mathcal{H}^3.
\end{gather*}

Acting with $A=\left(\begin{matrix}
a & 0 &0 \\
0 & 1 &0\\
0&0&c
\end{matrix} \right)^{-1} \in \operatorname{Stab}_{G}(C_2(X_1,X_2)+O(\mathcal{H}^2))$ on $\mathcal{R}^2$ we have
\begin{gather*}
 \mathcal{R}^2=C_2(X_1,X_2)+ { \mathcal{H}}^2 (c_8 X_1+c_9X_2 )+c_{10}\mathcal{H}^3 \\ \nonumber
\hphantom{\mathcal{R}^2}{} \longmapsto C_2(X_1,X_2)+\mathcal{H}^2(c'_8X_1+c'_9X_2)+c'_{10}\mathcal{H}^3,
\end{gather*}
where $c'_8=c^2a^{-1}c_8$, $c'_9=c^2a^{-2} c_9$, $c'_{10}=c^3 a^{-2}c_{10}$. For the canonical form, we normalize the f\/irst two non zero coef\/f\/icients from $c_8$, $c_9$, $c_{10}$ to be equal to 1.

\subsubsection[$\mathcal{F}^{(3)}(X_1,X_2)=C_2(X_1,X_2)$ and $c_6=1$]{$\boldsymbol{\mathcal{F}^{(3)}(X_1,X_2)=C_2(X_1,X_2)}$ and $\boldsymbol{c_6=1}$}

Suppose that
\begin{gather*}
 \mathcal{R}^2=C_2(X_1,X_2)+\mathcal{H}X^2_2+\mathcal{H}^2(c_8X_1+c_9X_2)+c_{10}\mathcal{H}^3.
\end{gather*}
Acting with $A=\left(\begin{matrix}
r& 0 &0 \\
0 & 1 &0\\
0&0&r^2
\end{matrix} \right)^{-1} \in \operatorname{Stab}_{G}(C_2(X_1,X_2)+\mathcal{H}X^2_2+O(\mathcal{H}^2))$ on $\mathcal{R}^2$ we have
\begin{gather*}
\mathcal{R}^2=C_2(X_1,X_2)+\mathcal{H}X^2_2++ { \mathcal{H}}^2 (c_8 X_1+c_9X_2 )+c_{10}\mathcal{H}^3 \\
\hphantom{\mathcal{R}^2}{} \longmapsto C_2(X_1,X_2)+\mathcal{H}X^2_2+\mathcal{H}^2(c'_8X_1+c'_9X_2)+c'_{10}\mathcal{H}^3,
\end{gather*}
where $c'_8=r^3c_8$, $c'_9=r^2 c_9$, $c'_{10}=r^4c_{10}$. We def\/ine the canonical form to be with $c_k=1$, where $k$ is the smallest integer among $\{8,9,10\}$ such that $c_k\neq 0$.

\subsection{Third case: a triple root}

Suppose that
\begin{gather*}
{\mathcal{F}}^{(3)}(X_1,X_2)+\mathcal{H}{\mathcal{F}}^{(2)}(X_1,X_2)= C_3(X_1,X_2)+ { \mathcal{H}}\big(c_5 X_1^2+c_6X_2^2+c_7 X_1X_2\big).
\end{gather*}
Acting with
\begin{gather*}
 A=\left(\begin{matrix}
1&0&-\frac{1}{3}c_5\\
0&1&0\\
0&0&1
\end{matrix}\right)^{-1}\in \operatorname{Stab}_{G}(C_3(X_1,X_2))
\end{gather*}
on $\mathcal{R}^2$ we have
\begin{gather*}
 C_3(X_1,X_2)+ { \mathcal{H}}\big(c_5 X_1^2+c_6X_2^2+c_7 X_1X_2\big)\\ \nonumber
\qquad{} \longmapsto C_3(X_1,X_2)+{ \mathcal{H}}\big(c'_6X_2^2+c'_7 X_1X_2\big)+\mathcal{H}^2(c'_8X_1+c'_9X_2)+c'_{10}\mathcal{H}^3
\end{gather*}
for some $c'_6$, $c'_7$, $c'_8$, $c'_9$, $c'_{10}$. Hence we can assume that the
\begin{gather*}
 {\mathcal{F}}^{(3)}(X_1,X_2)+\mathcal{H}{\mathcal{F}}^{(2)}(X_1,X_3)=C_3(X_1,X_2)+c_6 { \mathcal{H}} X^2_2+ c_7 { \mathcal{H}} X_1X_2
\end{gather*}
using a matrix of the form $A=\left(\begin{matrix}
d^2&0&0\\
0&d &0\\
0&0&r
\end{matrix}\right)$
we can further assume that $c_6, c_7\in \{0,1\}$. For the case of $c_6=c_7=0$ we obtain the following proposition:
\begin{Proposition}
The stabilizer of the form
\begin{gather*}
{\mathcal{F}}^{(3)}(X_1,X_2)+\mathcal{H}{\mathcal{F}}^{(2)}(X_1,X_3)+O\big(\mathcal{H}^2\big)= C_3(X_1,X_2)+O\big(\mathcal{H}^2\big)
\end{gather*}
is given by
\begin{gather*}
 \operatorname{Stab}_{G}\big(C_3(X_1,X_2)+O\big(\mathcal{H}^2\big)\big)=\left\{\left(\begin{matrix}
d^2 &0&0\\
\gamma & d &b\\
0&0&c
\end{matrix}\right)|\, b,\gamma \in \mathbb{C},\, d,c\in \mathbb{C}^* \right\}.
\end{gather*}
\end{Proposition}
For the case of $c_6=0$, $c_7=1$ we obtain the following proposition:
\begin{Proposition}
The stabilizer of the form
\begin{gather*}
{\mathcal{F}}^{(3)}(X_1,X_2)+\mathcal{H}{\mathcal{F}}^{(2)}(X_1,X_3)+O\big(\mathcal{H}^2\big)= C_3(X_1,X_2)+\mathcal{H}X_1X_2+O\big(\mathcal{H}^2\big)
\end{gather*}
is given by
\begin{gather*}
 \operatorname{Stab}_{G}(C_3(X_1,X_2)+\mathcal{H}X_1X_2+O\big(\mathcal{H}^2)\big)=\left\{\left(\begin{matrix}
d^2 &0&a\\
-\frac{3a}{d} & d &b\\
0&0&d^3
\end{matrix}\right)|\, a, b \in \mathbb{C}, \, d \in \mathbb{C}^* \right\}.
\end{gather*}
\end{Proposition}

For the case of $c_6=1$, $c_7=0$ we obtain the following proposition:
\begin{Proposition}
The stabilizer of the form
\begin{gather*}
{\mathcal{F}}^{(3)}(X_1,X_2)+\mathcal{H}{\mathcal{F}}^{(2)}(X_1,X_3)+O\big(\mathcal{H}^2\big)= C_3(X_1,X_2)+\mathcal{H}X^2_2+O\big(\mathcal{H}^2\big)
\end{gather*}
is given by
\begin{gather*}
\operatorname{Stab}_{G}\big(C_3(X_1,X_2)+\mathcal{H}X^2_2+O\big(\mathcal{H}^2\big)\big)=\left\{\left(\begin{matrix}
d^2 &0&0\\
0 & d &b\\
0&0&d^{4}
\end{matrix}\right)|\, b \in \mathbb{C}, \,d \in \mathbb{C}^* \right\}.
\end{gather*}
\end{Proposition}

For the case of $c_6=1$, $c_7=1$ we obtain the following proposition:
\begin{Proposition}
The stabilizer of the form
\begin{gather*}
{\mathcal{F}}^{(3)}(X_1,X_2)+\mathcal{H}{\mathcal{F}}^{(2)}(X_1,X_3)= C_3(X_1,X_2)+\mathcal{H}X^2_2+\mathcal{H}X_1X_2+O\big(\mathcal{H}^2\big)
\end{gather*}
is given by
\begin{gather*}
\operatorname{Stab}_{G}\big(C_3(X_1,X_2)+\mathcal{H}X_1X_2+\mathcal{H}X^2_2+O\big(\mathcal{H}^2\big)\big)\\
\qquad{}=\left\{\left(\begin{matrix}
d^2 &0&\frac{d^2}{12}(d^2-1)\\
\frac{1}{2}d(1-d) & d &b\\
0&0&d^4
\end{matrix}\right)| \,b \in \mathbb{C}, \,d \in \mathbb{C}^* \right\}.
\end{gather*}
\end{Proposition}

 \subsubsection[$\mathcal{F}^{(3)}(X_1,X_2)=C_3(X_1,X_2)$, $c_6=0$, and $c_7=0$]{$\boldsymbol{\mathcal{F}^{(3)}(X_1,X_2)=C_3(X_1,X_2)}$, $\boldsymbol{c_6=0}$, and $\boldsymbol{c_7=0}$}
Suppose that
\begin{gather*}
\mathcal{R}^2={\mathcal{F}}^{(3)}(X_1,X_2)+\mathcal{H}{\mathcal{F}}^{(2)}(X_1,X_2)+\mathcal{H}^2{\mathcal{F}}^{(1)}(X_1,X_2)+c_{10}\mathcal{H}^3 \\
\hphantom{\mathcal{R}^2}{}= C_3(X_1,X_2)+ { \mathcal{H}}^2\left(c_8 X_1+c_9X_2\right)+c_{10}\mathcal{H}^3.
\end{gather*}

Acting with $A=\left(\begin{matrix}
d^2& 0 &0 \\
\gamma & d &b\\
0&0&c
\end{matrix} \right)^{-1} \in \operatorname{Stab}_{G}(C_3(X_1,X_2) +O(\mathcal{H}^2))$ on $\mathcal{R}^2$ we have
\begin{gather*}
 \mathcal{R}^2=C_3(X_1,X_2)+ { \mathcal{H}}^2 (c_8 X_1+c_9X_2 )+c_{10}\mathcal{H}^3 \\ \nonumber
\hphantom{\mathcal{R}^2}{} \longmapsto C_3(X_1,X_2)+\mathcal{H}^2(c'_8X_1+c'_9X_2)+c'_{10}\mathcal{H}^3,
\end{gather*}
where $c'_8=c^2 (d^{-4}c_8+d^{-6}\gamma c_9 )$, $c'_9=c^2d^{-5} c_9$, $c'_{10}=d^{-6} (c^2bc_9+c^3c_{10} )$. If $c_9=0$ and $c_8\neq 0$ we def\/ine the canonical form to be with $c_8=1$ and $c_{10}=re^{i\theta}$ with $r\geq 0$ and $\theta \in [0,\pi)$. If $c_9=0$ and $c_8=0$ we def\/ine the canonical form to be with $c_{10}\in \{0,1\}$. If $c_9\neq 0$ then the canonical form is given by $\mathcal{R}^2=C_3(X_1,X_2)+ { \mathcal{H}}^2 X_2$.

 \subsubsection[$\mathcal{F}^{(3)}(X_1,X_2)=C_3(X_1,X_2)$, $c_6=0$, and $c_7=1$]{$\boldsymbol{\mathcal{F}^{(3)}(X_1,X_2)=C_3(X_1,X_2)}$, $\boldsymbol{c_6=0}$, and $\boldsymbol{c_7=1}$}
Suppose that
\begin{gather*}
\mathcal{R}^2={\mathcal{F}}^{(3)}(X_1,X_2)+\mathcal{H}{\mathcal{F}}^{(2)}(X_1,X_2)+\mathcal{H}^2{\mathcal{F}}^{(1)}(X_1,X_2)+c_{10}\mathcal{H}^3 \\ \hphantom{\mathcal{R}^2}{}= C_3(X_1,X_2)+\mathcal{H}X_1X_2+ { \mathcal{H}}^2 (c_8 X_1+c_9X_2 )+c_{10}\mathcal{H}^3.
\end{gather*}

Acting with $A=\left(\begin{matrix}
d^2 &0&a\\
-\frac{3a}{d} & d &b\\
0&0&d^3
\end{matrix}\right)^{-1} \in \operatorname{Stab}_{G}(C_3(X_1,X_2)+\mathcal{H}X_1X_2 +O(\mathcal{H}^2))$ on $\mathcal{R}^2$ we have
\begin{gather*}
\mathcal{R}^2=C_3(X_1,X_2)+\mathcal{H}X_1X_2+ { \mathcal{H}}^2 (c_8 X_1+c_9X_2 )+c_{10}\mathcal{H}^3 \\
 \hphantom{\mathcal{R}^2}{} \longmapsto C_3(X_1,X_2)+\mathcal{H}X_1X_2+\mathcal{H}^2(c'_8X_1+c'_9X_2)+c'_{10}\mathcal{H}^3,
\end{gather*}
where $c'_8= \frac{b}{d}+d^2c_8-3\frac{a}{d}c_9$, $c'_9=\frac{a}{d^2}+c_9d$, $c'_{10}=\frac{a^3}{d^6}+\frac{ab}{d^3}+ac_8+bc_9+d^3c_{10}$. Hence we can always arrange that $c_8=c_9=0$ and $c_{10}\in \{0,1\}$ and this will be the canonical form in this case.

 \subsubsection[$\mathcal{F}^{(3)}(X_1,X_2)=C_3(X_1,X_2)$, $c_6=1$, and $c_7=0$]{$\boldsymbol{\mathcal{F}^{(3)}(X_1,X_2)=C_3(X_1,X_2)}$, $\boldsymbol{c_6=1}$, and $\boldsymbol{c_7=0}$}
Suppose that
\begin{gather*}
\mathcal{R}^2={\mathcal{F}}^{(3)}(X_1,X_2)+\mathcal{H}{\mathcal{F}}^{(2)}(X_1,X_2)+\mathcal{H}^2{\mathcal{F}}^{(1)}(X_1,X_2)+c_{10}\mathcal{H}^3\\ \hphantom{\mathcal{R}^2}{}=C_3(X_1,X_2)+\mathcal{H}X_2^2+ { \mathcal{H}}^2 (c_8 X_1+c_9X_2 )+c_{10}\mathcal{H}^3.
\end{gather*}

Acting with $A=\left(\begin{matrix}
d^2 &0&0\\
0 & d &b\\
0&0&d^4
\end{matrix}\right)^{-1} \in \operatorname{Stab}_{G}(C_3(X_1,X_2)+\mathcal{H}X^2_2 +O(\mathcal{H}^2))$ on $\mathcal{R}^2$ we have
\begin{gather*}
 \mathcal{R}^2=C_3(X_1,X_2)+\mathcal{H}X^2_2+ { \mathcal{H}}^2 (c_8 X_1+c_9X_2 )+c_{10}\mathcal{H}^3 \\
\hphantom{\mathcal{R}^2}{} \longmapsto C_3(X_1,X_2)+\mathcal{H}X^2_2+\mathcal{H}^2(c'_8X_1+c'_9X_2)+c'_{10}\mathcal{H}^3,
\end{gather*}
where $c'_8= d^4c_8$, $c'_9=2\frac{b}{d}+c_9d^3$, $c'_{10}=\frac{b^2}{d^2}+d^2bc_9+d^6c_{10}$. Hence we can always arrange that $c_9=0$ and either $c_8=0$ and $c_{10}\in \{0,1\}$ or $c_8=1$ and $c_{10}=re^{i\theta}$ with $r\geq 0$ and $\theta\in [0,\frac{\pi}{2})$.

 \subsubsection[$\mathcal{F}^{(3)}(X_1,X_2)=C_3(X_1,X_2)$, $c_6=1$, and $c_7=1$]{$\boldsymbol{\mathcal{F}^{(3)}(X_1,X_2)=C_3(X_1,X_2)}$, $\boldsymbol{c_6=1}$, and $\boldsymbol{c_7=1}$}
Suppose that
\begin{gather*}
\mathcal{R}^2={\mathcal{F}}^{(3)}(X_1,X_2)+\mathcal{H}{\mathcal{F}}^{(2)}(X_1,X_2)+\mathcal{H}^2{\mathcal{F}}^{(1)}(X_1,X_2)+c_{10}\mathcal{H}^3 \\ \hphantom{\mathcal{R}^2}{}=C_3(X_1,X_2)+\mathcal{H}X_1X_2+\mathcal{H}X_2^2+ { \mathcal{H}}^2 (c_8 X_1+c_9X_2 )+c_{10}\mathcal{H}^3.
\end{gather*}

Acting with $A=\left(\begin{matrix}
d^2 &0&\frac{d^2}{12}(d^2-1)\\
\frac{1}{2}d(1-d) & d &b\\
0&0&d^4
\end{matrix}\right)^{-1} \in \operatorname{Stab}_{G}(C_3(X_1,X_2)+\mathcal{H}X_1X_2+\mathcal{H}X^2_2 +O(\mathcal{H}^2))$ on $\mathcal{R}^2$ we have
\begin{gather*}
\mathcal{R}^2=C_3(X_1,X_2)+\mathcal{H}X_1X_2+\mathcal{H}X^2_2+ { \mathcal{H}}^2 (c_8 X_1+c_9X_2 )+c_{10}\mathcal{H}^3 \\
\hphantom{\mathcal{R}^2}{} \longmapsto C_3(X_1,X_2)+\mathcal{H}X_1X_2+\mathcal{H}X^2_2+\mathcal{H}^2(c'_8X_1+c'_9X_2)+c'_{10}\mathcal{H}^3,
\end{gather*}
where $c'_8= \frac{b}{d} -\frac{1}{48}(d^2-1)(d-1)^2+d^4c_8+\frac{1}{2}d^3(1-d)c_9$, $c'_9=\frac{1}{12}d^3(d^2-1)+2\frac{b}{d}+c_9d^3$, $c'_{10}=\frac{1}{12^3}(d^2-1)^3+\frac{1}{12}(d^2-1)b+\frac{b^2}{d^2}+\frac{1}{12}d^4(d^2-1)c_8+ d^2bc_9+d^6c_{10}$. Hence we can assume that $c_9=c_8=0$ and the canonical form is given by
 \begin{gather*}
\mathcal{R}^2= X_1^3+\mathcal{H}X_1X_2+\mathcal{H}X^2_2+c_{10}\mathcal{H}^3
\end{gather*}
with $c_{10}\in \mathbb{C}$.

\subsection[Fourth case: $\widetilde{\mathcal{F}}^{(3)}=0$]{Fourth case: $\boldsymbol{\widetilde{\mathcal{F}}^{(3)}=0}$}
A similar (but simpler) calculation to the one that was done in the previous section leads to the possibilities for the canonical forms for $\mathcal{F}\in \mathbb{C}^{[3]}[x_1,x_2,x_3]$ with a~vanishing~${\mathcal{F}}^{(3)}$. For example it easy to show the following lemma.
\begin{Lemma}
Given $\mathcal{F}\in \mathbb{C}^{[3]}[x_1,x_2,x_3]$ with a vanishing~${\mathcal{F}}^{(3)}$ we can find an explicit matrix $A\in G$ such that the ${\mathcal{F}}^{(2)}$ part of $A\cdot \mathcal{F} $ is equal to exactly one of the following three cases: $X_1^2$, ${X}_1 X_2$, $0$.
\end{Lemma}

\begin{table}[t]\footnotesize\centering
\caption{List of canonical forms of $\mathcal{R}^2$ for the nondegenerate free quadratic algebras.}\label{table1}\vspace{1mm}

\begin{tabular}{|l|l|l|}
\hline \multicolumn{2}{|c|}{Canonical forms of $\mathcal{R}^2$ for the nondegenerate free quadratic algebras\tsep{2pt}} &
\\ \hline& $\mathcal{R}^2$\tsep{2pt} & domain of parameters
\\ \hline 1a& $X_1X_2(X_1+X_2)+c_8X_1\mathcal{H}^2+c_9X_2\mathcal{H}^2+\mathcal{H}^3$\tsep{2pt} & $c_8,c_9\in \mathbb{C}$, see remark below
\\ \hline 1b& $X_1X_2(X_1+X_2)+X_1\mathcal{H}^2+c_9X_2\mathcal{H}^2$\tsep{2pt} & $c_9\in \mathbb{C}$, see remark below
\\ \hline 1c& $X_1X_2(X_1+X_2)$\tsep{2pt} & \\ \hline
1d& $X_1X_2(X_1+X_2)+\mathcal{H}X_1X_2+c_8X_1\mathcal{H}^2+c_9X_2\mathcal{H}^2+c_{10}\mathcal{H}^3$\tsep{2pt} & $c_8,c_9,c_{10}\in \mathbb{C}$\\ \hline
2a&$X_1^2X_2+X_1\mathcal{H}^2+X_2\mathcal{H}^2+c_{10}\mathcal{H}^3$\tsep{2pt} & $c_{10}\in \mathbb{C}$\\ \hline
2b&$X_1^2X_2+c_9X_2\mathcal{H}^2+c_{10}\mathcal{H}^3$\tsep{2pt} & $c_9,c_{10}\in \{0,1\}$\\ \hline
2c&$X_1^2X_2+\mathcal{H}X_2^2 +X_1\mathcal{H}^2+c_9X_2\mathcal{H}^2+c_{10}\mathcal{H}^3$\tsep{2pt} & $c_9,c_{10}\in \mathbb{C}$\\ \hline
2d&$X_1^2X_2+\mathcal{H}X_2^2 +X_2\mathcal{H}^2+c_{10}\mathcal{H}^3$\tsep{2pt} & $c_{10}\in \mathbb{C}$\\ \hline
2e&$X_1^2X_2+\mathcal{H}X_2^2 +c_{10}\mathcal{H}^3$\tsep{2pt} & $c_{10}\in \{0,1\}$\\ \hline
3a&$X_1^3+X_1\mathcal{H}^2 +c_{10}\mathcal{H}^3$\tsep{2pt} & $c_{10}\in \mathbb{C}$\\ \hline
3b&$X_1^3+\mathcal{H}^3$\tsep{2pt} &\\ \hline
3c&$X_1^3+X_2\mathcal{H}^2 $\tsep{2pt} & \\ \hline
3d&$X_1^3+\mathcal{H}X_1X_2 +c_{10}\mathcal{H}^3$\tsep{2pt} & $c_{10}\in \{0,1\}$\\ \hline
3e&$X_1^3+\mathcal{H}X_2^2 +c_{10}\mathcal{H}^3$\tsep{2pt} & $c_{10}\in \{0,1\}$\\ \hline
3f&$X_1^3+\mathcal{H}X_2^2 +X_1\mathcal{H}^2+re^{i\theta}\mathcal{H}^3$\tsep{2pt} & $r\geq 0$, $\theta\in [0,\frac{\pi}{2})$\bsep{2pt}\\ \hline
3g&$X_1^3+\mathcal{H}X_1X_2+\mathcal{H}X_2^2 +c_{10}\mathcal{H}^3$\tsep{2pt} & $c_{10}\in \mathbb{C}$\\ \hline
4a&$\mathcal{H}X_1^2+\mathcal{H}^2X_2 $\tsep{2pt}&\\ \hline
4b&$\mathcal{H}X_1^2+\mathcal{H}^2X_1+c_{10}\mathcal{H}^3 $\tsep{2pt}&$c_{10}\in\mathbb{C}$\\ \hline
4c&$\mathcal{H}X_1^2+c_{10}\mathcal{H}^3 $\tsep{2pt}&$c_{10}\in\{0,1\}$\\ \hline
4d&$\mathcal{H}X_1X_2+\mathcal{H}^2(X_1+X_2)+c_{10}\mathcal{H}^3 $\tsep{2pt}&$c_{10}\in\mathbb{C}$\\ \hline
4e&$\mathcal{H}X_1X_2+c_8\mathcal{H}^2X_1+c_{10}\mathcal{H}^3 $\tsep{2pt}&$c_8,c_{10}\in\{0,1\}$\\ \hline
4f&$\mathcal{H}^2X_1 $\tsep{2pt}&\\ \hline
4g&$c_{10}\mathcal{H}^3 $\tsep{2pt}&$c_{10}\in \{0,1\}$\\ \hline
\end{tabular}
\end{table}

\begin{Remark}
For each value of the parameter in the f\/irst two lines of Table~\ref{table1} if
 \begin{gather*}
c'_8=c^2(\alpha c_8+ \gamma c_9),\qquad c'_9=c^2(\beta c_8+\delta c_9),\qquad c'_{10}=c^3 c_{10},
\end{gather*}
for $c\in \mathbb{C}^*$ and
\begin{gather*}
\left(\begin{matrix}
\alpha &\beta\\
\gamma & \delta
\end{matrix}\right)\in \Omega(C_1)=\left\{\left(\begin{matrix}
0 & 1\\
1 & 0
\end{matrix}\right), \left(\begin{matrix}
0 & 1\\
-1 & -1
\end{matrix}\right),\left(\begin{matrix}
1 & 0\\
0 & 1
\end{matrix}\right) \right\} \\ \nonumber
\hphantom{\left(\begin{matrix}
\alpha &\beta\\
\gamma & \delta
\end{matrix}\right)\in \Omega(C_1)=}{}
 \coprod\left\{\left(\begin{matrix}
-1 & -1\\
0 & -1
\end{matrix}\right), \left(\begin{matrix}
1 & 0\\
-1 & -1
\end{matrix}\right),\left(\begin{matrix}
-1 & -1\\
1 & 0
\end{matrix}\right) \right\}
\end{gather*}
then the system with parameters $c_8$, $c_9$, $c_{10}$ isomorphic to the one with $c'_8$, $c'_9$, $c'_{10}$.
\end{Remark}

\subsection{Comparison of geometric and abstract nondegenerate quadratic algebras}\label{Comparison}

There is a close relationship between the canonical forms of abstract quadratic algebras and St\"ackel equivalence classes of nondegenerate superintegrable systems. To demonstrate this we treat one example in detail. The superintegrable system~$S9$, with nondegenerate potential, can be def\/ined by
\begin{gather*}{\mathcal R}^2= {\mathcal L}_1^2{\mathcal L}_2+{\mathcal L}_1{\mathcal L}_2^2+{\mathcal L}_1{\mathcal L}_2({\mathcal H}-a_4)
-a_2({\mathcal H}-a_4)^2 -2a_2{\mathcal L}_1({\mathcal H}-a_4)\\
\hphantom{{\mathcal R}^2=}{}-2a_2{\mathcal L}_2({\mathcal H}-a_4) -(a_3+a_2){\mathcal L}_1^2 -(a_3+3a_2+a_1){\mathcal L}_1{\mathcal L}_2
 -(a_2+a_1){\mathcal L}_2^2\\
\hphantom{{\mathcal R}^2=}{} +\big(2a_2a_3+2a_2^2+2a_1a_2\big)({\mathcal H}-a_4) +2\big(a_2^2+a_2a_3+a_1a_2\big){\mathcal L}_1\\
\hphantom{{\mathcal R}^2=}{}+2\big(a_2^2+a_2a_3+a_1a_2\big){\mathcal L}_2 +2a_1a_2a_3-2a_1a_2^2-2a_2^2a_3-a_2a_3^2-a_2a_1^2-a_2^3,
 \end{gather*}
where the $a_j$ are the parameters in the potential. To perform a general St\"ackel transform of this system with nonsingular transform matrix $C=(c_{jk})$: 1) we set $a_j=\sum\limits_{k=1}^4 c_{jk}b_k$, $k=1,\dots,4$ where the~$b_k$ are the new parameters, 2)~we make the replacements ${\mathcal H}\to -b_4$, $b_4\to -{\mathcal H}$ and 3)~we then set all parameters $b_j=0$ to determine the free quadratic algebra. The result is
\begin{gather*}
{\mathcal R}^2=c_{24}\big(c_{14}^2+2c_{14}c_{24}-2c_{14}c_{34}+2c_{14}c_{44}+c_{24}^2+2c_{24}c_{34}+2c_{24}c_{44}+c_{34}^2+2c_{34}c_{44} \\
\hphantom{{\mathcal R}^2=}{}+c_{44}^2\big){\mathcal H}^3 +(2c_{24}(c_{14}+c_{24}+c_{34}+c_{44}){\mathcal L}_1+2c_{24}(c_{14}+c_{24}+c_{34}+c_{44}){\mathcal L}_2){\mathcal H}^2\\
\hphantom{{\mathcal R}^2=}{} +(c_{24}+c_{34}){\mathcal L}_1^2{\mathcal H} +(c_{14}+c_{24}){\mathcal L}_2^2{\mathcal H}+(c_{14}+3c_{24}+c_{34}+c_{44}){\mathcal L}_2{\mathcal L}_1{\mathcal H} \\
\hphantom{{\mathcal R}^2=}{} +{\mathcal L}_1^2{\mathcal L}_2+{\mathcal L}_1{\mathcal L}_2^2.
 \end{gather*}
We put this in canonical form by making the choices ${\mathcal L}_1=X_1+(c_{24}+c_{14}){\mathcal H}$, ${\mathcal L}_2=X_2+(c_{34}+c_{24}){\mathcal H}$. The f\/inal result is
\begin{gather*}
 [1111]\colon \quad {\mathcal R}^2= X_1^2X_2+X_1X_2^2 +A_1X_1{\mathcal H}^2+A_2X_2{\mathcal H}^2+A_3X_1X_2{\mathcal H} +A_4{\mathcal H}^3, \end{gather*}
 where
\begin{alignat*}{3}
& A_1=(c_{24}-c_{34})(c_{14}+c_{44}),\qquad && A_2=(c_{34}+c_{44})(c_{24}-c_{14}),& \\
& A_3=-c_{14}-c_{24}-c_{34}+c_{44},\qquad && A_4=(c_{14}-c_{24}+c_{34}+c_{44})(c_{14}c_{34}+c_{24}c_{44}).&
\end{alignat*}
The \looseness=1 possible canonical forms in Table 1 associated with the equivalence class $[1111]$ depend on the possible choices of $c_{ij}$ with $\det C\ne 0$. The possible canonical forms are $1a$, $1b$, $1d$ all cases.

The superintegrable system $E1$, with nondegenerate potential, can be def\/ined by
\begin{gather*}
 {\mathcal R}^2={\mathcal L}_1{\mathcal L}_2({\mathcal H}-a_4)+{\mathcal L}_2^2{\mathcal L}_1-a_3({\mathcal H}-a_4)^2
-2a_3{\mathcal L}_2({\mathcal H}-a_4)\\
\hphantom{{\mathcal R}^2=}{} -(a_3+a_2){\mathcal L}_2^2-a_1{\mathcal L}_1^2+4a_1a_2a_3.\end{gather*}
 Going through the same procedure as above, we obtain the equivalence class
 \begin{alignat*}{3}
 & [211]\colon \quad && {\mathcal R}^2=-X_2X_1^2 +\big(2c_{14}c_{24}+2c_{14}c_{34}+\tfrac14 c_{44}^2\big)X_2{\mathcal H}^2+c_{44}(-c_{34}+c_{24})X_1{\mathcal H}^2& \\
 &&& \hphantom{{\mathcal R}^2=}{} +c_{14}X_2^2{\mathcal H} \big({-}2c_{14}c_{24}c_{34}+c_{14}c_{24}^2+c_{14}c_{34}^2+\tfrac12 c_{44}^2c_{24}+\tfrac12 c_{34}c_{44}^2\big){\mathcal H}^3.& \end{alignat*}
 The canonical forms associated with this equivalence class are $2a$, $ 2b$, $2c$, $2d$, $2e$, all cases.

The superintegrable system $E8$, with nondegenerate potential, can be def\/ined by
\begin{gather*} {\mathcal R}^2={\mathcal L}_2^2{\mathcal L}_1-a_2({\mathcal H}-a_4){\mathcal L}_2+4a_1a_3{\mathcal L}_1+a_1({\mathcal H}-a_4)^2-a_3a_2^2.\end{gather*}
 The equivalence class is
 \begin{gather*}
 [22]\colon \quad {\mathcal R}^2= X_1^2X_2 -c_{24}c_{44}X_1{\mathcal H}^2+4c_{14}c_{34}X_2{\mathcal H}^2
 +\big({-}c_{14}c_{44}^2+c_{34}c_{24}^2\big){\mathcal H}^3.\end{gather*}
 The canonical form associated with this equivalence class is $2a$: all cases.

The superintegrable system $E2$ can be def\/ined by
\begin{gather*}{\mathcal R}^2= {\mathcal L}_1^3 +{\mathcal L}_1{\mathcal H}^2-2{\mathcal L}_1^2{\mathcal H}
+(-2a_4{\mathcal L}_1-a_2{\mathcal L}_2){\mathcal H}+2a_4{\mathcal L}_1^2+\big(a_2{\mathcal L}_2+4a_1a_3+a_4^2\big){\mathcal L}_1\\
\hphantom{{\mathcal R}^2=}{}
+4a_1{\mathcal L}_2^2 +a_2a_4{\mathcal L}_2-\tfrac14 a_2^2a_3.\end{gather*}
The equivalence class is
\begin{alignat*}{3}
& [31]\colon \quad && {\mathcal R}^2=X_1^3 +\big(c_{14}X_1X_2-4c_{34}X_2^2+c_{44}X_1^2\big){\mathcal H}+4c_{34}c_{24}X_1{\mathcal H}^2&\\
&&& \hphantom{{\mathcal R}^2=}{} +\tfrac14 c_{24}\big(c_{14}^2+16c_{34}c_{44}\big){\mathcal H}^3.&
\end{alignat*}
The canonical forms associated with this equivalence class are $3d$: all cases, $3e$: $c_{10}=0$, $3f$: all cases, $3g$: $c_{10}=0$.

The superintegrable system $E10$ can be def\/ined by
\begin{gather*}{\mathcal R}^2= {\mathcal L}_1^3 +2a_1{\mathcal L}_1^2 -a_3{\mathcal L}_1{\mathcal L}_2+ a_3({\mathcal H}-a_4)^2+2a_2 {\mathcal L}_1({\mathcal H}-a_4)\\
\hphantom{{\mathcal R}^2=}{}
 +2a_1a_2({\mathcal H}-a_4)+a_1^2{\mathcal L}_1+a_2^2{\mathcal L}_2.\end{gather*}
 The equivalence class contains
\begin{alignat*}{3}
& [4]\colon \quad&& {\mathcal R}^2=X_1^3+c_{34}X_1X_2{\mathcal H} +\left(c_{24}^2+\frac23 c_{14}c_{34}\right)X_2{\mathcal H}^2&\\
&&& \hphantom{{\mathcal R}^2=}{} +\frac{1}{27}\frac{\big(8 c_{14}^3c_{34}+9c_{14}^2c_{24}^2 +54c_{14}c_{24}c_{34}c_{44} +54c_{24}^3c_{44}-27c_{34}^2c_{44}^2\big)}{c_{34}}{\mathcal H}^3,& \end{alignat*}
if $c_{34}\ne 0$. If $c_{34}=0$, $c_{24}\ne 0$ it contains
\begin{gather*}[4]'\colon \quad {\mathcal R}^2=X_1^3-2 c_{14}^2X_1^2{\mathcal H}+c_{24}^2X_2{\mathcal H}^2+2c_{14}c_{24}c_{44}{\mathcal H}^3,\end{gather*}
and if $c_{34}=c_{24}=0$ it contains
\begin{gather*}[4]''\colon \quad {\mathcal R}^2=X_1^3+ c_{14}^2X_1{\mathcal H}^2-2c_{14}X_1^2{\mathcal H}.\end{gather*}
 The canonical form associated with $[4]$ is $3d$ all cases. The canonical form associated with $[4]'$ is $3c$: all cases, and the canonical form associated with $[4]''$ is $3a$: $c_{10}\ne 0$.

The superintegrable system $E3'$ can be def\/ined by
\begin{gather*} {\mathcal R}^2= -4a_1\big({\mathcal L}_1^2+{\mathcal L}_2^2-{\mathcal L}_2{\mathcal H}\big) -2a_2a_3{\mathcal L}_1
+\big(a_2^2-a_3^2-4a_1a_4\big){\mathcal L}_2 -a_3^2a_4+a_3^2{\mathcal H}.\end{gather*}

The canonical form is
\begin{gather*}
[0]\colon \quad {\mathcal R}^2=4c_{14}\big(X_1^2+X_2^2\big){\mathcal H} -
\frac{\big(4c_{14}c_{44}-c_{24}^2-c_{34}^2\big)^2}{16c_{14}}{\mathcal H}^3,\end{gather*}
if $c_{14}\ne 0$; if $c_{14}=0$ it is
\begin{gather*} [0]'\colon \quad {\mathcal R}^2=-2c_{24}c_{34}X_1{\mathcal H}^2+\big(c_{24}^2-c_{34}^2\big)X_2{\mathcal H}^2+c_{34}^2c_{44}{\mathcal H}^3.\end{gather*}
The canonical forms associated with $[0]$ are $4d$: all cases, $4e$: all cases, and the canonical forms associated with $[0]'$ are $ 4f$: all cases.

 {\bf Heisenberg systems.} In addition there are systems that can be obtained from the geometric systems above by contractions from ${\mathfrak{so}}(4,\mathbb{C})$ to ${\mathfrak{e}}(3,\mathbb{C})$. These are not B\^ocher contractions and the contracted systems are not superintegrable, because the Hamiltonians become singular. However, they do form quadratic algebras and many have the interpretation of time-dependent Schr\"odinger equations in 2D spacetime, so we also consider them geometrical. Some of these were classif\/ied in \cite{KM2014} where they were called Heisenberg systems since they appeared in quadratic algebras formed from 2nd order elements in the Heisenberg algebra with generators ${\cal M}_1=p_x$, ${\cal M}_2=xp_y$, ${\cal E}=p_y$, where ${\cal E}^2={\cal H}$. The systems are all of type~4. We will devote a future paper to their study. The ones classif\/ied so far are~$4a$: all cases, $4c$: $c_{10}=0$, $4e$: $c_{10}=0$, $4f$: all cases, $4g$: all cases.

All these results relating geometric systems to abstract systems are summarized in Table~\ref{table2}.

\begin{table}[t]\footnotesize\centering
\caption{Matching of geometric with abstract quadratic algebras.}\label{table2}\vspace{1mm}

\begin{tabular}{|l|l|l|l|}
\hline Class&\multicolumn{3}{c|}{Canonical form}\tsep{1pt}\bsep{1pt}\\ \hline
1& $ a$: all cases & $b$: all cases &$c$: no \tsep{1pt}\bsep{1pt}\\ \hline
1& $d$: all cases&&\tsep{1pt}\bsep{1pt}\\ \hline
2& $a$: all cases& $b$ all cases &$c$: all cases\tsep{1pt}\bsep{1pt}\\ \hline
 2& $d$: all cases & $e$: all cases& \tsep{1pt}\bsep{1pt}\\ \hline
3& $a$: $c_{10}\ne 0$ & $b$: no & $c$: all cases \tsep{1pt}\bsep{1pt}\\ \hline
3&$d$: all cases& $e$: $c_{10}=0$ & $f$: all cases \tsep{1pt}\bsep{1pt}\\ \hline
3& $g$: $c_{10}=0 $ &&\tsep{1pt}\bsep{1pt}\\ \hline
4& $a$: all cases & $b$: no & $c$: $c_{10}=0$ \tsep{1pt}\bsep{1pt}\\ \hline
4&$d$: all cases & $e$: all cases& $f$: all cases \tsep{1pt}\bsep{1pt}\\ \hline
4& $4g$: all cases &&\tsep{1pt}\bsep{1pt}\\ \hline
\end{tabular}
\end{table}

\section[The quadratic algebras of the free 2D second order superintegrable systems]{The quadratic algebras of the free 2D second order\\ superintegrable systems}

In this section we list all canonical forms of the Casimirs of the quadratic algebras of free nondegenerate 2D superintegrable systems on a constant curvature space or a Darboux space. We list the canonical forms arising from superintegrable systems on a constant curvature spaces in Table~\ref{CanND} and those arising from superintegrable systems on a Darboux space in Table~\ref{CanNDD}. In the next section we study contractions between these quadratic algebras.

\begin{table}[h!]\footnotesize\centering
\caption{Canonical forms of the Casimirs of quadratic algebras of free nondegenerate 2D superintegrable systems that lie
inside $\mathcal{U}({\mathfrak{so}}(3,\mathbb {C}))$ and
$\mathcal{U}({\mathfrak{e}}(2,\mathbb {C}))$.}\label{CanND}\vspace{1mm}
\begin{tabular}{|l|l|}
\hline
System & Canonical forms of $\mathcal{R}^2$\tsep{2pt}\bsep{1pt}
\\ \hline
 $\widetilde{E}_{17}$& $\mathcal{L}_1^2 \mathcal{L}_2$ \tsep{2pt}\bsep{1pt}\\ \hline
 $\widetilde{E}_{16}$ & $\mathcal{L}_1^2 \mathcal{L}_2+\mathcal{H}\mathcal{L}_2^2$ \tsep{2pt}\bsep{1pt}\\ \hline
 $\widetilde{E}_{1}$ & $\mathcal{L}_1^2 \mathcal{L}_2+\mathcal{H}^2\mathcal{L}_2$ \tsep{2pt}\bsep{1pt}\\ \hline
 $\widetilde{E}_{8}$ & $\mathcal{L}_1^2 \mathcal{L}_2$ \tsep{2pt}\bsep{1pt}\\ \hline
 $\widetilde{E}_{3}'$ & $0$ \tsep{2pt}\bsep{1pt}\\ \hline
 $\widetilde{E}_{2}$ & $\mathcal{L}_1^3 +\mathcal{H}^2\mathcal{L}_1+\frac{2i}{3\sqrt{3}}\mathcal{H}^3$ \tsep{2pt}\bsep{1pt}\\ \hline
 $\widetilde{E}_{7}$ & $\mathcal{L}_1^2\mathcal{L}_2 $, $\forall\, a$ \tsep{2pt}\bsep{1pt}\\ \hline
 $\widetilde{E}_{9}$ & $\mathcal{L}_1^3 +\mathcal{H}^2\mathcal{L}_1+\frac{2i}{3\sqrt{3}}\mathcal{H}^3$ \tsep{2pt}\bsep{1pt}\\ \hline
 $\widetilde{E}_{11}$ & $\mathcal{H}^2\mathcal{L}_1 $ \tsep{2pt}\bsep{1pt}\\ \hline
 $\widetilde{E}_{10}$ & $\mathcal{L}_1^3$\tsep{2pt}\bsep{1pt}\\ \hline
 $\widetilde{E}_{15}$ & $\mathcal{L}_1^3$ \tsep{2pt}\bsep{1pt}\\ \hline
 $\widetilde{E}_{20}$ & $\mathcal{H}\mathcal{L}_1\mathcal{L}_2$ \tsep{2pt}\bsep{1pt}\\ \hline
 $\widetilde{E}_{19}$ & $\mathcal{L}_1^2 \mathcal{L}_2+\mathcal{H}^2\mathcal{L}_2$ \tsep{2pt}\bsep{1pt}\\ \hline
 $\widetilde{S}_{9}$ & $\mathcal{L}_1 \mathcal{L}_2(\mathcal{L}_1+ \mathcal{L}_2)+\mathcal{H}\mathcal{L}_1 \mathcal{L}_2$ \tsep{2pt}\bsep{1pt}\\ \hline
 $\widetilde{S}_{4}$ & $\mathcal{L}_1^2 \mathcal{L}_2$ \tsep{2pt}\bsep{1pt}\\ \hline
$\widetilde{S}_{7}$ & $\mathcal{L}_1 \mathcal{L}_2(\mathcal{L}_1+ \mathcal{L}_2)+\mathcal{H}\mathcal{L}_1 \mathcal{L}_2
-\frac{1}{4}\mathcal{H}^2\mathcal{L}_1-\frac{1}{4}\mathcal{H}^2\mathcal{L}_2-\frac{1}{4}\mathcal{H}^3$ \tsep{2pt}\bsep{1pt}\\ \hline
$\widetilde{S}_{8}$ & $\mathcal{L}_1 \mathcal{L}_2(\mathcal{L}_1+ \mathcal{L}_2)+\mathcal{H}\mathcal{L}_1 \mathcal{L}_2$ \tsep{2pt}\bsep{1pt}\\ \hline
$\widetilde{S}_{2}$ & $\mathcal{L}_1^2 \mathcal{L}_2$ \tsep{2pt}\bsep{1pt}\\ \hline
$\widetilde{S}_{1}$ & $\mathcal{L}_1^3$ \tsep{2pt}\bsep{1pt}\\ \hline
\end{tabular}
\end{table}

\newpage

\begin{table}[t!]\footnotesize\centering
\caption{Canonical forms of the Casimirs of quadratic algebras of free nondegenerate 2D Darboux superintegrable systems.}\label{CanNDD} \vspace{1mm}
\begin{tabular}{|l|l|}
\hline
System & Canonical forms of $\mathcal{R}^2$\tsep{2pt}\bsep{1pt}
\\ \hline
 $\widetilde{D}1A$, $b=0$& $\mathcal{L}_1^3 +\mathcal{H}\mathcal{L}_1\mathcal{L}_2$ \tsep{2pt}\bsep{1pt}\\ \hline
 $\widetilde{D}1A$, $b\neq 0$& $\mathcal{L}_1^3 +\mathcal{H}\mathcal{L}_1\mathcal{L}_2+\mathcal{H}^3$ \tsep{2pt}\bsep{1pt}\\ \hline
 $\widetilde{D}1B$& $\mathcal{L}_1^3 +\mathcal{H}\mathcal{L}_1\mathcal{L}_2$ \tsep{2pt}\bsep{1pt}\\ \hline
 $\widetilde{D}1C$& $\mathcal{H}^2\mathcal{L}_1$ \tsep{2pt}\bsep{1pt}\\ \hline
 $\widetilde{D}2A$& $\mathcal{L}_1^3 +\mathcal{H}^2\mathcal{L}_1+\frac{2i}{3\sqrt{3}}\mathcal{H}^3$ \tsep{2pt}\bsep{1pt}\\ \hline
 $\widetilde{D}2B$& $\mathcal{L}_1^2\mathcal{L}_2 +\mathcal{H}^2\mathcal{L}_1+\mathcal{H}^2\mathcal{L}_2+i\mathcal{H}^3$ \tsep{2pt}\bsep{1pt}\\ \hline
 $\widetilde{D}2C$& $\mathcal{L}_1^2\mathcal{L}_2 +\mathcal{H}\mathcal{L}_2^2+\mathcal{H}^2\mathcal{L}_2$ \tsep{2pt}\bsep{1pt}\\ \hline
 $\widetilde{D}3A$& $\mathcal{H}\mathcal{L}_1\mathcal{L}_2+\mathcal{H}^3$ \tsep{2pt}\bsep{1pt}\\ \hline
 $\widetilde{D}3B$& $\mathcal{L}_1^2\mathcal{L}_2 +\mathcal{H}\mathcal{L}_2^2+\mathcal{H}^2\mathcal{L}_2$ \tsep{2pt}\bsep{1pt}\\ \hline
 $\widetilde{D}3C$& $\mathcal{L}_1^2\mathcal{L}_2 +\mathcal{H}\mathcal{L}_2^2+\mathcal{H}^2\mathcal{L}_2$ \tsep{2pt}\bsep{1pt}\\ \hline
 $\widetilde{D}3D$& $\mathcal{L}_1^2\mathcal{L}_2+\mathcal{H}\mathcal{L}_1^2+\mathcal{H}\mathcal{L}_2^2+i3\sqrt{2}\mathcal{H}^3$ \tsep{2pt}\bsep{1pt}\\ \hline
 $\widetilde{D}4A$& $\mathcal{L}_1^2\mathcal{L}_2$ \tsep{2pt}\bsep{1pt}\\ \hline
 $\widetilde{D}4(b)B$, $b\neq0$& $\mathcal{L}_1\mathcal{L}_2(\mathcal{L}_1+\mathcal{L}_2)
 +\mathcal{H}\mathcal{L}_1\mathcal{L}_2+\frac{b^2-4}{4b^2}\mathcal{H}^2\mathcal{L}_1$ \tsep{2pt}\bsep{1pt}\\ \hline
 $\widetilde{D}4(b)B$, $b=0$& $\mathcal{L}_1\mathcal{L}_2(\mathcal{L}_1+\mathcal{L}_2)+\mathcal{H}^2\mathcal{L}_1$ \tsep{2pt}\bsep{1pt}\\ \hline
 $\widetilde{D}4(b)C$, $b\neq0$& $\mathcal{L}_1\mathcal{L}_2(\mathcal{L}_1+\mathcal{L}_2)
 +\mathcal{H}\mathcal{L}_1\mathcal{L}_2+\frac{1}{b^2}\mathcal{H}^2\mathcal{L}_1$ \tsep{2pt}\bsep{1pt}\\ \hline
 $\widetilde{D}4(b)C$, $b=0$& $\mathcal{L}_1\mathcal{L}_2(\mathcal{L}_1+\mathcal{L}_2)+\mathcal{H}^2\mathcal{L}_1$ \tsep{2pt}\bsep{1pt}\\ \hline
 \end{tabular}
\end{table}

\section[Abstract contractions of nondegenerate quadratic algebras arising from 2D second order superintegrable systems on constant curvature spaces and Darboux spaces]{Abstract contractions of nondegenerate quadratic algebras\\ arising from 2D second order superintegrable systems\\ on constant curvature spaces and Darboux spaces}

We f\/irst recall the def\/inition of contraction of quadratic algebras.
\begin{Definition}
Let $\mathcal{A}$ and $\mathcal{A}_0$ be quadratic algebras with generating sets $\{\mathcal{H},\mathcal{L}_1,\mathcal{L}_2 \}$ and $\{\mathcal{H}^0,\mathcal{L}_1^0,\mathcal{L}_2^0 \}$ respectively, satisfying the conditions of Def\/inition~\ref{def5}. Let $\mathcal{F}(\mathcal{H}, \mathcal{L}_1, \mathcal{L}_2)$ be the realization of the Casimir of $\mathcal{A}$ in the generating set $\{\mathcal{H}, \mathcal{L}_1, \mathcal{L}_2 \}$ and similarly $\mathcal{F}^0(\mathcal{H}^0,\mathcal{L}_1^0,\mathcal{L}_2^0)$ the Casimir of $\mathcal{A}^0$ in the generating set $\{\mathcal{H}^0, \mathcal{L}^0_1, \mathcal{L}^0_2 \}$. We say that $\mathcal{A}_0$ is a contraction of $\mathcal{A}$ if there is a continuous curve
\begin{gather*}
(0,1]\longrightarrow G,\qquad
 \epsilon \longmapsto A(\epsilon)=\left(\begin{matrix}
A_{1,1}(\epsilon) & A_{1,2}(\epsilon)&A_{1,3}(\epsilon) \\
A_{2,1}(\epsilon)&A_{2,2}(\epsilon) &A_{2,3}(\epsilon) \\
0 &0 &A_{3,3}(\epsilon)
\end{matrix}\right)
\end{gather*}
such that
\begin{gather*}
\lim_{\epsilon \longrightarrow 0^+}A(\epsilon)\cdot F(X_1,X_2,X_3)=F^0(X_1,X_2,X_3).
\end{gather*}
Note that the action of $G$ is def\/ined in (\ref{e8}).
\end{Definition}

Note that if $\mathcal{A}_0$ is a contraction of $\mathcal{A}$ then $\mathcal{A}_0$ is in the closure of the orbit of $G$ that contains~$\mathcal{A}$.

\subsection{Contractions of quadratic algebras}
In this section we study contractions between the quadratic algebras that arise from free nondegenerate 2D second order superintegrable system on a~constant curvature space or a Darboux space. As we shall see below there are essentially 18 relevant quadratic algebras for classif\/ication purposes. For any two such quadratic algebras one can ask weather there is a contraction from one to the other. In principal there are $324=18^2$ cases to consider. We have studied most of these cases but our results do not give a complete classif\/ication. We discus our results in more details below. We shall give several contractions explicitly and write all those contractions that we were able to f\/ind in a diagram. At the end of this section we shall compare abstract contractions with B\^ocher contractions.

\subsubsection{The relevant quadratic algebras}\label{sec721}
We f\/irst note that some quadratic algebras of dif\/ferent superintegrable systems coincide:
\begin{enumerate}\itemsep=0pt
 \item[1)] $\mathcal{L}_1 \mathcal{L}_2(\mathcal{L}_1+ \mathcal{L}_2)+\mathcal{H}\mathcal{L}_1 \mathcal{L}_2$:
 $\widetilde{S}_{8}$, $\widetilde{S}_{9}$ , $\widetilde{D}4(b=\pm 2)C$,
\item[2)] $\mathcal{L}_1 \mathcal{L}_2(\mathcal{L}_1+ \mathcal{L}_2)+\mathcal{H}^2\mathcal{L}_1$: $\widetilde{D}4(b=0)B$, $\widetilde{D}4(b=0)C$,
\item[3)] $\mathcal{L}_1 \mathcal{L}_2(\mathcal{L}_1+ \mathcal{L}_2)+\mathcal{H}\mathcal{L}_1 \mathcal{L}_2+\gamma\mathcal{H}^2\mathcal{L}_1$:
$\widetilde{D}4(\gamma=b^{-2})B$, $\widetilde{D}4(\gamma =\frac{b^2-4}{4b^2})C$,
 \item[4)] $\mathcal{L}_1^2\mathcal{L}_2 +\mathcal{H}\mathcal{L}_2^2+\mathcal{H}^2\mathcal{L}_2$: $\widetilde{D}2C$,
 $\widetilde{D}3B$, $\widetilde{D}3C$,
 \item[5)] $\mathcal{L}_1^2 \mathcal{L}_2$: $\widetilde{E}_{17}$,$\widetilde{E}_{8}$, $\widetilde{S}_{2}$, $\widetilde{S}_{4}$, $\widetilde{E}_{7}$,
 $\widetilde{D}4A$,
 \item[6)] $\mathcal{L}_1^2 \mathcal{L}_2+\mathcal{H}^2\mathcal{L}_2$: $\widetilde{E}_{1}$, $\widetilde{E}_{19}$,
 \item[7)] $\mathcal{L}_1^3$: $\widetilde{E}_{10}$, $\widetilde{E}_{15}$, $\widetilde{S}_{1}$,
 \item[8)] $\mathcal{L}_1^3+\mathcal{H}^2\mathcal{L}_1+i\frac{2}{3\sqrt{3}}\mathcal{H}^3$: $\widetilde{E}_{2}$, $\widetilde{E}_{9}$, $\widetilde{D}2A$,
 \item[9)] $\mathcal{L}_1^3+\mathcal{H}\mathcal{L}_1\mathcal{L}_2$: $\widetilde{D}1A(b=0)$, $\widetilde{D}1B$,
 \item[10)] $\mathcal{H}^2\mathcal{L}_1$: $\widetilde{E}_{11}$, $\widetilde{D}1C$.
\end{enumerate}
Hence it is enough to consider the eighteen quadratic algebras:
\begin{gather*}
\widetilde{E}_{17},\quad \widetilde{E}_{16}, \quad \widetilde{E}_{1}, \quad \widetilde{E}_{3}', \quad \widetilde{E}_{2}, \quad \widetilde{E}_{11}, \widetilde{E}_{10}, \quad \widetilde{E}_{20}, \quad \widetilde{S}_{9}, \quad \widetilde{S}_{7}, \quad \widetilde{D}4C \ (b\neq 0),\\
\widetilde{D}4C \ (b=0), \quad \widetilde{D}2B,\quad \widetilde{D}2C, \quad \widetilde{D}1A \ (b\neq0), \quad \widetilde{D}1A \ (b=0), \quad \widetilde{D}3A, \quad \widetilde{D}3D.
\end{gather*}
We divide the quadratic algebras into four sets according to the highest non-vanishing $F^{(i)}$ term in the decomposition
\begin{gather*} {\mathcal{R}}^2={\mathcal{F}}(\mathcal{H} ,\mathcal{L}_1,\mathcal{L}_2)=
{\mathcal{F}}^{(3)}({\mathcal{L}}_1,{\mathcal{L}}_2)+{\mathcal{H}}{\mathcal{F}}^{(2)}({\mathcal{L}}_1,{\mathcal{L}}_2)
 +{\mathcal{H}}^2{\mathcal{F}}^{(1)}({\mathcal{L}}_1,{\mathcal{L}}_2)+{\mathcal{H}}^3{\mathcal{F}}^{(0)}.\end{gather*}
Explicitly we def\/ine
\begin{itemize}\itemsep=0pt
 \item subset $A$: $F^{(3)}\neq 0$: $\widetilde{E}_{17}$, $\widetilde{E}_{16}$, $\widetilde{E}_{1}$, $\widetilde{E}_{2}$, $\widetilde{E}_{10}$, $\widetilde{S}_{9}$, $\widetilde{S}_{7}$, $\widetilde{D}4C$ $(b\neq 0)$, $\widetilde{D}4C$ $(b=0)$, $\widetilde{D}2B$, $\widetilde{D}2C$, $\widetilde{D}1A$ $(b\neq 0)$, $\widetilde{D}1A$ $(b=0)$, $\widetilde{D}3D$,
 \item subset $B$: $F^{(3)}=0$, $F^{(2)}\neq 0$: $\widetilde{E}_{20}$, $\widetilde{D}3A$,
 \item subset $C$: $F^{(3)}=F^{(2)}=0, F^{(1)}\neq 0$: $\widetilde{E}_{11}$,
 \item subset $D$: $F^{(3)}=F^{(2)}=F^{(1)}= 0$: $\widetilde{E}_{3}'$.
\end{itemize}
Since $F^{(3)}$ is a homogeneous polynomial of degree three in two variables, it has exactly three roots (zeros) on $\mathbb{C}\mathbb{P}^1$ counting multiplicities. We divide subset $A$ according to the number of dif\/ferent roots of $F^{(3)}$ as follows
 \begin{itemize}\itemsep=0pt
 \item three distinct roots, subset $A_1$: $\widetilde{S}_{9}$, $\widetilde{S}_{7}$, $\widetilde{D}4C$ $(b\neq 0)$, $\widetilde{D}4C$ $(b=0)$,
 \item a repeated root, subset $A_2$: $\widetilde{E}_{17}$, $\widetilde{E}_{16}$, $\widetilde{E}_{1}$, $\widetilde{D}2B$, $\widetilde{D}2C$,
 $\widetilde{D}3D$,
 \item a triple root, subset $A_3$: $\widetilde{E}_{2}$, $\widetilde{E}_{10}$, $\widetilde{D}1A$ $(b\neq 0)$, $\widetilde{D}1A$ $(b=0)$.
 \end{itemize}

 \subsubsection{Some general observations on contractions of quadratic algebras}
Note that the group \begin{gather*}
G=\left\{\left(\begin{matrix}
A_{1,1} & A_{1,2}&A_{1,3} \\
A_{2,1}&A_{2,2} &A_{2,3} \\
0 &0 &A_{3,3}
\end{matrix}\right)\in {\rm GL}(3,\mathbb {C})\right\}
\end{gather*} is a complex algebraic group. The formula
\begin{gather*}
\left( A\cdot \mathcal{F} \right)(x_1,x_2,x_3) = \det (A_2)^2{\mathcal{F}}\big(A^{-1}(x_1,x_2,x_3) \big)
 \end{gather*}
def\/ines an algebraic action of $G$ on the complex algebraic variety $\mathbb {C}^{[3]}[x_1,x_2,x_3]$, of homogeneous polynomials of degree three in three variables. It is well known (see, e.g., \cite[Section~1.8]{Borel}) that any orbit is an algebraic variety and the boundary of any orbit is also an algebraic variety of a smaller dimension. From this consideration it is clear that if $O_1$ and $O_2$ are two orbits such that $O_2\subset \overline{O_1}\setminus {O_1}$ then $O_1\nsubseteq \overline{O_2}$. This imply that we have a partial order by inclusion of orbit closure. In our language this implies that if a quadratic algebra $B$ is a contraction of a quadratic algebra~$A$ and~$A$ and~$B$ are not isomorphic then~$A$ is not a contraction of~$B$. Hence for any contraction of
quadratic algebras between non isomorphic ones we automatically get a proof of the nonexistence of a contraction in the opposite direction.

Furthermore, under the action of $G$ on $\mathbb {C}^{[3]}[x_1,x_2,x_3]$ the sets $A$, $A_1$, $A_2$, $A_3$, $B$, $C$, $D$ are stable and hence consists
of a union of orbits. It is easy to see that the hierarchy of the orbits allow us to consider contractions only in the following direction
\begin{gather*}
 A_1 \longrightarrow A_2 \longrightarrow A_3 \longrightarrow B \longrightarrow C \longrightarrow D.
\end{gather*}

We further note that every quadratic algebra can be contracted to $\widetilde{E}_{3}'$ and $\widetilde{E}_{3}'$ can not be contracted further, hence we we shall ignore this system. In the rest of this section we realize many contraction of quadratic algebras and demonstrate how one can prove that some contractions do not exist. At the end of the section we summarize our results in a diagram.

 \subsection{Explicit contractions}\label{sec7.3}
Using matrices of the form
\begin{gather*}
A(\epsilon)=\left(\begin{matrix}
1& 0&0 \\
0&1 &0 \\
0 &0 &\epsilon
\end{matrix}\right)^{-1}, \qquad A(\epsilon)=\left(\begin{matrix}
1& 0&0 \\
0&1 &0 \\
0 &0 &\epsilon
\end{matrix}\right)^{-1}, \qquad A(\epsilon)=\left(\begin{matrix}
1& 0&0 \\
\epsilon^{-2}&\epsilon^{-1} &0 \\
0 &0 &\epsilon^{-3}
\end{matrix}\right)^{-1},\\
A(\epsilon)=\left(\begin{matrix}
\epsilon^{-2} & \epsilon^{-1}/\sqrt{2}&0 \\
\epsilon^{-2} &-\epsilon^{-1}/\sqrt{2} &0 \\
0 &0 &1
\end{matrix}\right)^{-1}, \qquad A(\epsilon)=\left(\begin{matrix}
\epsilon^{-1} & 0&0 \\
0&1 &0 \\
0 &0 &\epsilon^{-3}
\end{matrix}\right)^{-1}
\end{gather*}
we can (respectively) realize contractions of the following forms:
\begin{alignat*}{3}
& \mathcal{L}_1^2\mathcal{L}_{2}+O(\mathcal{H})\longrightarrow \mathcal{L}_1^2\mathcal{L}_{2}\colon \quad &&
 {D}_{3D}, \ {D}_{2C}, \ {D}_{2B}, \ {E}_{16}, \ {E}_{1} \longrightarrow {E}_{17},&\\
&\mathcal{L}_1^3+O(\mathcal{H})\longrightarrow \mathcal{L}_1^3\colon \quad && {D}_{1A}, \ {D}_{1A}, \ {E}_{2}\longrightarrow {E}_{10},&\\
& \mathcal{L}_1^2\mathcal{L}_{2}+O(\mathcal{H})\longrightarrow \mathcal{L}_1^3\colon \quad&&
{D}_{3D}, \ {D}_{2C}, \ {D}_{2B}, \ {E}_{16}, \ {E}_{1}, \ {E}_{17} \longrightarrow {E}_{10},&\\
& \mathcal{L}_1\mathcal{L}_{2}(\mathcal{L}_1+\mathcal{L}_{2})+O(\mathcal{H})\longrightarrow \mathcal{L}_1^3\colon \quad &&
S_{9}, \ {S}_{7}, \ D_{4C}, \ {D}_{4C} \longrightarrow {E}_{10},&\\
& \mathcal{L}_1\mathcal{L}_{2}(\mathcal{L}_1+\mathcal{L}_{2})+O(\mathcal{H})\longrightarrow
 \mathcal{L}_1^2\mathcal{L}_{2}\colon \quad &&S_{9}, \ {S}_{7}, \ D_{4C}, \ {D}_{4C} \longrightarrow \ {E}_{17}.&
\end{alignat*}

To get an idea of the type of contractions that exist, below we list realizations of all other abstract contractions of~$S_9$ that we have found.

 \textbf{Contraction of $\boldsymbol{S_9}$ to $\boldsymbol{E_{20}}$}:
$A(\epsilon)=\left(\begin{matrix}
\epsilon& 0&0 \\
0&\epsilon &0 \\
0 &0 &\epsilon^{2}
\end{matrix}\right)$.

 \textbf{Contraction of $\boldsymbol{S_9}$ to $\boldsymbol{E_1}$}:
$A(\epsilon)=\left(\begin{matrix}
\epsilon^{-1}& 0&-i\epsilon^{-1} \\
0&1&0\\
0 &0 &2i\epsilon^{-1}
\end{matrix}\right)^{-1}$.

 \textbf{Contraction of $\boldsymbol{S_9}$ to $\boldsymbol{E_{11}}$}:
$A(\epsilon)=\left(\begin{matrix}
\epsilon^{-1}& 0&0 \\
0&\epsilon^{-1}&-\epsilon^{-3/2} \\
0 &0 &1
\end{matrix}\right)^{-1}$.

 \textbf{Contraction of $\boldsymbol{S_9}$ to $\boldsymbol{E_2}$ }:
$A(\epsilon)=\left(\begin{matrix}
64 \epsilon^2&64 \epsilon^2&64\epsilon^2+i\frac{128}{\sqrt{3}}\epsilon \\
i8 \epsilon&-i8 \epsilon&0 \\
0 &0 &-i128\sqrt{3}\epsilon
\end{matrix}\right)$.

\subsection{Non-contractions}
Here we demonstrate how one can show that there are some quadratic algebras that can not be contracted to some others.

\textbf{Non-contraction of $\boldsymbol{E_{10}}$ to $\boldsymbol{E_{11}}$.} Under a transformation of the form
\begin{gather*}
\left(\begin{matrix}
\mathcal{L}_1 \\
\mathcal{L}_2 \\
\mathcal{H}
\end{matrix}\right) =\left(\begin{matrix}
\alpha(\epsilon) & \beta(\epsilon)&a(\epsilon) \\
\gamma(\epsilon)&\delta(\epsilon) &b(\epsilon) \\
0 &0 &c(\epsilon)
\end{matrix}\right)\left(\begin{matrix}
\mathcal{L}_1^{\epsilon} \\
\mathcal{L}_2^{\epsilon} \\
\mathcal{H}^{\epsilon}
\end{matrix}\right)=A\left(\begin{matrix}
\mathcal{L}_1^{\epsilon} \\
\mathcal{L}_2^{\epsilon} \\
\mathcal{H}^{\epsilon}
\end{matrix}\right).
\end{gather*}
We let $(\alpha\delta-\beta \gamma)=|A|$ and we denote the coef\/f\/icient of $\mathcal{L}_1^i\mathcal{L}_2^j\mathcal{L}_3^k$ in the transformed expression for $\mathcal{R}^2$ by $C_{i,j,k}$. Then we see that
\begin{gather*} C_{3,0,0}=\frac{\alpha^3}{A^2}\longrightarrow 0,\qquad
C_{2,0,1}=\frac{3a^2 \alpha}{A^2}\longrightarrow 1,\qquad
C_{0,0,3}=\frac{a^3}{A^2}\longrightarrow 0,\end{gather*}
which imply that $\frac{\alpha}{a}\longrightarrow 0$, $\frac{a}{\alpha}\longrightarrow 0$, which is a contradiction.

 All abstract contractions relating free constant curvature and Darboux quadratic algebras are listed in Diagram~1. There is an abstract contraction of~$Q(A)$ to $Q(B)$ if and only if there is an arrow in the diagram pointing from~$A$ to~$B$.

 \begin{figure}[t]\centering
 $$\xymatrix{
&S9\ar[ddddddrr]\ar[rrrrd]\ar@/^/[rrdd]\ar@/_/[rrrddd]\ar@/^/[rrdddd]\ar[ddddd] & D4C(b\neq0)\ar[ddddddr]\ar[rdd]\ar[rdddd]\ar[dddddl] & &D4C(b=0)\ar[ldd]\ar[ldddd]\ar[ddddddl]&S7\ar[llld]\ar@/_/[lldd]\ar[dddl]\ar[ddddll]\ar@/^/[ddddddll] \\
D3D\ar@/_/[rrrddd]\ar[rrrd]\ar[dddddrrr]&D2C\ar[rrd]\ar[r]\ar[rrddd]\ar[rrrdd]\ar[dddddrr]& E16\ar[rd]\ar[rddd]\ar@/^/[rrdd]\ar[ddddl]\ar[dddddr]& &D2B\ar[lddd]\ar[ld]\ar[dd]\ar[dddddl]& E1\ar[lld]\ar[llddd]\ar[ldd]\ar[lllldddd]\ar[dddddll]\\
&&&E17\ar[dd]\ar[rd]\ar@/^/[dddd]&&\\
{D1A}(b\neq0)\ar[rrrddd]\ar[rrrd]\ar@/^/[rrrr]\ar[ddr]\ar[dd]& &{D1A}(b=0)\ar[rd]\ar[ddl]\ar@/_/[rddd]& &E2\ar[ld]\ar[lddd]\\
&&&E10&&\\
D3A\ar[r]\ar[rrrd] & E20\ar[rrd] &&&&\\
&&&E11&&
}$$
{\small \textbf{Diagram 1.} Abstract contractions relating free nondegenerate 2D quadratic algebras.}
\end{figure}

\subsection{Comparison between abstract contractions and B\^ocher contractions}
In this section we compare abstract contractions and B\^ocher contractions. In previous sections we studied abstract contractions between the quadratic algebras of the free 2D nondegenerate second order superintegrable systems:
\begin{gather*}
\widetilde{E}_{17}, \ \widetilde{E}_{16}, \ \widetilde{E}_{1}, \ \widetilde{E}_{3}', \ \widetilde{E}_{2}, \ \widetilde{E}_{11}, \
\widetilde{E}_{10}, \ \widetilde{E}_{20}, \ \widetilde{S}_{9}, \ \widetilde{S}_{7}, \ \widetilde{D}4C \ (b\neq 0), \ \widetilde{D}4C \ (b=0), \ \widetilde{D}2B,\\
 \widetilde{D}2C, \ \widetilde{D}1A \ (b\neq0), \ \widetilde{D}1A \ (b=0), \ \widetilde{D}3A, \ \widetilde{D}3D.
\end{gather*}
By abuse of notation we denoted a superintegrable system and its corresponding free quadratic algebra by the same symbol (one of those 18 options above). It should be noted that dif\/ferent superintegrable systems may have the same free quadratic algebra, as was shown in Section~\ref{sec721}. For this section we shall use the symbol~$\widetilde{S}9$ to denote the superintegrable system on the complex two sphere and use the symbol $Q(\widetilde{S}9)$ to denote the free quadratic algebra of~$\widetilde{S}9$. Similar conventions will be used for all other systems. For example,
\begin{gather*}
Q\big(\widetilde{E}_{17}\big)=Q\big(\widetilde{E}_{8}\big)=Q\big(\widetilde{S}_{2}\big)=Q\big(\widetilde{S}_{4}\big) =Q\big(\widetilde{E}_{7}\big)=Q\big(\widetilde{D}4A\big).
\end{gather*}
As we just observed superintegrable systems that share the same free quadratic algebra can still live on dif\/ferent manifolds. Note that in general superintegrable systems with identical free quadratic algebras are not even related by a St\"ackel transform. In the above mentioned cases, $\widetilde{E}_{17}$, $\widetilde{E}_{8}$, and $\widetilde{E}_{7}$ belong to the same St\"ackel equivalence class which is not the St\"ackel equivalence class of the (St\"ackel equivalent) systems $\widetilde{S}_{2}$, $\widetilde{S}_{4}$, and $\widetilde{D}4A$. Since the classif\/ication of abstract contractions of abstract quadratic algebras is not complete we cannot simply compare B\^ocher contractions and abstract contractions of quadratic algebras. Instead we are led to ask the following.

\begin{Question}
Let $A$ and $B$ be 2D second order nondegenerate superintegrable systems. Suppose that there is a contraction of free abstract quadratic
algebras $Q(A)\longrightarrow Q(B)$. Are there necessarily superintegrable systems $A'$ and $B'$ such that
\begin{enumerate}\itemsep=0pt
\item[1)] $Q(A)=Q(A')$, $Q(B)=Q(B')$,
\item[2)] there is a B\^ocher contraction from $A'$ to $B'$.
\end{enumerate}
\end{Question}

The answer is no. Indeed the following 7 abstract contractions have no geometric counterpart as B\^ocher contractions:
\begin{alignat*}{3}
&1)\ &&Q(S7)\to Q(E16),&\\
&2)\ &&Q(D4C)=Q(D4B)\to Q(E20),&\\
&3)\ &&Q(D2C)=Q(D3B)=Q(D3C)\to Q(E16),&\\
&4)\ &&Q(E16)\to Q(E20),&\\
&5)\ &&Q(E17)=Q(E8)=Q(S2)=Q(S4)=Q(E7)=Q(D4A)\to Q(E20),&\\
&6)\ &&Q(D1A)\to Q(D3A),&\\
&7)\ &&Q(D3A)\to Q(E20).&
\end{alignat*}
These contractions are indicated in Diagram~1. In~\cite[Table~1]{KMS2016} all B\^ocher contractions of these systems are given. In these cases there is no chain of B\^ocher contractions linking any of the origin systems to the target system. However, there are ways that these abstract contractions can have practical signif\/icance. In the paper~\cite{Spost} Post shows that the structure equations for all of the quantum 2D quadratic algebras can be represented by either dif\/ferential or dif\/ference operators depending on one complex variable.

\looseness=1 In some cases a model of one quadratic algebra contracts to a model of another quadratic algebra, even though there is no geometrical counterpart. An example of this can be found in~\cite{KMP2014} where the Askey scheme is described through contraction of a~dif\/ference operator model of~$S9$ to dif\/ferential and dif\/ference operator models of other quadratic algebras, see Fig.~\ref{caption1}. This is the part of the scheme related to contractions of nondegenerate systems, the top half. The bottom half corresponds to restrictions of nondegenerate to degenerate systems, contractions of degenerate systems and contractions to Heisenberg (singular) systems. On the left side are the orthogonal polynomials that realize f\/inite-dimensional representations of the quadratic algebras and on the right those that realize inf\/inite-dimensional bounded below representations. Note that some of the contractions go from a superintegrable system to itself in a nontrivial manner. We did not explicitly mention these in our classif\/ication since they are so numerous, but they are pointed out in references~\cite{KMS2016} and~\cite{KM2014}. All of the contractions of the quadratic algebra representations are induced by geometric contractions of the corresponding superintegrable systems except for the~2 on the left and~2 on the right with the longest arrows, contractions of~$E1$ to~$E3'$. The limits of Hahn and dual Hahn polynomials to Krawtchouk polynomials and continuous Hahn and dual Hahn polynomials to Meixner--Pollaczek polynomials are abstract contractions of~$E1$ to~$E3'$ not induced by geometric contractions. This is an example of how abstract quadratic algebra contractions can be realized and shown to have practical signif\/icance.

\begin{figure}[t]\centering
\includegraphics[scale=0.607]{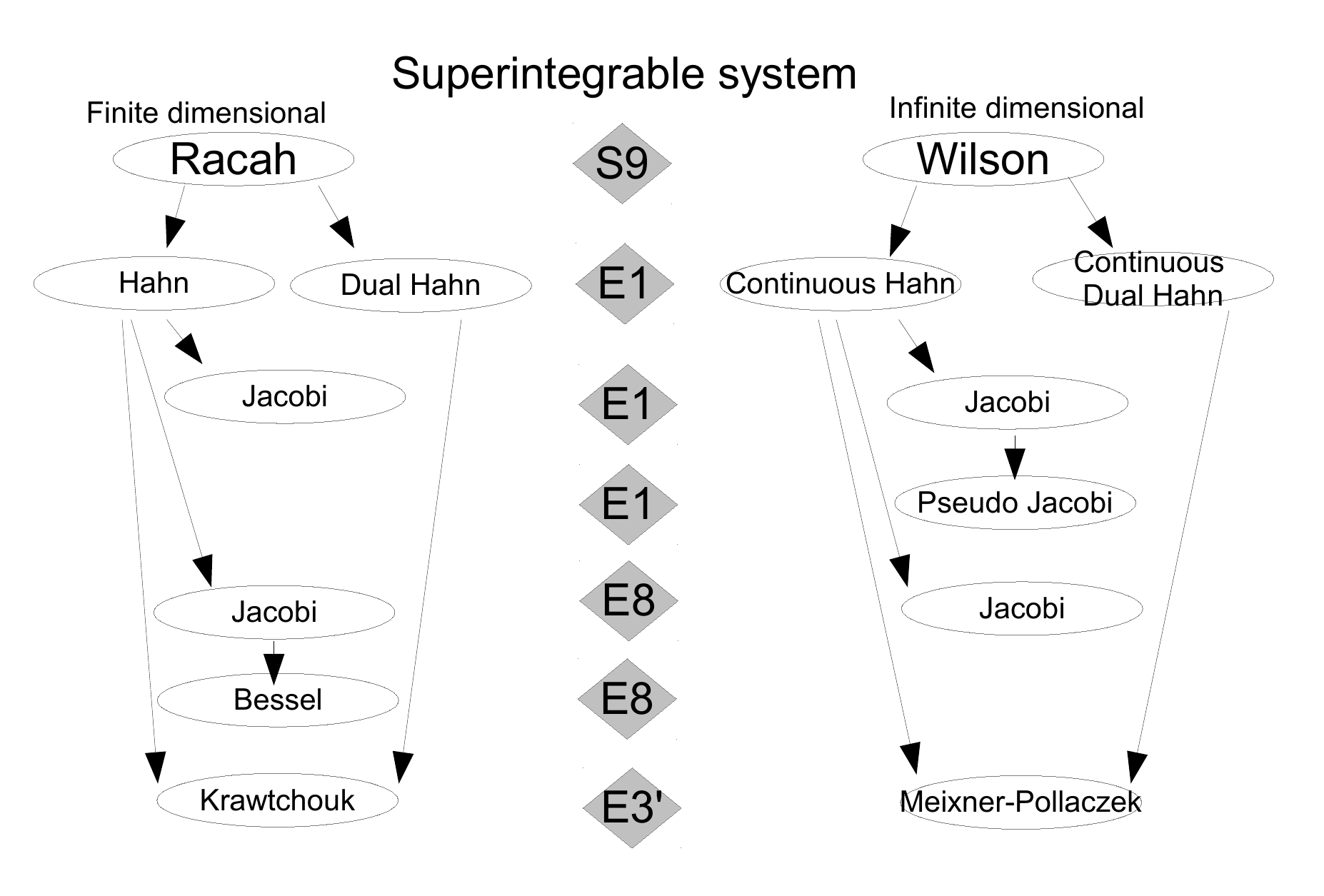}
\caption{Contractions of nondegenerate systems and the top half of the Askey scheme.}\label{caption1}
\end{figure}

\subsection[Contractions between geometric quadratic algebras and abstract quadratic algebras]{Contractions between geometric quadratic algebras\\ and abstract quadratic algebras}

In Section~\ref{Comparison} we identif\/ied the canonical forms of the geometric quadratic algebras inside the space of all canonical forms of abstract quadratic algebras. In this section we give examples for contractions between geometric and abstract quadratic algebras.

\subsubsection{Contraction of an abstract quadratic algebra to a geometric one} There are plenty of such contractions. The canonical forms of the geometric system $\widetilde{E}_{17}$ is given by $\mathcal{L}_1^2 \mathcal{L}_2$. As noted in Section~\ref{Comparison} (and following the labeling of Table~\ref{table1}), the case of $2a$, that is, a canonical form that is given by
\begin{gather*}\mathcal{L}_1^2\mathcal{L}_2+\mathcal{L}_1\mathcal{H}^2+\mathcal{L}_2\mathcal{H}^2+c_{10}\mathcal{H}^3 \end{gather*}
with $c_{10}\in \mathbb{C}$ is not arising from any free 2D, second order nondegenerate superintegrable system. The matrices $A(\epsilon)=\operatorname{diag}(1,1,\epsilon^{-1})$ contract any of the systems above to the geometric system~$\mathcal{L}_1^2 \mathcal{L}_2$. Similarly, the same matrices realize contractions from the non-geometric quadratic algebras with canonical forms $3a$ with $c_{10}=0$: $\mathcal{L}_1^3+\mathcal{L}_1\mathcal{H}^2 $, $3b$: $\mathcal{L}_1^3+\mathcal{H}^3$, and $3e$ with $c_{10}=1$: $\mathcal{L}_1^3+\mathcal{H}\mathcal{L}_2^2 + \mathcal{H}^3$ to $\mathcal{L}_1^3$ that arises from the superintegrable system~$\widetilde{E}_{10}$.

\subsubsection{Contraction of a geometric quadratic algebra to a non-geometric one}
As noted in Section~\ref{Comparison} the canonical form $1c$, $\mathcal{L}_1\mathcal{L}_2(\mathcal{L}_1+\mathcal{L}_2)$ is not arising from any free~2D, second order nondegenerate superintegrable system. The matrices $A(\epsilon)=\operatorname{diag}(1,1,\epsilon^{-1})$ realize contractions from the geometric quadratic algebras~$\widetilde{D}4(b)B$, $\widetilde{D}4(b)C$ (with any value of~$b$), $\widetilde{S}_{7}$~and~$\widetilde{S}_{9}$ to $\mathcal{L}_1\mathcal{L}_2 (\mathcal{L}_1+\mathcal{L}_2)$. There are many other examples.

\section{Conclusions and discussion}
In this paper we have solved the problem of classifying all 2D nondegenerate free abstract quadratic algebras, and have made major steps in determining which of these can be realized as the symmetry algebras of 2D 2nd order superintegrable systems with nondegenerate potential. We have given a precise def\/inition and classif\/ication of B\^ocher contractions, which are the principle mechanisms for relating superintegrable systems via limit relations. We have made major steps toward a classif\/ication of contractions of abstract quadratic algebras and determining which of these can be realized as B\^ocher contractions. In each case we have found some abstract algebras and contractions that cannot be realized geometrically as superintegrable systems or as B\^ocher contractions. We know that some of these cases correspond to contractions of models irreducible representations of quadratic algebras belonging to superintegrable systems where the algebraic representations contract, but the geometrical systems do not. They already occur in the Askey scheme. However, other cases are as yet unclear. In his theory B\^ocher introduces and some of the authors developed a limit procedure for obtaining so-called type~2 separable coordinate systems, see~\cite{KMR1984}, which can be interpreted as limits where the null cone is preserved but the action is nonlinear. This may f\/ill in gaps in our classif\/ication but has not been worked out.

Up to now we have only classif\/ied abstract contractions of quadratic algebras that arise from superintegrable systems on constant curvature and Darboux spaces. We have not yet solved the problem of classifying contractions of abstract quadratic algebras that do not arise in this way, though the B\^ocher contractions are known.

One can see from the tables in \cite{KMS2016} that in general there are often multiple distinct contractions that link two geometric quadratic algebras, even multiple distinct contractions that take a~quadratic algebra to itself. The abstract contractions classif\/ied here should be though of as providing existence proofs that a contraction between to abstract quadratic algebras does or does not exist, not giving information on the multiplicities
of such contractions.

In a paper under preparation we classify all abstract 2D 2nd order superintegrable systems with degenerate potential and, in this case, work out {\it all} possible abstract contractions and compare the results with those for B\^ocher contractions of geometric superintegrable systems.

All of the concepts introduced here are clearly also applicable for dimensions $n\ge 3$~\cite{CKP2015}. Already we have used the special B\^ocher contractions for $n=3$ to derive new families of superintegrable systems in 3 dimensions~\cite{EM2017}. This paper can be considered as part of the preparation for these more complicated cases.

\subsection*{Acknowledgements}
This work was partially supported by a grant from the Simons Foundation (\#~208754 to Willard Miller Jr.\ and by CONACYT grant (\# 250881 to M.A.~Escobar-Ruiz). The author M.A.~Escobar-Ruiz is grateful to ICN UNAM for the kind hospitality during his visit, where a~part of the research was done, he was supported in part by DGAPA grant IN108815 (Mexico).

\pdfbookmark[1]{References}{ref}
\LastPageEnding


\begin{thebibliography}{99}
\footnotesize\itemsep=0pt

\bibitem{Bocher}
B\^ocher M., \"Uber die Riehenentwickelungen der Potentialtheory, B.G.~Teubner,
 Leipzig, 1894.

\bibitem{Borel}
Borel A., Linear algebraic groups, \href{https://doi.org/10.1007/978-1-4612-0941-6}{\textit{Graduate Texts in Mathematics}}, Vol.~126, 2nd~ed., Springer-Verlag, New York, 1991.

\bibitem{CKP2015}
Capel J.J., Kress J.M., Post S., Invariant classif\/ication and limits of
 maximally superintegrable systems in~3{D}, \href{https://doi.org/10.3842/SIGMA.2015.038}{\textit{SIGMA}} \textbf{11} (2015),
 038, 17~pages, \href{http://arxiv.org/abs/1501.06601}{arXiv:1501.06601}.

\bibitem{DASK2007}
Daskaloyannis C., Tanoudis Y., Quantum superintegrable systems with quadratic
 integrals on a two dimensional manifold, \href{https://doi.org/10.1063/1.2746132}{\textit{J.~Math. Phys.}} \textbf{48}
 (2007), 072108, 22~pages, \href{http://arxiv.org/abs/math-ph/0607058}{math-ph/0607058}.

\bibitem{RKM2016}
Escobar-Ruiz M.A., Kalnins E.G., Miller Jr. W., 2D 2nd order Laplace
 superintegrable systems, Heun equations, QES and B\^ocher contractions,
 \href{http://arxiv.org/abs/1609.03917}{arXiv:1609.03917}.

\bibitem{EM2017}
Escobar-Ruiz M.A., Miller Jr. W., Toward a classif\/ication of semidegenerate
 {3D} superintegrable systems, \href{https://doi.org/10.1088/1751-8121/aa5843}{\textit{J.~Phys.~A: Math. Theor.}} \textbf{50}
 (2017), 095203, 22~pages, \href{http://arxiv.org/abs/1611.02977}{arXiv:1611.02977}.

\bibitem{EVAN}
Evans N.W., Super-integrability of the {W}internitz system, \href{https://doi.org/10.1016/0375-9601(90)90611-Q}{\textit{Phys.
 Lett.~A}} \textbf{147} (1990), 483--486.

\bibitem{FORDY}
Fordy A.P., Quantum super-integrable systems as exactly solvable models,
 \href{https://doi.org/10.3842/SIGMA.2007.025}{\textit{SIGMA}} \textbf{3} (2007), 025, 10~pages, \href{http://arxiv.org/abs/math-ph/0702048}{math-ph/0702048}.

\bibitem{Gant}
Gantmacher F.R., The theory of matrices, Vol.~II, Chelsea, New York, 1959.

\bibitem{HKMS2015}
Heinonen R., Kalnins E.G., Miller Jr. W., Subag E., Structure relations and
 {D}arboux contractions for 2{D} 2nd order superintegrable systems,
 \href{https://doi.org/10.3842/SIGMA.2015.043}{\textit{SIGMA}} \textbf{11} (2015), 043, 33~pages, \href{http://arxiv.org/abs/1502.00128}{arXiv:1502.00128}.

\bibitem{Wigner}
In\"on\"u E., Wigner E.P., On the contraction of groups and their
 representations, \href{https://doi.org/10.1073/pnas.39.6.510}{\textit{Proc. Nat. Acad. Sci. USA}} \textbf{39} (1953),
 510--524.

\bibitem{Pog96}
Izmest'ev A.A., Pogosyan G.S., Sissakian A.N., Winternitz P., Contractions of
 {L}ie algebras and separation of variables, \href{https://doi.org/10.1088/0305-4470/29/18/024}{\textit{J.~Phys.~A: Math. Gen.}}
 \textbf{29} (1996), 5949--5962.

\bibitem{Pog01}
Izmest'ev A.A., Pogosyan G.S., Sissakian A.N., Winternitz P., Contractions of
 {L}ie algebras and the separation of variables: interbase expansions,
 \href{https://doi.org/10.1088/0305-4470/34/3/314}{\textit{J.~Phys.~A: Math. Gen.}} \textbf{34} (2001), 521--554.

\bibitem{KKM20041}
Kalnins E.G., Kress J.M., Miller Jr. W., Second-order superintegrable systems
 in conformally f\/lat spaces. {I}.~{T}wo-dimensional classical structure
 theory, \href{https://doi.org/10.1063/1.1897183}{\textit{J.~Math. Phys.}} \textbf{46} (2005), 053509, 28~pages.

\bibitem{KKM20041-II}
Kalnins E.G., Kress J.M., Miller Jr. W., Second order superintegrable systems
 in conformally f\/lat spaces. {II}.~{T}he classical two-dimensional {S}t\"ackel
 transform, \href{https://doi.org/10.1063/1.1894985}{\textit{J.~Math. Phys.}} \textbf{46} (2005), 053510, 15~pages.

\bibitem{KKM20041-III}
Kalnins E.G., Kress J.M., Miller Jr. W., Second order superintegrable systems
 in conformally f\/lat spaces. {III}.~{T}hree-dimensional classical structure
 theory, \href{https://doi.org/10.1063/1.2037567}{\textit{J.~Math. Phys.}} \textbf{46} (2005), 103507, 28~pages.

\bibitem{KKM20041-IV}
Kalnins E.G., Kress J.M., Miller Jr. W., Second order superintegrable systems
 in conformally f\/lat spaces. {IV}.~{T}he classical 3{D} {S}t\"ackel transform
 and 3{D} classif\/ication theory, \href{https://doi.org/10.1063/1.2191789}{\textit{J.~Math. Phys.}} \textbf{47} (2006),
 043514, 26~pages.

\bibitem{KKM20041-V}
Kalnins E.G., Kress J.M., Miller Jr. W., Second-order superintegrable systems
 in conformally f\/lat spaces. {V}.~{T}wo- and three-dimensional quantum
 systems, \href{https://doi.org/10.1063/1.2337849}{\textit{J.~Math. Phys.}} \textbf{47} (2006), 093501, 25~pages.

\bibitem{KKM20041-VI}
Kalnins E.G., Kress J.M., Miller Jr. W., Nondegenerate 2{D} complex {E}uclidean
 superintegrable systems and algebraic varieties, \href{https://doi.org/10.1088/1751-8113/40/13/008}{\textit{J.~Phys.~A: Math.
 Theor.}} \textbf{40} (2007), 3399--3411, \href{http://arxiv.org/abs/0708.3044}{arXiv:0708.3044}.

\bibitem{Laplace2011}
Kalnins E.G., Kress J.M., Miller Jr. W., Post S., Laplace-type equations as
 conformal superintegrable systems, \href{https://doi.org/10.1016/j.aam.2009.11.014}{\textit{Adv. in Appl. Math.}} \textbf{46}
 (2011), 396--416, \href{http://arxiv.org/abs/0908.4316}{arXiv:0908.4316}.

\bibitem{KKMW}
Kalnins E.G., Kress J.M., Miller Jr. W., Winternitz P., Superintegrable systems
 in {D}arboux spaces, \href{https://doi.org/10.1063/1.1619580}{\textit{J.~Math. Phys.}} \textbf{44} (2003), 5811--5848,
 \href{http://arxiv.org/abs/math-ph/0307039}{math-ph/0307039}.

\bibitem{KKMP}
Kalnins E.G., Kress J.M., Pogosyan G.S., Miller Jr. W., Completeness of
 superintegrability in two-dimensional constant-curvature spaces,
 \href{https://doi.org/10.1088/0305-4470/34/22/311}{\textit{J.~Phys.~A: Math. Gen.}} \textbf{34} (2001), 4705--4720,
 \href{http://arxiv.org/abs/math-ph/0102006}{math-ph/0102006}.

\bibitem{KM2014}
Kalnins E.G., Miller Jr. W., Quadratic algebra contractions and second-order
 superintegrable systems, \href{https://doi.org/10.1142/S0219530514500377}{\textit{Anal. Appl. (Singap.)}} \textbf{12} (2014),
 583--612, \href{http://arxiv.org/abs/1401.0830}{arXiv:1401.0830}.

\bibitem{KMP2010}
Kalnins E.G., Miller Jr. W., Post S., Coupling constant metamorphosis and
 {$N$}th-order symmetries in classical and quantum mechanics,
 \href{https://doi.org/10.1088/1751-8113/43/3/035202}{\textit{J.~Phys.~A: Math. Theor.}} \textbf{43} (2010), 035202, 20~pages,
 \href{http://arxiv.org/abs/0908.4393}{arXiv:0908.4393}.

\bibitem{KMP2014}
Kalnins E.G., Miller Jr. W., Post S., Contractions of 2{D} 2nd order quantum
 superintegrable systems and the {A}skey scheme for hypergeometric orthogonal
 polynomials, \href{https://doi.org/10.3842/SIGMA.2013.057}{\textit{SIGMA}} \textbf{9} (2013), 057, 28~pages,
 \href{http://arxiv.org/abs/1212.4766}{arXiv:1212.4766}.

\bibitem{KMR1984}
Kalnins E.G., Miller Jr. W., Reid G.J., Separation of variables for complex
 {R}iemannian spaces of constant curvature. {I}.~{O}rthogonal separable
 coordinates for {${\rm S}_{n{\bf C}}$} and {${\rm E}_{n{\bf C}}$},
 \href{https://doi.org/10.1098/rspa.1984.0075}{\textit{Proc. Roy. Soc. London Ser.~A}} \textbf{394} (1984), 183--206.

\bibitem{KMS2016}
Kalnins E.G., Miller Jr. W., Subag E., B\^ocher contractions of conformally
 superintegrable {L}aplace equations, \href{https://doi.org/10.3842/SIGMA.2016.038}{\textit{SIGMA}} \textbf{12} (2016), 038,
 31~pages, \href{http://arxiv.org/abs/1512.09315}{arXiv:1512.09315}.

\bibitem{Koenigs}
Koenigs G.X.P., Sur les g{\'e}od{\'e}siques a integrales quadratiques, in
 Le{\,c}ons sur la th{\'e}orie g{\'e}n{\'e}rale des surfaces, Vol.~4, Editor
 J.G.~Darboux, Chelsea Publishing, 1972, 368--404.

\bibitem{Kress2007}
Kress J.M., Equivalence of superintegrable systems in two dimensions,
 \href{https://doi.org/10.1134/S1063778807030167}{\textit{Phys. Atomic Nuclei}} \textbf{70} (2007), 560--566.

\bibitem{MPW2013}
Miller Jr. W., Post S., Winternitz P., Classical and quantum superintegrability
 with applications, \href{https://doi.org/10.1088/1751-8113/46/42/423001}{\textit{J.~Phys.~A: Math. Theor.}} \textbf{46} (2013),
 423001, 97~pages, \href{http://arxiv.org/abs/1309.2694}{arXiv:1309.2694}.

\bibitem{NP}
Nesterenko M., Popovych R., Contractions of low-dimensional {L}ie algebras,
 \href{https://doi.org/10.1063/1.2400834}{\textit{J.~Math. Phys.}} \textbf{47} (2006), 123515, 45~pages,
 \href{http://arxiv.org/abs/math-ph/0608018}{math-ph/0608018}.

\bibitem{Spost}
Post S., Models of quadratic algebras generated by superintegrable systems in
 2{D}, \href{https://doi.org/10.3842/SIGMA.2011.036}{\textit{SIGMA}} \textbf{7} (2011), 036, 20~pages, \href{http://arxiv.org/abs/1104.0734}{arXiv:1104.0734}.

\bibitem{TTW2001}
Tempesta P., Turbiner A.V., Winternitz P., Exact solvability of superintegrable
 systems, \href{https://doi.org/10.1063/1.1386927}{\textit{J.~Math. Phys.}} \textbf{42} (2001), 4248--4257,
 \href{http://arxiv.org/abs/hep-th/0011209}{hep-th/0011209}.

\bibitem{SCQS}
Tempesta P., Winternitz P., Harnad J., Miller W., Pogosyan G., Rodriguez M.
 (Editors), Superintegrability in classical and quantum systems, \textit{CRM
 Proceedings and Lecture Notes}, Vol.~37, Amer. Math. Soc.,
 Providence, RI, 2004.

\bibitem{TurbinerHeun}
Turbiner A.V., The {H}eun operator as a {H}amiltonian, \href{https://doi.org/10.1088/1751-8113/49/26/26LT01}{\textit{J.~Phys.~A:
 Math. Theor.}} \textbf{49} (2016), 26LT01, 8~pages, \href{http://arxiv.org/abs/1603.02053}{arXiv:1603.02053}.

\bibitem{Turbiner2016}
Turbiner A.V., One-dimensional quasi-exactly solvable {S}chr\"odinger
 equations, \href{https://doi.org/10.1016/j.physrep.2016.06.002}{\textit{Phys. Rep.}} \textbf{642} (2016), 1--71,
 \href{http://arxiv.org/abs/1603.02992}{arXiv:1603.02992}.

\bibitem{WW}
Weimar-Woods E., The three-dimensional real {L}ie algebras and their
 contractions, \href{https://doi.org/10.1063/1.529222}{\textit{J.~Math. Phys.}} \textbf{32} (1991), 2028--2033.

\end{thebibliography}
\end{document}